\documentclass[final,3p,times]{elsarticle}

\usepackage{latexsym}
\usepackage{amssymb}
\usepackage{amsthm}
\usepackage[cmex10]{amsmath}
\usepackage{graphicx}
\usepackage{float}
\usepackage{subfig}
\usepackage[]{placeins}
\usepackage[]{algorithm}
\usepackage{algpseudocode}
\usepackage{multirow}
\usepackage[subnum]{cases}
\usepackage{tabularx}

\usepackage[]{amsrefs} % without, The elsarticle class will remove the bibliography heading by declaring \let\bibsection\relax
\newcommand{\E}{\mathbb{E}}

\renewcommand{\algorithmicforall}{\textbf{for each}}

\makeatletter
\renewcommand\appendix{\par
  \setcounter{section}{0}%
  \setcounter{subsection}{0}%
  \setcounter{equation}{0}%
  \setcounter{table}{0}%------------ << add
  \setcounter{figure}{0}%----------- << add
  \gdef\theequation{\@Alph\c@section.\arabic{equation}}%
  \gdef\thefigure{\@Alph\c@section.\arabic{figure}}%
  \gdef\thetable{\@Alph\c@section.\arabic{table}}%
  \gdef\thesection{\Alph{section}}%
  \@addtoreset{equation}{section}%
  \@addtoreset{table}{section}%----- << add
  \@addtoreset{figure}{section}%---- << add
}
\makeatother

\begin{document}
\begin{frontmatter}
\title{Mining frequent items in the time fading model}
\author [unile] {Massimo~Cafaro\corref{cor1}}
\ead{massimo.cafaro@unisalento.it}
\cortext[cor1]{Corresponding author}
\author [unile] {Marco Pulimeno}
\ead{marco.pulimeno@unisalento.it}
\author [unile] {Italo Epicoco}
\ead{italo.epicoco@unisalento.it}
\author [unile] {Giovanni Aloisio}
\ead{giovanni.aloisio@unisalento.it}
\address[unile]{University of Salento, Lecce, Italy}

\begin{abstract} We present FDCMSS, a new sketch--based algorithm for mining frequent items in data streams. The algorithm cleverly combines key ideas borrowed from forward decay, the Count-Min and the Space Saving algorithms. It  works in the time fading model, mining data streams according to the cash register model. We formally prove its correctness and show, through extensive experimental results, that our algorithm outperforms $\lambda$-HCount, a recently developed algorithm, with regard to speed, space used, precision attained and error committed on both synthetic and real datasets.  
\end{abstract}


\newtheorem{thm}{Theorem}
\newtheorem{lem}[thm]{Lemma}
\newdefinition{rmk}{Remark}
\newproof{pf}{Proof}
% \newproof{pot}{Proof of Theorem \ref{thm2}}
\newtheorem{prop}[thm]{Proposition}
\newtheorem*{cor}{Corollary}
\newdefinition{defn}{Definition}
\newtheorem{conj}{Conjecture}
\newtheorem{exmp}{Example}
\newtheorem{case}{Case}
%\newtheorem{claim}[thm]{Claim}
%\newtheorem{fact}[thm]{Fact}
%\newtheorem{assumption}[thm]{Assumption}

%\tableofcontents
\end{frontmatter}

%\tableofcontents

\section{Introduction}
\label{intro}

A data stream $\sigma$ consists of a sequence of $n$ items drawn from a universe $\mathcal{U}$. Without loss of generality, let $m$ be the number of distinct items in $\sigma$ i.e., let $\mathcal{U}=\{1,2,\ldots,m\}$. Depending on the applications, items may be numbers, IP addresses, points, graph edges etc. Owing to the huge size of $\sigma$, an algorithm in charge of processing its items is subject to the stringent requirement that no more than one pass over the data is allowed. In practice, storing the items is not a feasible option. We refer the interested reader to \cite{TCS-002}, a very good survey of streaming algorithms, for additional details and underlying reasons motivating research in this area. In this paper, we deal with mining of frequent items in a data stream. This problem has been extensively studied, and is recognized as one of the most important in the streaming algorithms literature, where it is also called, depending on the specific context, \textit{hot list analysis} \cite{Gibbons}, market basket analysis \cite{Brin} and \textit{iceberg query} \cite{Fang98computingiceberg}, \cite{Beyer99bottom-upcomputation}. 

Among the many possible applications, we recall here network traffic analysis \cite{DemaineLM02},  \cite{Estan}, \cite{Pan}, analysis of web logs \cite{Charikar}, Computational and theoretical Linguistics \cite{CICLing}.

Letting $f_i$ denote the frequency of the item $i \in \mathcal{U}$ (i.e., its number of occurrences in $\sigma$), $\textbf{f} = (f_1,\ldots,f_m)$ the frequency vector, $0 < \phi < 1$ a support threshold and $||\textbf{f}||_1$ the 1-norm of $\textbf{f}$ (which represents the total number of occurrences of all of the stream items), an approximate solution requires returning all of the items which are frequent, i.e., those items $i$ such that $f_i > \phi ||\textbf{f}||_1$ and, letting $0 < \epsilon < 1$ denote the error committed, an algorithm must not return any item $i$ such that $f_i \leq (\phi - \epsilon) ||\textbf{f}||_1$. In particular, the error is such that $\epsilon < \phi$.

Beyond the traditional distinction between deterministic and randomized algorithms, in this area algorithms for detecting frequent items are also often referred to as being either \emph{counter} or \emph{sketch} based. In counter--based algorithms, a fixed number of counters is used to keep track of stream items. Indeed, given a support threshold $0 < \phi < 1$, the number of possible frequent items is an integer belonging to the open interval $(0, 1/\phi)$. Sketch--based algorithms monitor the data stream by using a set of counters, stored in a sketch data structure, usually a bi-dimensional array. Stream items are mapped by hash functions to their corresponding cells in the sketch. Whilst counter--based algorithms are deterministic, sketch--based ones are randomized and provide a probabilistic guarantee. 

The streaming model we have described so far is called \emph{cash register} or \emph{strict turnstile} model \cite{TCS-002}, since only \emph{insertions} are allowed. On the other hand, in the more general \emph{turnstile} model, \emph{deletions} are also allowed. An advantage of sketch--based algorithms is that they can easily support deletions and can therefore work in the turnstile model; counter--based algorithms only work in the cash register model.

Among the counter based algorithms, the first sequential algorithm has been proposed by Misra and Gries \cite{Misra82}. Later, this algorithm was rediscovered independently by Demaine et al. \cite{DemaineLM02} (the so-called \emph{Frequent} algorithm) and Karp et al. \cite{Karp}. Recently developed counters--based algorithms include \emph{Sticky Sampling}  and \emph{Lossy Counting} \cite{Manku02approximatefrequency}, and \emph{Space Saving} \cite{Metwally2006}. Notable sketch--based algorithms are \emph{CountSketch} \cite{Charikar}, \emph{Group Test} \cite{Cormode-grouptest}, \emph{Count-Min} \cite{Cormode05} and \emph{hCount} \cite{Jin03}.

Regarding parallel algorithms, \cite{cafaro-tempesta} and \cite{Cafaro-Pulimeno-Tempesta} present message-passing based parallel versions of the Frequent and Space Saving algorithms. Among the algorithms for shared-memory architectures we recall here a parallel version of Frequent \cite{Zhang2013}, a parallel version of Lossy Counting \cite{Zhang2012}, and parallel versions of Space Saving \cite{Roy2012} \cite{Das2009}. Novel shared-memory parallel algorithms for frequent items were recently proposed in \cite{Tangwongsan2014}. Accelerator based algorithms for frequent items exploiting a GPU (Graphics Processing Unit) include \cite{Govindaraju2005} and \cite{Erra2012}. 

In this paper, we are concerned with the problem of detecting frequent items in a stream with the additional constraint that recent items must be weighted more than former items. The underlying assumption is that, in some applications, recent data is certainly more useful and valuable than older, stale data. Therefore, each item in the stream has an associated timestamp that will be used to determine its weight. In practice, instead of estimating frequency counts, we are required to estimate \emph{decayed counts}. Two different models have been proposed in the literature: the \emph{sliding window} and the \emph{time fading} model.

In the sliding window model \cite{Datar} \cite{TCS-002}, freshness of recent items is captured by a time window, i.e., a temporal interval of fixed size in which only the most recent $N$ items are taken into account; detection of frequent items is strictly related to those items falling in the window. The items in the stream become stale over time, since the window periodically slides forward.

The time fading model \cite{recent-freq-items} \cite{exp-decay} \cite{Chen-Mei} does not use a window sliding over time; freshness of more recent items is instead emphasized by \emph{fading} the frequency count of older items. This is achieved by using a decaying factor $0 < \lambda < 1$ to compute an item's \textit{decayed count} (also called \textit{decayed frequency}) through decay functions that assign greater weight to more recent elements. The older an item, the lower its decayed count is: in the case of exponential decay, the weight of an item occurred $n$ time units in the past is $e^{-\lambda n}$, which is an exponentially decreasing quantity. 

This paper is organized as follows. We recall in Section \ref{ideas} key definitions and concepts that will be used in the rest of the manuscript, and introduce in Section \ref{chenmei} the $\lambda$-HCount algorithm \cite{Chen-Mei}, a recently published sketch--based algorithm that detects frequent items in the time fading model. Then, we introduce in Section \ref{alg} our FDCMSS algorithm and formally prove in Sections \ref{bound} and \ref{correctness}, respectively, its error bound and correctness. Next, we compare $\lambda$-HCount to FDCMSS from a theoretical perspective, and show that our algorithm achieves its error bound using a tiny fraction of the space required by $\lambda$-HCount. Then, we provide extensive experimental results in Section \ref{results}, in which we compare again FDCMSS versus $\lambda$-HCount from a quantitative, practical perspective, and show that FDCMSS outperforms $\lambda$-HCount with regard to speed, space used, precision attained and error committed on both synthetic and real datasets. Finally, we draw our conclusions in Section \ref{conclusions}.

\section{Key ideas}
\label{ideas}

Our algorithm cleverly combines ideas borrowed from forward decay \cite{forward-decay}, the Count-Min sketch--based algorithm \cite{Cormode05} , and the Space Saving counter--based algorithm \cite{Metwally2006}. In this section, we recall preliminary definitions and key concepts. 

\begin{defn}
\label{decay-function}
Given an item $i$ with arrival time $t_i$, a \textit{decay function} returns a weight for the item. In our algorithm the weight $w(i,t)$ determined at time $t$, depends on the timestamp $t_i$ associated to the item. Decay functions satisfy the following properties: (i) $w(i,t) = 1$ when $t_i = t$ and $0 \leq w(i,t) \leq 1$ for all $t \geq t_i$; (ii) $w$ is monotone non-increasing as time increases, i.e., $t' \geq t \implies w(i, t') \leq w(i, t)$.
\end{defn}

Regarding the time fading model, related work has mostly exploited \textit{backward decay} functions, in which the weight of an item is a function of its age, $a$, where the age at time $t > t_i$ is simply $a = t-t_i$. The term backward decay stems from the aim of measuring from the current time back to the item's timestamp.

\begin{defn}
A backward decay function is defined by a positive monotone non-increasing function $f$ so that the weight of the $i$th item with arrival time $t_i$ determined at time $t$ is given by $w(i, t) = \frac{f(t-t_i)}{f(t-t)}=\frac{f(t-t_i)}{f(0)}$. The denominator in the expression normalizes the weight, so that it obeys condition $(i)$ of Definition \ref{decay-function}.
\end{defn}

Prior algorithms and applications have been using backward exponential decay functions such as $f(a) = e^{-\lambda a}$ with $\lambda > 0$. In our algorithm, we use instead a forward decay function, defined as follows. Under forward decay, the weight of an item is computed on the amount of time between the arrival of an item and a fixed point $L$, called the \textit{landmark} time, which, by convention, is some time earlier than the timestamps of all of the items. The idea is to look forward in time from the landmark to see an item, instead of looking backward from the current time.

\begin{defn}
Given a positive monotone non-decreasing function $g$, and a landmark time $L$, the forward decayed weight of an item $i$ with arrival time $t_i > L$ measured at time $t \geq t_i$ is given by $w(i, t) = \frac{g(t_i-L)}{g(t-L)}$.
\end{defn}

When $t = t_i$ the weight is 1 (condition $(i)$ of Definition \ref{decay-function}). Since $g$ is monotone non-decreasing, as $t$ increases the weight does not increase, and $0 \leq w(i, t) \leq 1$. 

It's easy to prove that backward and forward exponential decay coincide. In our algorithm, we could use the exponential decay; however, we prefer to use the polynomial forward decay function $g(n) = n^2$. It has been proved that this class of functions (i.e., $g(n) = n^\beta$) satisfies a \textit{relative decay} property, which states that for any time $t$ after a landmark time $L$, the weight for items with timestamp $\gamma t + (1-\gamma) L$ is the same. In practice, relative decay holds for forward decay functions assigning the weight of an item depending only on where the item falls as a fraction in the window defined by $L$ and $t$. As an example, a function for which relative decay holds assigns to an item arriving half way between $L$ and $t$ the same weight as $t$ increases. 

Intuitively, this property requires assigning to an item a weight which is a function of its relative age, i.e., its age as a fraction of the total time period observed. Since backward decay is only concerned with absolute age, it can not provide relative decay. An important consequence of relative decay is that it allows selecting a meaningful landmark time $L$ to choose for forward decay.

Besides its flexibility (e.g., choosing an appropriate polynomial function we can select and control a slower rate of decay with regard to an exponential), another advantage of forward decay is related to its ability to deal with out of order arrival of stream items. Indeed, forward decay does not rely on items arriving in increasing order of timestamps. On the contrary, under backward decay, handling out of order arrivals can require significant effort to accommodate.

\begin{defn}
\label{decayed-count}
The \textit{decayed count}, $C$, of a stream $\sigma$ of $n$ items is the sum of decayed weights of items: $C=\sum_{i=1}^n \frac{g(t_i-L)}{g(t-L)}$.
\end{defn}

We can now formally state the problem solved by our algorithm: \textit{approximate frequent items under forward decay}.

\begin{defn}
(Frequent items under forward decay) For each item in the input, $v$, its decayed count is given by $f_v = \sum_{v_i = v} \frac{g(t_i-L)}{g(t-L)}$. Given a threshold value $\phi$, the frequent items are all of the items $v$ satisfying $f_v > \phi C$.
\end{defn}

\begin{defn}
(Approximate frequent items under forward decay problem) Given an error bound $\epsilon$ and a threshold $\phi$, determine all of the items satisfying $f_v > \phi C$, and report no items with $f_v \leq (\phi-\epsilon) C$.
\end{defn}

Our goal is to design an algorithm solving the \textit{Approximate frequent items under forward decay} problem by providing the following ($\epsilon$, $\delta$) \textit{approximation}. 

\begin{defn} 
(($\epsilon$, $\delta$) approximation) Let $A(\sigma)$ denote the output of a randomized streaming algorithm $A$ on input $\sigma$; it is worth noting here that $A(\sigma)$ is a random variable. Moreover, let $f(\sigma)$ be the function that $A$ is supposed to compute. The algorithm $A$  ($\epsilon$, $\delta$) approximates $f$ if $\Pr[|A(\sigma) - f(\sigma)| > \epsilon] \leq \delta$.
\end{defn}

Count-Min is based on a sketch whose dimensions are derived by the input parameters $\epsilon$, the error, and $\delta$, the probability of failure. In particular, for Count-Min $d=\lceil \ln 1/\delta \rceil$ is the number of rows in the sketch and $w=\lceil e/\epsilon \rceil$ is the number of columns. Every cell in the sketch is a counter, which is updated by hash functions. By using this data structure, the algorithm solves with high probability (i.e., with probability greater than or equal to 1 - $\delta$) the \textit{frequency estimation} problem for arbitrary items. The algorithm may also be extended to solve the \textit{approximate frequent items} problem as well, by using an additional heap data structure which is updated each time a cell is updated. Since in Count-Min the frequencies stored in the cells overestimate the true frequencies, a point query for an arbitrary item simply inspects all of the $d$ cells in which the item is mapped to by the corresponding hash functions and returns the minimum of those $d$ counters. 

In addition, Count-Min allows us reusing the same underlying data structure to solve (if needed), beside frequent items, additional problems related to the same input stream (e.g., quantiles, frequency estimation, medians, etc). Moreover, the Count-Min data structure requires less space with regard to other sketch-based algorithms.

We shall prove later (see Theorem \ref{thm-correctness}) that in our algorithm, with high probability, if an item $i$ is frequent, then it appears as a majority item candidate in at least one of the $d$ sketch cells in which it falls. Therefore, in order to detect frequent items, we decided to use the Space Saving algorithm \cite{Metwally2006}. This is a counter--based algorithm, designed to solve the \textit{approximate frequent items} problem using $k \geq \lceil 1/\phi \rceil$ counters in order to determine the frequent items in the input data stream.

We could use a different counter--based algorithm (e.g., \textit{Frequent} \cite{DemaineLM02}); our choice stems from the well-known fact that Space Saving provides the greatest accuracy (precision, total and average relative error) among the counter-based algorithms \cite{Cormode} \cite{Manerikar}. Moreover, exploiting in FDCMSS the Space Saving algorithm within the sketch cells allows us avoiding the need for an additional, separate data structure to keep track of frequent items. We now briefly recall how Space Saving works.  

Let $\mathcal{S}$ denote the Space Saving stream summary data structure. Updating $\mathcal{S}$ upon arrival of an item works as shown in the pseudocode of Algorithm \ref{ss}. We denote by $c_j.i$ and $c_j.f$ respectively the item monitored by the $j$th counter of $\mathcal{S}$ and its corresponding estimated frequency. When processing an item which is already monitored by a counter, its estimated frequency is incremented by its weight $w$. When processing an item which is not already monitored by one of the available counters, there are two possibilities. If a counter is available, it will be in charge of monitoring the item and its estimated frequency is set to its  weight $w$.  Otherwise, if all of the counters are already occupied (their frequencies are different from zero), the counter storing the item with minimum frequency is incremented by its weight $w$. Then the monitored item is evicted from the counter and replaced by the new item. This happens since an item which is not monitored can not have a frequency greater than the minimal frequency. The complexity of the Space Saving update procedure is $O(1)$.

Let $\sigma$ be the input stream and $\mathcal{S}$ the stream summary data structure at the end of the sequential Space Saving algorithm's execution. Moreover, let $\sum\limits_{{c_j} \in S} {{c_j}}.f$ be the sum of the counters in $\mathcal{S}$, $f_v$ the exact frequency of an item $v$, $\hat{f}_v$ its estimated frequency, $\textbf{f} = (f_1,\ldots,f_m)$ the frequency vector, $\hat{f}^{min}$ the minimum frequency in $\mathcal{S}$ and $\hat{\varepsilon}_v$ the estimated error of item $v$, i.e. an over-estimation of the difference between the estimated and exact frequency. 

Finally, denote by $\mathcal{S_{\phi}}$ the set of counters in $\mathcal{S}$ which are monitoring items ($\left\vert \mathcal{S_{\phi}} \right\vert \leq k$). It is worth noting here that $\hat{f}^{min} = 0$ when  $\left\vert{\mathcal{S_{\phi}}}\right\vert < k$. The following relations hold (as shown in \cite{Metwally2006}):

\begin{equation}
\label{ss1}
\sum\limits_{{c_j} \in S} {{c_j}}.f = ||\textbf{f}||_1,
\end{equation}

\begin{equation}
\label{ss2}
\hat{f}_v - \hat{f}^{min} \leq\hat{f}_v - \hat{\varepsilon}_v \leq f_v \leq \hat{f}_v,  \qquad v \in \mathcal{S_{\phi}},
\end{equation}

\begin{equation}
\label{ss3}
f_v  \leq \hat{f}^{min}, \qquad \hspace{27mm} v \notin \mathcal{S_{\phi}},
\end{equation}

\begin{equation}
\label{ss4}
\hat{f}^{min}  \leq \left\lfloor\frac{||\textbf{f}||_1}{k}\right\rfloor.
\end{equation}

Therefore, it holds that

\begin{equation}
\label{ss5}
\hat{f}_v - f_v \leq \hat{f}^{min} \leq \left\lfloor\frac{||\textbf{f}||_1}{k}\right\rfloor, \qquad \hspace{7mm} v \in \mathcal{U}.
\end{equation}

\begin{algorithm}
\begin{algorithmic}[1]
\Require $\mathcal{S}$, a stream summary; $j$, an item; $w$, the weight of item $j$
\Ensure a stream summary $\mathcal{S}$ containing frequent items
\Procedure {SpaceSavingUpdate}{$\mathcal{S}, j, w$}
\If{$j$ is monitored}
		\State let $c_l$ be the counter monitoring $j$ 
		\State $c_l.f \leftarrow c_l.f + w$
	\Else
		\If{there is a counter $c_r$ which is not monitoring any item}
			\State $c_r.i \leftarrow j$ 
			\State $c_r.f \leftarrow w$
		\Else
			\State let $c_s$ be the counter monitoring the item with least hits
			\State $c_s.i \leftarrow j$ 
			\State $c_s.f \leftarrow c_s.f + w$
		\EndIf
\EndIf
\EndProcedure
\caption{Space Saving update}
\label{ss}
\end{algorithmic}
\end{algorithm}

To recap, we end this Section summarizing the reasons for combining forward decay, the Count-Min and the Space Saving algorithms:

\begin{itemize}
  \item forward exponential decay coincides with backward exponential decay, so that we can still use exponentials;
 \item forward decay allows for greater flexibility (e.g., choosing a polynomial function we can select a different, slower rate of decay with regard to an exponential);
 \item forward decay allows easily dealing with out of order arrival of stream items;
 \item forward decay satisfies a relative decay property, which states that for any time $t$ after a landmark time $L$, the weight for items with timestamp $\gamma t + (1-\gamma) L$ is the same;
 \item the Count-Min algorithm allows us reusing the same underlying data structure to solve (if needed), beside frequent items, additional problems related to the same input stream (e.g., quantiles, frequency estimation, medians, etc);
 \item the Count-Min data structure requires less space with regard to other sketch-based algorithms;
 \item we use Space Saving to detect frequent items within the sketch cells;
 \item the Space Saving algorithm provides the greatest accuracy (precision, total and average relative error) among the counter-based algorithms;
 \item the Space Saving algorithm allows us avoiding the need for an additional, separate data structure to keep track of frequent items.\\

\end{itemize}

\section{The $\lambda$-HCount algorithm}
\label{chenmei}

We now introduce $\lambda$-HCount \cite{Chen-Mei}, a recently published sketch--based algorithm that detects frequent items in the time fading model by using a backward decay exponential function. The $\lambda$-HCount algorithm, shown in pseudocode as Algorithm \ref{l_count}, requires a two dimensional sketch $D$ of size $r \times m$ to store decayed weights and timestamps, and a doubly linked list $F$ to store frequent items candidates, accessed through a hash function. $\lambda$-HCount is based on the use of $r$ \textit{FNV} hash functions $h_i(x), i=1,\ldots,r$ which uniformly and independently map an item to an integer in the interval $[1,m]$; the algorithm requires a support threshold $s$ and an error bound $\epsilon$. The occurrence of an item at time $t_a$ is weighted in time by a factor $\lambda^{t-t_a}$, where $\lambda$ represents the fading factor ($0 <\lambda < 1$). The decayed count of an item is given by the sum of its decayed weight in time. The decayed count of the stream $\sigma$, as proved by the authors, is bounded by $\frac{1}{1 - \lambda}$.

Each entry $D[i, h_i(x)]$ in the sketch $D$ stores $D[i, h_i (x)].s$ which is the decayed count of item $x$ (also called \textit{density}), and $D[i, h_i (x)].t$ which is the last time the value of $D[i, h_i (x)].s$ was updated. Whenever an item $x$ arrives, its decayed count is updated in all of the $r$ cells $D[i, h_i(x)], i=1,\ldots,r$. Then, the algorithm computes ${\hat{f}_x=min_{1\leq i\leq r} \{D[i, h_i(x)].s\}}$ as the estimated decayed count of $x$; if $\hat{f}_x$ is greater than or equal to the threshold $\frac{s-\epsilon}{1-\lambda}$ than the tuple $\{x, t_x, \hat{f}_x\}$ is created or updated in the linked list $F$. The authors proved that the list $F$ requires at most $\frac{r}{s-\epsilon}$ entries to store all of the items with decayed count greater than $\frac{s-\epsilon}{1-\lambda}$. Basically, $\lambda$-HCount can be considered a variant of Count-Min, designed to support frequent items detection in the time fading model.

The authors of $\lambda$-HCount proved that with a sketch requiring $\frac{e (1-\lambda)\ln{(-\frac{M}{\ln{p}})}}{\epsilon^2}$ space, where $M$ is the number of distinct items and $p$ is the success probability, and an additional data structure requiring at most $\frac{r}{s - \epsilon}$ space, where $r$ is the number of hash functions used, their algorithm is able to estimate the frequent items decayed count with an error less than $\frac{\epsilon}{1-\lambda}$ with probability greater than $p$. The algorithm's analysis also shows that all of the items whose exact decayed count exceeds $\frac{s}{1-\lambda}$ will be output (there are no false negatives) and no items whose decayed count is less than $\frac{s-\epsilon}{1-\lambda}$ will be output.

The worst case complexity of $\lambda$-HCount depends on the time required for updating the sketch $D$ and the linked list $F$. When an item is received from the stream, the algorithm computes $r$ hash functions and updates $r$ entries in $D$; the linked list $F$ is also updated accordingly. Since the linked list $F$ is accessed through a hash function, and its update is done in constant time, overall the worst case complexity of per item update is $O(1)$. Therefore, the whole algorithm has worst case complexity $O(r)$, i.e., $O\big(\ln(- \frac{M}{{\ln p}})\big)$. The space complexity is given by the memory required by the sketch $D$ and  the linked list $F$. Overall, the worst case space complexity is $O(r \cdot m) = O\big(\frac{\ln ( - \frac{M}{{\ln p}})}{\epsilon^2} \big)$.

\begin{algorithm}
\begin{algorithmic}[1]
\Require $\lambda$: fading factor; $\epsilon$: error bound; $s$: support threshold; $x$: received item; $t$: arrival time;
\Ensure update of sketch $D$ and linked list $F$ related to item $x$
\Procedure {$\lambda$-HCount Update}{$\lambda, \epsilon, s, x, t$}
\State $\hat{f}_x \leftarrow \infty$
\For {$k=1$ to $r$}
	\State $y \leftarrow h_k(x)$
	\State $D[k, y].s \leftarrow D[k, y].s \cdot \lambda^{t-D[k, y].t} + 1$
	\State $D[k, y].t \leftarrow t$
	\If{$D[k, y].s < \hat{f}_x$}
		\State $\hat{f}_x \leftarrow D[k, y].s$
	\EndIf
\EndFor
\If {$\hat{f}_x > \frac{s-\epsilon}{1-\lambda}$}
	\If {$x \in F$}
		\State change its entry to $\{x, \hat{f}_x, t\}$ and move it to the tail of the queue $F$
	\Else
		\If {the queue $F$ is full}
			\State delete the item at the head of the queue $F$
		\EndIf	
		\State insert $\{x, \hat{f}_x, t\}$ at the tail of the queue $F$
	\EndIf
\EndIf
\EndProcedure
\caption{$\lambda$-HCount algorithm: the update phase}
\label{l_count}
\end{algorithmic}
\end{algorithm}

\section{The FDCMSS algorithm}
\label{alg}

In this section, we introduce our algorithm, distinguishing three different phases: initialization, stream processing, and querying. 

The algorithm's initialization, shown in pseudo-code as Algorithm \ref{init}, requires as input parameters $\epsilon$, the error; $\delta$, the probability of failure; $\phi$, the support threshold; and $t_{init}$, a timestamp. Initialization returns a sketch $D$. The procedure starts deriving $d=\lceil \ln 1/\delta \rceil$, the number of rows in the sketch and $w=\lceil \frac{e}{2\epsilon} \rceil$, the number of columns in the sketch. We shall explain the reason why we set $w$ to this value in Section \ref{bound}.

Then, for each of the $d \times w$ cells available in the sketch $D$ we allocate a data structure $\mathcal{S}$ with two Space Saving counters $c_1$ and $c_2$. Given a counter $c_j, j=1,2$, we denote by $c_j.i$ and $c_j.f$ respectively the counter's item and estimated decayed count. Finally, we set the support threshold to $\phi$, select $d$ pairwise independent hash functions $h_1,\ldots,h_d:[m] \rightarrow [w]$ mapping $m$ distinct items into $w$ cells, initialize the \textit{count} variable, representing the total decayed count of all of the items in the stream (see Definition \ref{decayed-count}) to zero and the \textit{L} variable (our landmark time) to $t_{init}$, which is a timestamp less than or equal to all of the items' timestamps. The worst case complexity of the initialization procedure is $O(\frac{1}{\epsilon} \ln \frac{1}{\delta})$.

Updating the sketch upon arrival of a stream item $i$ with timestamp $t_i$, shown in pseudo-code as Algorithm \ref{process}, requires computing $x$, which is the forward decayed weight of the item, and incrementing \textit{count} by $x$. Note that when computing $x$, we do not normalize  the result (dividing by $g(t-L)$ where $t$ is the query time, since we do not know in advance the query time); normalization occurs instead at query time. Then, we update the $d$ cells in which the item is mapped to by the corresponding hash functions by using the Space Saving item update procedure. The worst case complexity of the update procedure is $O(\ln \frac{1}{\delta})$.

Finally, in order to retrieve the frequent items, a query can be posed when needed. Let $t$ be the query time. The query, shown in pseudo-code as Algorithm \ref{query}, initializes $R$, an empty set, and then it inspects each of the $d \times w$ cells in the sketch $D$. For a given cell, we determine $c_m$, the counter in the data structure $\mathcal{S}$ with maximum decayed count. We normalize the decayed count stored in $c_m$ dividing by $g(t-L)$, and then compare this quantity with $\phi  \frac{count}{g(t-L)}$. If the normalized decayed count is greater, we pose a point query for the item $c_m.i$, shown in pseudo-code as Algorithm \ref{estimate}. If $p$, the returned value, is greater than $\phi \frac{count}{g(t-L)}$, then we insert in $R$ the pair $(c_m.i,p)$.

The point query for an item $j$ returns its estimated decayed count. After initializing the \textit{answer} variable to infinity, we inspect each of the $d$ cells in which the item is mapped to by the corresponding hash functions, to determine the minimum decayed count of the item. In each cell, if the item is stored by one of the Space Saving counters, we set \textit{answer} to the minimum between \textit{answer} and the corresponding counter's decayed count. Otherwise (none of the two counters monitors the item $j$), we set \textit{answer} to the minimum between \textit{answer} and the minimum decayed count stored in the counters. We return the normalized \textit{answer}, dividing by $g(t-L)$.

From the previous discussion it is clear that our algorithm also solves the \textit{decayed count estimation} problem for arbitrary items. Indeed, given an item, it suffices to pose a point query for that item. Finally, since the worst case complexity of a point query is $O(\ln \frac{1}{\delta})$, the worst case complexity of the query procedure is $O(\frac{1}{\epsilon} (\ln \frac{1}{\delta})^2)$. We shall argue in section \ref{cmp-alg} that a query only takes a few milliseconds and therefore its complexity is, in practice, negligible.

\begin{algorithm}
\begin{algorithmic}[1]
\Require $\epsilon$, error; $\delta$, probability of failure; $\phi$, threshold;
\Ensure a sketch $D[1 \ldots d][1 \dots w]$ properly initialized
\Procedure {initialize}{$\epsilon, \delta, \phi, t_{init}$}
\State $d \leftarrow \lceil \ln 1/\delta \rceil$
\State $w \leftarrow \lceil \frac{e}{2\epsilon} \rceil$
\For{$i=1$ to $d$}
	\For{$j=1$ to $w$}
		\Comment{allocate a data structure $\mathcal{S}$ with two counters $c_1$, $c_2$ for $D[i][j]$}
		\State $D[i][j] \leftarrow \mathcal{S}$
	\EndFor
\EndFor
\State Set support threshold to $\phi$
\State Choose $d$ pairwise independent hash functions $h_1,\ldots,h_d:[m] \rightarrow [w]$
\State $count \leftarrow 0$
\State $L \leftarrow t_{init}$ \Comment{$t_{init}$ must be $\leq$ of all of the items' timestamps}
\EndProcedure
\caption{Initialize}
\label{init}
\end{algorithmic}
\end{algorithm}

\begin{algorithm}
\begin{algorithmic}[1]
\Require $i$, an item; $t_i$, timestamp of item $i$;
\Ensure update of sketch related to item $i$
\Procedure {process}{$i, t_i$}
\Comment{compute the decayed weight of item $i$ and update the sketch}
\State $x \leftarrow g(t_i-L)$
\State $count \leftarrow count + x$
\For{$j=1$ to $d$}
	\State $\mathcal{S} \leftarrow D[j][h_j(i)]$
	\State \Call{SpaceSavingUpdate}{$\mathcal{S}, i, x$}
\EndFor
\EndProcedure
\caption{Process}
\label{process}
\end{algorithmic}
\end{algorithm}

\begin{algorithm}
\begin{algorithmic}[1]
\Require $t$, query time
\Ensure set of frequent items
\Procedure {query}{$t$}
\State $R=\emptyset$
\For{$i=1$ to $d$}
	\For{$j=1$ to $w$}
		\State $\mathcal{S} \leftarrow D[i][j]$
		\State let $c_1$ and $c_2$ be the counters in $\mathcal{S}$, and $c_m$ the counter with maximum decayed count
		\State $c_m \leftarrow \Call{argmax}{c_1, c_2}$
		\If{$\frac{c_m.f}{g(t-L)}  > \phi \frac{count}{g(t-L)}$}	
			\State $p \leftarrow \Call{PointEstimate}{c_m.i, t}$
			\If{$p > \phi \frac{count}{g(t-L)}$}
				\State $R \leftarrow R \cup \{(c_m.i, p)\}$
			\EndIf
		\EndIf
	\EndFor
\EndFor
\State \Return $R$
\EndProcedure
\caption{Query}
\label{query}
\end{algorithmic}
\end{algorithm}

\begin{algorithm}
\begin{algorithmic}[1]
\Require $j$, an item; $t$, query time
\Ensure estimation of item $j$ decayed count;
\Procedure {pointestimate}{$j, t$}
\State $answer \leftarrow \infty$
\For{$i=1$ to $d$}
	\State $\mathcal{S} \leftarrow D[i][h_i(j)]$
	\State let $c_1$ and $c_2$ be the counters in $\mathcal{S}$
	\If{$j == c_1.i$}
		\State $answer \leftarrow \Call{min}{answer, c_1.f}$
	\Else
		\If{$j == c_2.i$}
			\State $answer \leftarrow \Call{min}{answer, c_2.f}$
		\Else
			\State $m \leftarrow \Call{min}{c_1.f, c_2.f}$
			\State $answer \leftarrow \Call{min}{answer, m}$
		\EndIf
	\EndIf
\EndFor
\State \Return $\frac{answer}{g(t-L)}$
\EndProcedure
\caption{PointEstimate}
\label{estimate}
\end{algorithmic}
\end{algorithm}

\subsection{Example}

Here, we provide a concise example of the FDCMSS update and query algorithms. Of course, illustrating a really significant example will require using a sketch of suitable dimensions, but, in the interest of clarity and to avoid wasting space, we use a sketch $D$ with $d = 2$ rows and $w = 5$ columns. Setting $d$ and $w$ implicitly determine $\delta$ and $\epsilon$. The other parameters are $\rho = 1.1$, $\phi = 0.025$ and $\lambda = 0.999$. Table \ref{fdcmss-update-example-a} shows the state of the sketch after processing 1000 items coming from a stream derived from the universe $\mathcal{U} = \{x \in \mathbb{N}: 1 \leq x \leq 20 \}$.

Table \ref{fdcmss-update-example-b} depicts the sketch state after updating the data structure upon arrival of item 6 with timestamp 1001. The hash functions associated to each row map the item 6 respectively to column 5 in the first row and to column 3 in the second row of the sketch. Since item 6 was already monitored by a counter in sketch cells $D[1][5]$ and $D[2][3]$ (see Table \ref{fdcmss-update-example-a}), Space Saving increments these counters by adding the corresponding weight which is computed as $g(t - L) = (1/ \lambda)^{t - L} = 2.72$ where $t = 1001$ is the timestamp of item 6 and $L = 0$ is the landmark time. 

Upon arrival of item 5 with timestamp 1002 (see Table \ref{fdcmss-update-example-c}), the hash functions associated to each row map the item respectively to column 1 in the first row and to column 5 in the second row of the sketch. However, this time the item is not monitored and both counters are full in the sketch cells $D[1][1]$ and $D[2][5]$. Therefore, the counters monitoring the item with minimum weight are selected and updated by Space Saving evicting the monitored items, substituting them with item 5 and incrementing them by the corresponding weight. 

 We now show how to query the sketch to retrieve the frequent item candidates. We query the sketch after updating it upon arrival of item 5 and before processing the next item. When querying the sketch data structure, all of the involved weights are normalized dividing them by $(1/ \lambda)^{t_q - L} = 2.725$, where $t_q = 1003$ is the query's timestamp. 
 
 The estimated normalized decayed count is equal at this point to  $C = 632.671$; since $\phi = 0.025$, the threshold required for an item to be considered frequent is given by $\phi C = 15.817$.
 
 We scan each of the sketch cells, determine the monitored item with maximum weight and, if the normalized weight of this item exceeds the threshold, we execute a point query for this item, which returns the normalized minimum weight associated to the item in the sketch. Then, we compare  the normalized minimum weight to the threshold again, and consider the item frequent if it exceeds the threshold.
 
  For instance, we determine that item 2 is frequent as follows. Since item 2 is the majority item in $D[1][1]$ and its normalized weight $203.78$ is greater than the threshold, we execute a point query which produces the best frequency estimate for the item. In our case the point query for item 2 returns $198.08$, and item 2 therefore is selected as candidate frequent item because its normalized minimum weight is still greater than the threshold. Table \ref{frequent-items} includes all of the frequent item candidates as returned by the query algorithm. None of the other items is selected as candidate frequent item. For instance, item 14 is not selected as candidate frequent item since its normalized minimum weight is $13.35 < 15.817$.

\begin{table}[htp]

\centering
\caption{Sketch state after 1000 updates}\label{fdcmss-update-example-a}

	\begin{tabular}{ l | r r | r r | r r | r r | r r }
		%\hline
		& Item & Weight & Item & Weight & Item & Weight & Item & Weight & Item & Weight \\
		\hline
		$\mathbf{c_1}$ & 2 & 555.33 & 3 & 262.06 & 12 & 36.55 & 10 & 52.27 & 6 & 98.22 \\
		$\mathbf{c_2}$ & 4 & 537.23 & 14 & 103.54 & 17 & 14.78 & 18 & 21.88 & 11 & 36.76 \\
		\hline 
		$\mathbf{c_1}$  & 4 & 172.20 & 14 & 36.40 & 6 & 125.75 & 2 & 539.78 & 3 & 263.07\\
		$\mathbf{c_2}$  & 12 & 109.28 & 16 & 35.78 & 7 & 125.15 & 8 & 117.33 & 10 & 193.90  \\
		\hline
	\end{tabular}
 
    \bigskip
           
    \caption{Sketch state after arrival of item 6 with timestamp 1001}\label{fdcmss-update-example-b}

		\begin{tabular}{ l | r r | r r | r r | r r | r r }
			%\hline
			& Item & Weight & Item & Weight & Item & Weight & Item & Weight & Item & Weight\\
			\hline
			$\mathbf{c_1}$ & 2 & 555.33 & 3 & 262.06 & 12 & 36.55 & 10 & 52.27 & 6 & 100.94\\ 
			$\mathbf{c_2}$ & 4 & 537.23 & 14 & 103.54 & 17 & 14.78 & 18 & 21.88 & 11 & 36.76\\
			\hline 
			$\mathbf{c_1}$  & 4 & 172.20 & 14 & 36.40 & 6 & 128.47 & 2 & 539.78 & 3 & 263.07\\ 
			$\mathbf{c_2}$  & 12 & 109.28 & 16 & 35.78 & 7 & 125.15 & 8 & 117.33 & 10 & 193.90\\
			\hline
		\end{tabular}

	\bigskip
	    
	\caption{Sketch state after arrival of item 5 with timestamp 1002}\label{fdcmss-update-example-c}

			\begin{tabular}{ l | r r | r r | r r | r r | r r }
			%\hline
			& Item & Weight & Item & Weight & Item & Weight & Item & Weight & Item & Weight \\
			\hline
			$\mathbf{c_1}$ & 2 & 555.33 & 3 & 262.06 & 12 & 36.55 & 10 & 52.27 & 6 & 100.94\\
			$\mathbf{c_2}$ & 5 & 539.95 & 14 & 103.54 & 17 & 14.78 & 18 & 21.88 & 11 & 36.76 \\ 
			\hline
			$\mathbf{c_1}$  & 4 & 172.20 & 14 & 36.40 & 6 & 128.47 & 2 & 539.78 & 3 & 263.07\\
			$\mathbf{c_2}$  & 12 & 109.28 & 16 & 35.78 & 7 & 125.15 & 8 & 117.33 & 5 & 196.62 \\
			\hline
		\end{tabular}

\end{table}

\begin{table}
\renewcommand{\arraystretch}{1.3}
 \caption{Frequent Item Candidates}
      \label{frequent-items}
	\centering
    \begin{tabular}{| c |  c | }
    \hline
      Item & Normalized Decayed Weight \\ \hline
      \hline
    2 & 198.08 \\ \hline
    3 & 96.16 \\ \hline
    5 & 72.15 \\ \hline
    6 & 37.04  \\ \hline
    \end{tabular}
\end{table}

\section{Error bound}
\label{bound}

In this section, we formally prove the error bound of our algorithm. Let $f_i$ and $\hat{f}_i$ be respectively the exact and estimated decayed count of item $i$. We denote by $\mathcal{S}$ a Space Saving summary, which is a data structure present in each one of the cells available in the sketch $D$. Let $C$ be the decayed count of all of the items in the stream $\sigma$ (see Definition \ref{decayed-count}). The main result of this section is the following theorem.

\begin{thm}
\label{error-bound}
$\forall k \in [m]$, $\hat{f}_k$ estimates the exact decayed count $f_k$ with error less than $\epsilon C$ and probability greater than $1-\delta$.
\end{thm}

\begin{proof}
The algorithm is based on the use of $d$ pairwise independent hash functions $h_1,\ldots,h_d:[m] \rightarrow [w]$. In the following, we shall use related indicator random variables $I_{i,j,k}$ defined as follows:

\begin{equation}
{I_{i,j,k}}{\text{ = }}\left\{ {\begin{array}{*{20}{c}}
		1&{ {h_i}(k) = j} \\ 
		0&{{\text{otherwise}}} 
		\end{array}} \right..
\end{equation}
	
In other words, the indicator random variables $I_{i,j,k}$ are equal to one when the item $k \in [m]$ falls in the $D[i][j]$ cell. By pairwise independence of the hash functions, it follows that the expected value of $I_{i,j,k}$ is

\begin{equation}
\E[I_{i,j,k}] = \Pr[h_i(k)=j] = \frac{1}{w} =\frac{2\epsilon}{e}.
\end{equation}

Denote by $S_{i,j}$ the sum of the decayed counts of the items falling in the cell $D[i][j]$, i.e., 

\begin{equation}
S_{i,j} = \sum_{k=1}^m{f_k I_{i,j,k}}.
\end{equation}

Our algorithm processes each item falling in the same cell by using Space Saving with two counters. Now, we bound the error committed by Space Saving. By eq. (\ref{ss5}), the error of an item monitored by a Space Saving counter is bounded by the sum of the decayed counts of the items which fall in the same cell divided by the number of counters. Therefore, denoting by  $\hat{f}_{i,k}$ the estimated decayed count of item $k$ returned by Space Saving for the $D[i][h_i(k)]$ cell and by  $f_k$ its exact decayed count, it holds that

\begin{equation}
\hat{f}_{i,k} - f_k \leq \frac{S_{i,j}}{2}.
\end{equation}

By linearity of expectation, 

\begin{equation}
\E \left[ S_{i,j} \right] = \E\left[\sum_{k=1}^m{f_k I_{i,j,k}}\right] = \sum_{k=1}^m{f_k  \E[I_{i,j,k}]} = \frac{C}{w} = \frac{2 \epsilon C}{e},
\label{linearity-expectation}
\end{equation}

with $C$ equal to the total decayed count, as defined in Definition \ref{decayed-count}. It holds that, on average,

\begin{equation}
\E [\hat{f}_{i,k} - f_k] \leq \frac{1}{2} \E\left [S_{i,j} \right ] = \frac{2 \epsilon C}{2 e} = \frac{\epsilon C}{e}.
\end{equation}

Since the error $\hat{f}_{i,k} - f_k$ is nonnegative, we can apply the Markov inequality, obtaining

\begin{equation}
	\Pr[\hat{f}_{i,k} - f_k \geq \epsilon C] \leq \frac{\E[\hat{f}_{i,k} - f_k]}{\epsilon C}
	\leq \frac{\epsilon C}{e\epsilon C}=e^{-1}.
	\end{equation}
	
It follows that, when considering all of the $d$ cells in which items are mapped to by the independent hash functions ${h_i, i=1,\ldots,d}$, and recalling that the estimated decayed count $\hat{f}_k = \min{\{\hat{f}_{1,k}, \dots, \hat{f}_{d,k}\}}$, we get

\begin{equation}
	\begin{aligned}
	\Pr[\hat{f}_k - f_k \geq \epsilon C] &= 
	\Pr [\min{\{\hat{f}_{1,k}, \dots, \hat{f}_{d,k}\}} - f_k \geq \epsilon C] \\ 
	&= \Pr \left[\bigwedge_{i=1}^d (\hat{f}_{i,k} - f_k \geq \epsilon C) \right] \\
	&= \prod_{i=1}^{d} \Pr [\hat{f}_{i,k} - f_k \geq \epsilon C] \\
	&\leq e^{-d} = \delta.
	\end{aligned}
	\end{equation}

Consequently, 

\begin{equation}
\Pr[\hat{f}_k - f_k < \epsilon C] > 1-\delta.
\end{equation}
\end{proof}

\begin{rmk}
Here, we explain the reason why we set $w = \frac{e}{2 \epsilon}$ when we initialize the sketch data structure. In Theorem \ref{error-bound}, letting $w$ unspecified in equation (\ref{linearity-expectation}), we can easily derive the relationship between $d$ and $w$ at the end of the Theorem: $d = \frac{{\log \frac{1}{\delta }}}{{\log 2w\varepsilon}}$. Setting $w = \frac{e}{2 \epsilon}$ we minimize the total space occupied by the sketch data structure. Indeed, the total space is $wd = w\frac{{\log \frac{1}{\delta }}}{{\log 2w\varepsilon }}$; minimizing analytically $wd$ with regard to $w$, we obtain the minimum for $w = \frac{e}{2 \epsilon}$. 
\end{rmk}

\section {Correctness}
\label{correctness}
We are going to formally prove the correctness of our algorithm. Before discussing its correctness, it is worth noting here that given a cell in the sketch $D$, the sum of the decayed counts stored by its two Space Saving counters is equal to the value that the Count-Min algorithm would store in that cell. However, Count-Min relies on an external heap data structure to keep track of frequent items. By using a data structure $\mathcal{S}$, with just two Space Saving counters per cell, we are able to dynamically maintain frequent items. Therefore, by using $O(dw)$ space as in Count-Min, our algorithm can solve both the \textit{decayed count estimation} and the \textit{approximate frequent items under forward decay} problems.
 
Here, we show that, with high probability, if an item is frequent, our algorithm will detect it. Indeed, given a cell, by using a data structure $\mathcal{S}$ with two counters, we are able to detect the \textit{majority item candidate} with regard to the sub-stream of items falling in that cell. Letting $S_{i,j}$ denote the total decayed count of the items falling in the cell $D[i][j]$, the majority item is, if it exists, the item whose decayed count is greater than $\frac{S_{i,j}}{2}$. The corresponding majority item candidate in the cell is the item monitored by the Space Saving counter whose estimated decayed count is maximum.

The main result of this section is the following theorem.

\begin{thm}
\label{thm-correctness}
If an item $i$ is frequent, then it appears as a majority item candidate in at least one of the $d$ cells in which it falls with probability greater than or equal to $1 - (\frac{1}{2 \phi w})^d$. 
\end{thm}

\begin{proof}
Let $k$ be a frequent item, $j = h_i(k)$, $D[i][j]$ one of the $d$ cells in which the item $k$ is mapped to by the corresponding hash function and $f^{(i,j)}_{min}$ the minimum of the two Space Saving counters available in the data structure $\mathcal{S}$ monitoring the items falling in $D[i][j]$. Moreover, let $\hat{f}_{i,k}$ be the estimated decayed count of item $k$ returned by Space Saving for the $D[i][h_i(k)]$ cell.

Our algorithm will not output the item $k$ (and therefore will not be correct) iff for all of the $d$ cells in which the item $k$ is mapped to by the corresponding hash functions, the item $k$ is not reported as a majority item candidate, so that

\begin{equation}
\label{failure_eq}
\hat{f}_{i,k} \leq f^{(i,j)}_{min}, \forall i=1,\ldots,d.
\end{equation}

Indeed, by construction, our algorithm during a frequent item query only checks if an item is frequent when the item is reported in the cell as a majority item candidate. By assumption, since the item $k$ is frequent, its decayed count is  $f_k > \phi C$; since $f_k \leq \hat{f}_{i,k} \forall i=1,\ldots,d$, it holds that

\begin{equation}
\phi C < f_k \leq \hat{f}_{i,k} \forall i=1,\ldots,d.
\end{equation}

Let $S_{i,j}$ denote the total decayed count of the items falling in the cell $D[i][j]$. By eq. (\ref{ss5}), for the minimum decayed count in a cell $D[i][j]$ it holds that $f^{(i,j)}_{min} \leq \frac{S_{i,j}}{2}$.

We must now determine the probability that the event described in eq. (\ref{failure_eq}) occurs. This is the probability of failing to correctly recognize a frequent item. Taking into account that $\phi C <\hat{f}_{i,k}$ and $f^{(i,j)}_{min} \leq \frac{S_{i,j}}{2}$, it holds that

\begin{equation}
\label{failure_final_eq}
S_{i,j} > 2\phi C.
\end{equation}

By the previous argument, it follows that

\begin{equation}
\label{prob}
\Pr[\hat{f}_{i,k} \leq f^{(i,j)}_{min}] < \Pr[S_{i,j} > 2\phi C].
\end{equation}

Reasoning as in section \ref{bound}, $\E[S_{i,j}] = \frac{C}{w}$ with $C$ equal to the total decayed count as defined in Definition \ref{decayed-count}. Using the Markov inequality we can bound the probability of failure (i.e., the probability of the item not being reported as a majority item candidate) in a single cell $D[i][j]$ taking into account eq. (\ref{failure_final_eq}) and (\ref{prob}) 

\begin{equation}
\label{single-cell-failure}
\Pr[\hat{f}_{i,k} \leq f^{(i,j)}_{min}] < \Pr[S_{i,j} > 2\phi C] \leq \frac{\E[S_{i,j}]}{2\phi C} = \frac{C}{2 \phi w C} = \frac{1}{2 \phi w}.
\end{equation}

Therefore, we fail to identify a frequent item $k$ when in all of the $d$ cells $D[i][j], i=1,\ldots,d$ in which the frequent item falls it is not reported as a majority item candidate in the data structure $\mathcal{S}$. We now estimate the corresponding probability. By eq. (\ref{single-cell-failure}), the probability of failure is 

\begin{equation}
\Pr\left[\bigwedge_{i=1}^{d} (\hat{f}_{i,k} \leq f^{(i,j)}_{min})\right]  < \Pr\left[\bigwedge_{i=1}^{d} (S_{i,j} > 2\phi C)\right] \leq \left(\frac{1}{2 \phi w}\right)^d.
\end{equation}

Consequently, we succeed with probability greater than or equal to $1 - \left(\frac{1}{2 \phi w}\right)^d$. 

\end{proof}

\begin{rmk}
Since we proved in Theorem \ref{thm-correctness} that, if an item is frequent, then it appears as a majority item candidate in at least one of the cells in the sketch with high probability, it follows by the Space Saving design that two Space Saving counters are necessary and sufficient to determine this majority item candidate. Using more than two counters is useless for our purposes, and only wastes precious space.
\end{rmk}

\section{Theoretical comparison of $FDCMSS$ and $\lambda$-HCount}
\label{cmp-alg}

We provide here a thorough comparison of FDCMSS and $\lambda$-HCount. Both algorithms use a sketch data structure; FDCMSS is based on Count-Min and $\lambda$-HCount on the \emph{hCount} algorithm \cite{Jin03}. However, $\lambda$-HCount relies on a backward decay exponential function, whilst FDCMSS can use either an exponential function or any other forward decay function. In particular, the use of a polynomial function allows more flexibility with regard to time fading. Indeed, using an exponential function the time fades faster, whilst with a polynomial function the times fades more slowly. Another advantage of using forward decay is that FDCMSS can easily deal with out of order arrival of stream items \cite{forward-decay}, something requiring significant effort to accommodate for $\lambda$-HCount.

Another difference is the use in $\lambda$-HCount of a dedicated data structure $F$ to keep track of frequent items. Instead, our algorithm FDCMSS does not require additional space beyond its sketch data structure. Even though a query for frequent items requires in the worst case $O(\frac{1}{\epsilon} (\ln \frac{1}{\delta})^2)$, a query execution only takes a few milliseconds and therefore its complexity is, in practice, negligible (this has been verified in all of the experimental tests carried out). Moreover, queries are posed to FDCMSS from time to time whilst updates happen with high frequency, especially in high-speed streams. Therefore, FDCMSS has been designed to provide very fast updates, besides accurate results.

Theoretically, the main drawback of $\lambda$-HCount lies in the huge amount of space required to attain its error bound. In particular, the $\lambda$-HCount sketch requires $\frac{e (1-\lambda)\ln{(-\frac{M}{\ln{p}})}}{\epsilon^2}$  cells, where $M$ is the number of distinct items and $p$ is the success probability. Without taking into account the additional data structure $F$ requiring $\frac{r}{s - \epsilon}$ entries, where $r$ is the number of hash functions used, for a sketch using a total of $r \times m$ cells, we have:

\begin{equation}
r \times m =  \left \lceil \frac{e (1-\lambda)\ln{(-\frac{M}{\ln{p}})}} {\epsilon^2} \right \rceil,
\end{equation}

whilst, in our case, FDCMSS requires only

\begin{equation}
d \times w = \left \lceil \ln{\frac{1}{\delta}} \frac{e}{2\epsilon} \right \rceil.
\end{equation}

As an example, fixing $\lambda = 0.99$, $M = 1048575$, $p = 0.96$ and $\epsilon = 0.001$, a total of $r \times m = 463779$ cells are required by $\lambda$-HCount. In order to achieve the same success probability $p$, in FDCMSS we need to set $\delta = 0.04$ and $\epsilon = 0.001$, so that $p = 1 - \delta = 0.96$ and a total of just $d \times w = 4375$ cells is required instead by our algorithm.  

Let us now consider only the sketch size, without taking into account the $\lambda$-HCount additional data structure $F$ required for tracking the frequent items. Each cell in the $\lambda$-HCount sketch stores a decayed count (a double, 8 bytes) and a timestamp (a long, 8 bytes), whilst a FDCMSS cell stores two Space Saving counters. A counter keeps track of an item (an unsigned int, 4 bytes) and its decayed count (a double, 8 bytes). Therefore, $\lambda$-HCount requires 16 bytes per cell and FDCMSS 24 bytes per cell.

Figures \ref{sketch-size-p} and \ref{sketch-size-epsilon} plot, using a logarithmic scale, the sketch size in kilobytes required respectively as a function of the success probability $p$ and as a function of $\epsilon$ for both algorithms. Here, we have fixed for the first plot the values $\lambda = 0.99$, $M = 1048575$, $\epsilon = 0.001$, and let $p$ vary from 0.7 to 0.99. Similarly, for the second plot we have fixed $\lambda = 0.99$, $M = 1048575$, $p = 0.96$, and let $\epsilon$ vary from 0.001 to 0.01.

\begin{figure}[hbt]
  \centering
  \begin{tabular}{cc}
     \subfloat[Sketch size required as a function of $p$]{
           \includegraphics[scale=0.8]{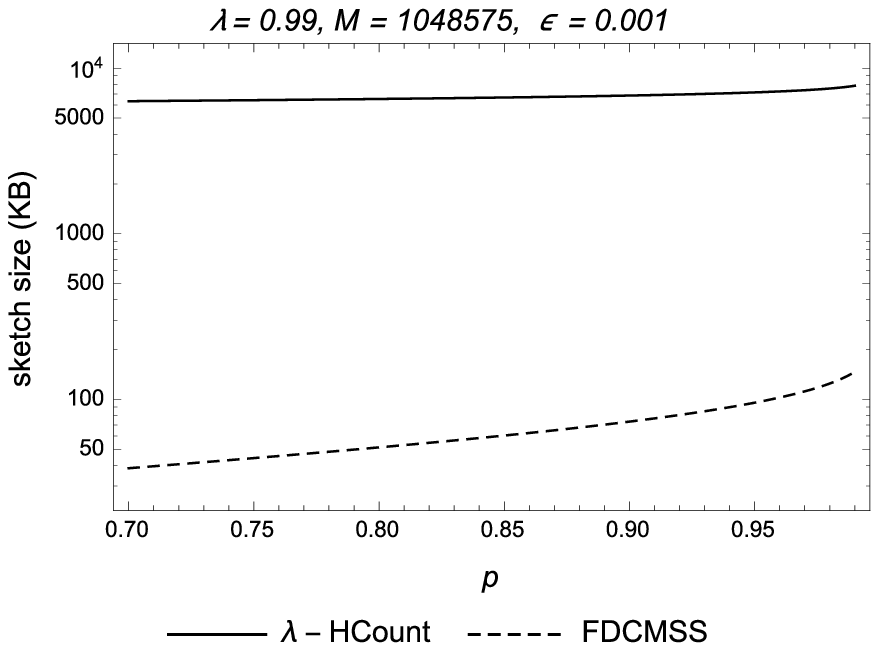}
           \label{sketch-size-p}
        } &
        
      \subfloat[Sketch size required as a function of $\epsilon$]{
           \includegraphics[scale=0.8]{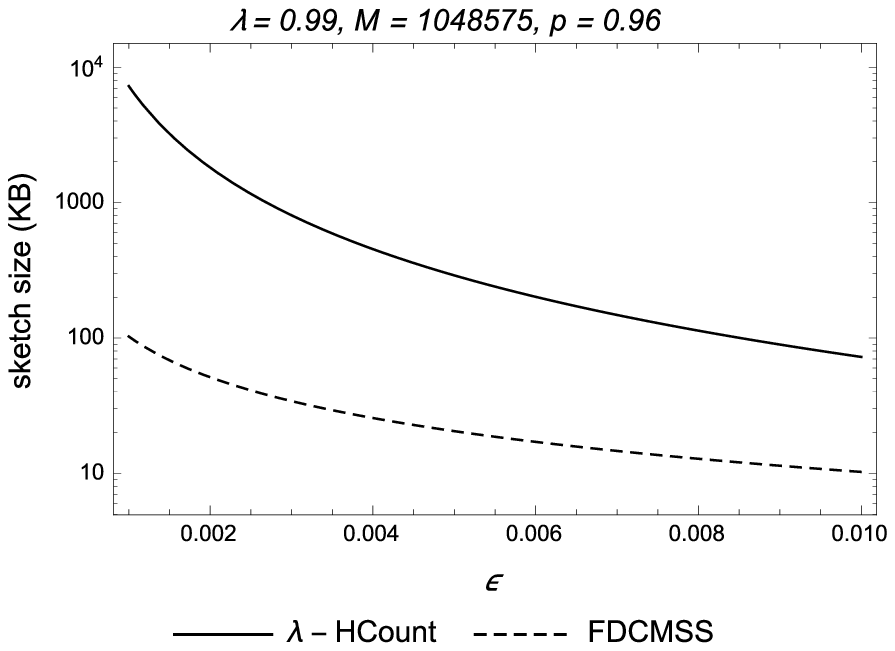}
           \label{sketch-size-epsilon}
        }

 \end{tabular}
 
 \caption{Sketch sizes in Kilobytes} 
 \label{sketch-space}
\end{figure}

It is immediate verifying that, in order to attain its theoretical error bound, $\lambda$-HCount requires a huge amount of space. On the contrary, FDCMSS achieves its bound using a tiny fraction of the space required by $\lambda$-HCount.

\section{Experimental results}
\label{results}

In this section, we report experimental results on both synthetic and real datasets. Here, we thoroughly compare our algorithm against $\lambda$-HCount \cite{Chen-Mei} with regard to the performances on synthetic and real datasets.

\subsection{Synthetic datasets} 

We have implemented FDCMSS and $\lambda$-HCount in C++. FDCMSS uses the \textit{xxhash} hash function, and $\lambda$-HCount the \textit{FNV} hash function as stated by its authors in \cite{Chen-Mei}. The code has been compiled using the clang c++ compiler v7.0 on Mac OS X v10.11.2 with the following flags: -Os -std=c++11. We recall here that, on Mac OS X, the optimization flag -Os provides better optimization than the -O3 flag and is the standard for building the release build of an application. The tests have been carried out on a machine equipped wth 16 GB of RAM and a 3.2 GHz quad-core Intel Core i5 processor with 6 MB of cache level 3.

Regarding synthetic datasets, the input distribution used in our experiments is the Zipf distribution. For each different value of $n$ (number of items), $\phi$ (support threshold), $\rho$ (skew of distribution) and sketch size, the algorithms have been run 20 times using a different seed for the pseudo-random number generator associated to the distribution (using the same seeds in the corresponding executions of different algorithms). For each input distribution generated, the results have been averaged over all of the runs. The input elements are 32 bits unsigned integers. 

In order to provide a fair comparison of the algorithms, we make sure that the decayed frequencies computed by both algorithms are equal. To this end, we use in FDCMSS the same exponential decay function $g(t - L) = (\frac{1}{\lambda})^{t - L}$ (in which the landmark time is $L = 0$) and the same $\lambda = 0.99$ parameter as in $\lambda$-HCount. This way, for a given input stream, the decayed counts of the input items and the set of frequent items computed by an exact algorithm are the same for both algorithms. However, it is worth reporting here that the use of a forward decay polynomial function provides even better results for FDCMSS in terms of speed, measured as \textit{updates per millisecond}.

We compare our algorithm against $\lambda$-HCount taking into account the following standard metrics: \textit{recall}, \textit{precision}, \textit{mean absolute error}, \textit{max absolute error}, \textit{96-th percentile absolute error}, \textit{updates per millisecond}. For each metric, we plot the values (mean and confidence intervals) obtained varying $n$, $\phi$, $\rho$ and the sketch size. In particular, for each plot we shall always compare the algorithms by using exactly the same sketch size in kilobytes. We will not take into account the additional space required by $\lambda$-HCount for its $F$ data structure.

Recall, shown in Figure \ref{recall}, is the total number of true frequent items reported over the number of true frequent items given by an exact algorithm. Therefore, an algorithm is correct iff its recall is equal to 1 (or 100\%). We note here that $\lambda$-HCount recall is always 1 since the algorithm inserts an item in the $F$ data structure only when that item has been detected as frequent. FDCMSS may instead provide a recall value lower than 1 (this follows immediately by Theorem \ref{thm-correctness}; however, the probability of failing to detect a frequent item may be made arbitrarily small and close to zero by the user, setting appropriately the input parameters $\delta$ and $\epsilon$). This happened in our experiments in only one case, when using a very small sketch size of only 6 KB. However, the measured recall value is 99.58\% even in this extreme case. 

Precision, shown in Figure \ref{precision}, is defined as the total number of true frequent items reported over the total number of items reported. As such, this metric quantifies the number of false positives outputted by an algorithm. It follows that, from this point of view, an algorithm's precision should ideally be 1 (or 100\%). The precision achieved by FDCMSS is 1 in the majority of the experiments, and our algorithm outperformed $\lambda$-HCount in particular when varying $n$ and the sketch size, whilst providing anyway higher precision when varying $\phi$ and $\rho$. 

Denoting with $f_i$ the true decayed count of item $i$ and with $\hat{f_i}$ the corresponding decayed count computed by an algorithm, then the absolute error is, by definition, the difference $\left| f_i - \hat{f_i} \right|$. Denoting with $M$ the number of distinct items in a stream (i.e., the stream domain size), the \textit{mean absolute error} is then defined as $\sum\limits_{i = 1}^M {\frac{{\left| f_i - \hat{f_i} \right|}}{M}}$, i.e., the mean of the absolute errors. Similarly, the \textit{max absolute error} is defined as ${\max _i}\left| {{f_i} - {{\hat f}_i}} \right|$. Finally, consider the $M$ absolute errors in ascending sorted order: the 96-th percentile is the absolute error found in the position corresponding to 96\% of $M$.

Figures \ref{mae}, \ref{maxae} and \ref{percentile} show, respectively, the mean, max and 96-th percentile of absolute errors. FDCMSS outperforms $\lambda$-HCount, in particular with regard to the experiments in which we vary $n$ and $\phi$. 

Regarding the skew, when the value of $\rho$ increases, a reduction in the error committed and a corresponding increase in the accuracy is expected, owing to the Zipfian distribution. Indeed, the number of frequent items depends on the value of $\rho$ as follows. Increasing $\rho$ decreases the number of distinct items, i.e., we have less items but with higher frequency. Vice versa, decreasing $\rho$, we have more items but with lower frequency. It follows that, considering the Count-Min sketch, the number of collisions decreases because there are less distinct numbers and the sketch cells can therefore better estimate the items' frequencies.

Let us now discuss the actual performances of the algorithms in terms of updates per millisecond, where an update is defined respectively as in Algorithm \ref {l_count} for $\lambda$-HCount and Algorithm \ref{process} for FDCMSS. As shown in Figure \ref{updates}, FDCMSS outperforms $\lambda$-HCount in all of the experiments carried out.

\begin{figure}[hbt]
  \centering
  \begin{tabular}{cccc}
     \subfloat[varying $n$]{
           \includegraphics[scale=0.36]{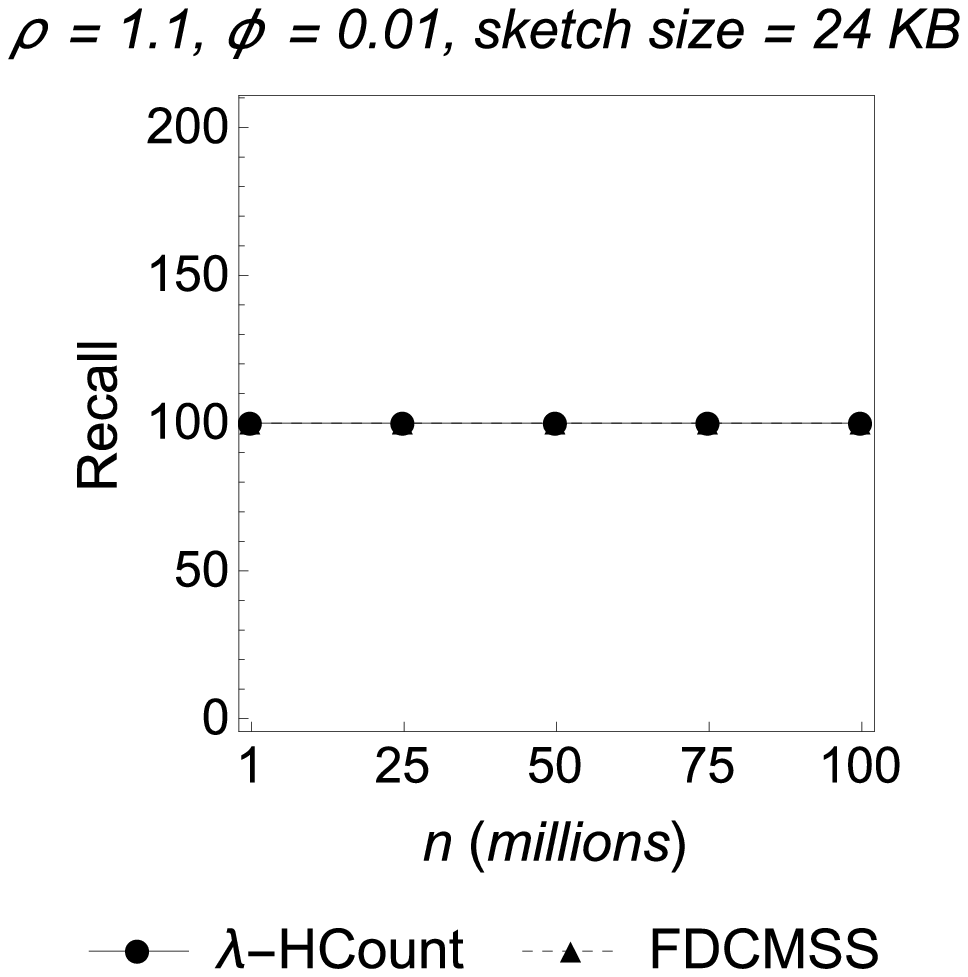}
           \label{ni-recall}
        } &
        
      \subfloat[varying $\phi$]{
           \includegraphics[scale=0.36]{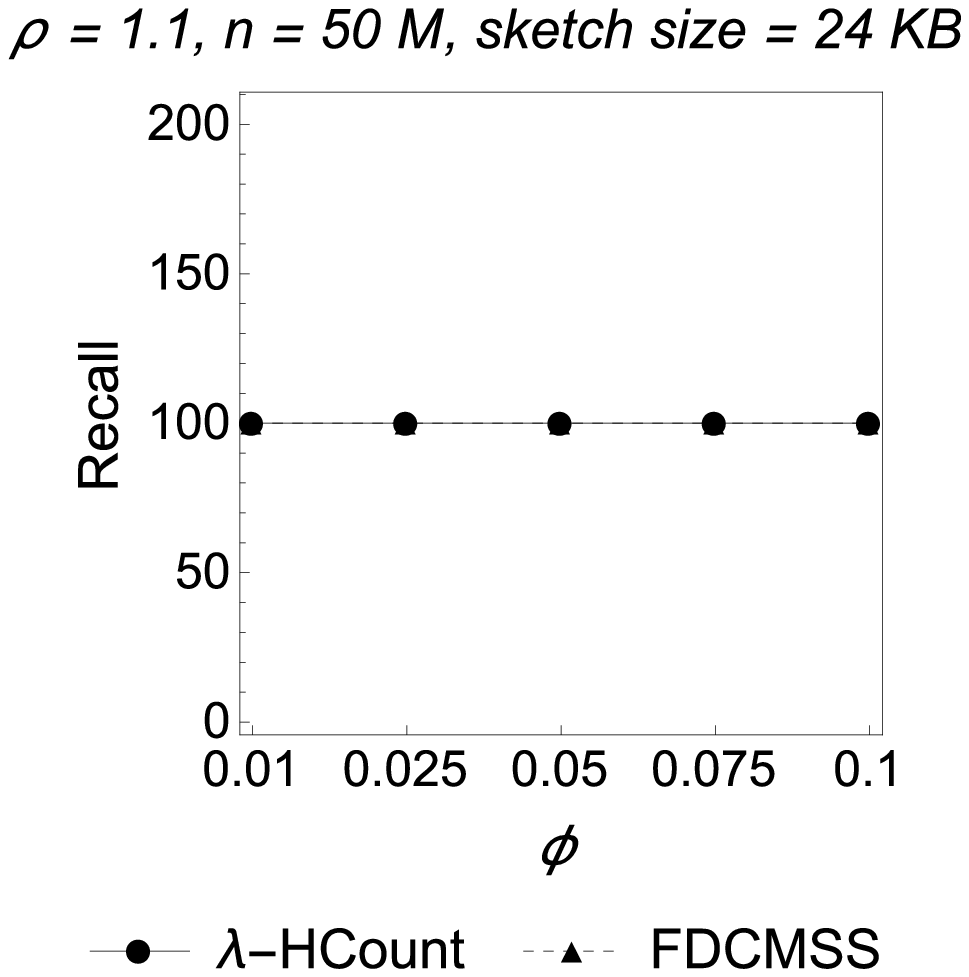}
           \label{phi-recall}
        } & 

      \subfloat[varying $\rho$]{
           \includegraphics[scale=0.36]{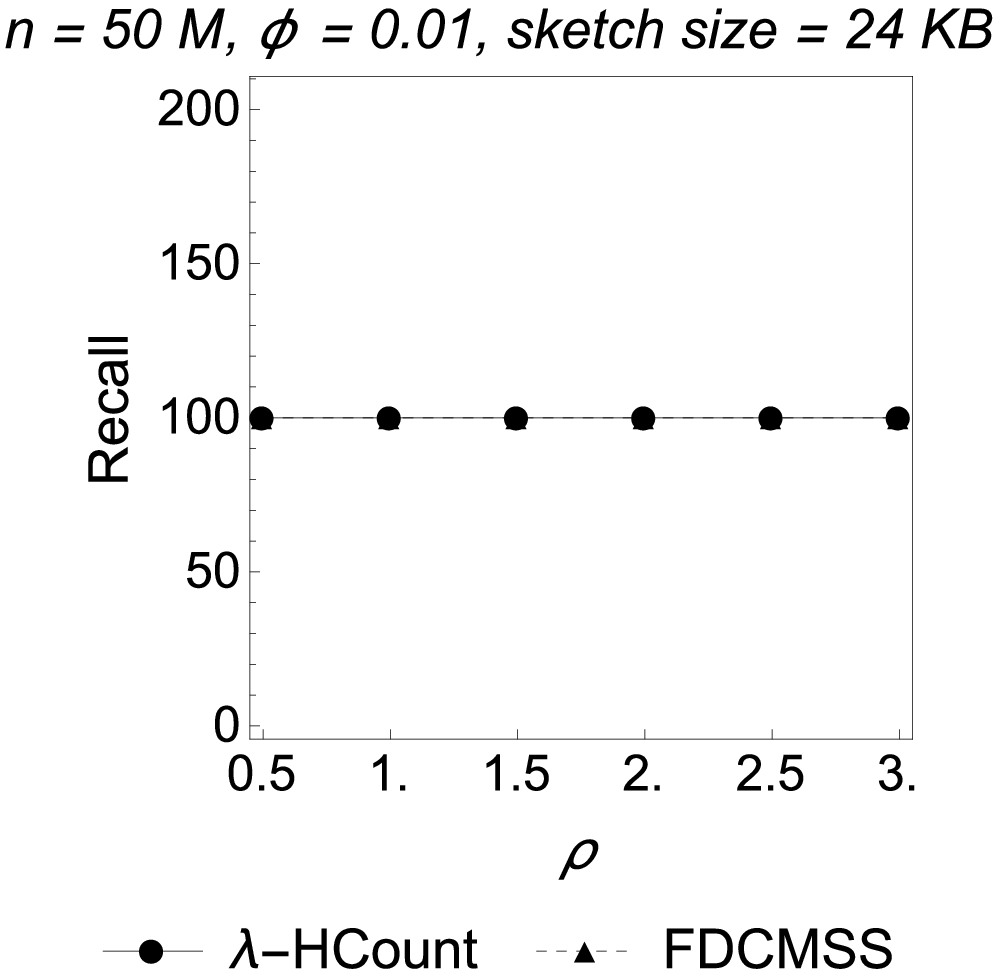}
           \label{sk-recall}
        } &
        
      \subfloat[varying the sketch size]{
           \includegraphics[scale=0.36]{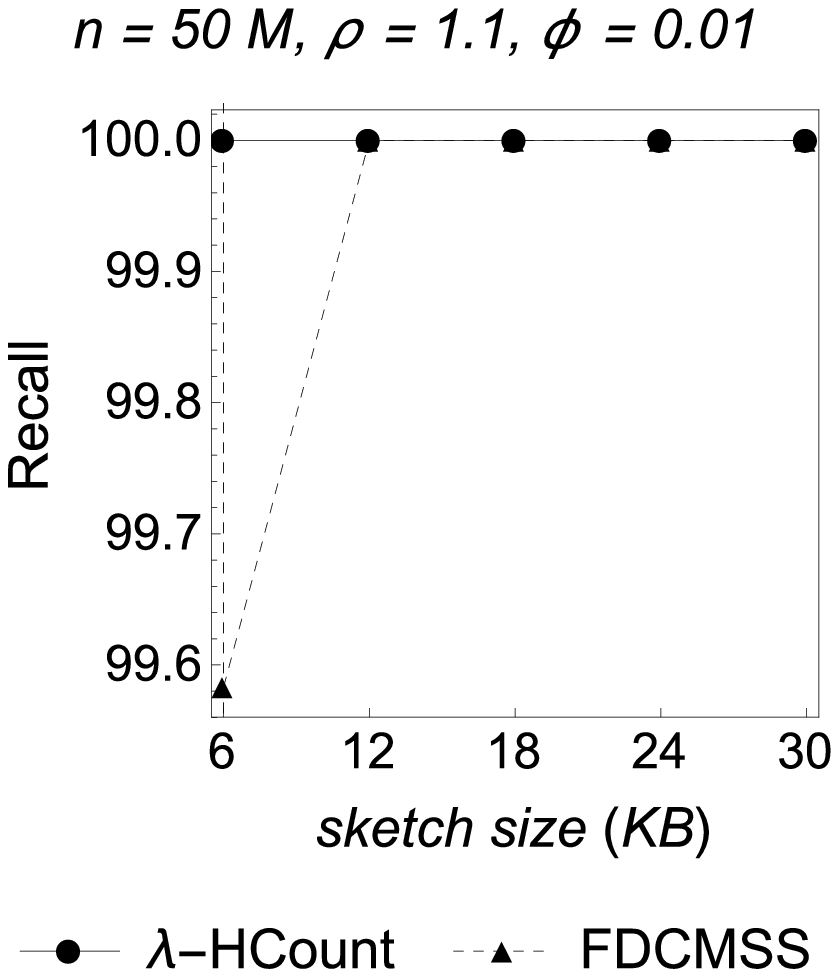}
           \label{sw-recall}
        } 

\end{tabular}
 
 \caption{Recall (mean and confidence interval)} 
 \label{recall}
\end{figure}

\begin{figure}[hbt]
  \centering
  \begin{tabular}{cccc}
     \subfloat[varying $n$]{
           \includegraphics[scale=0.36]{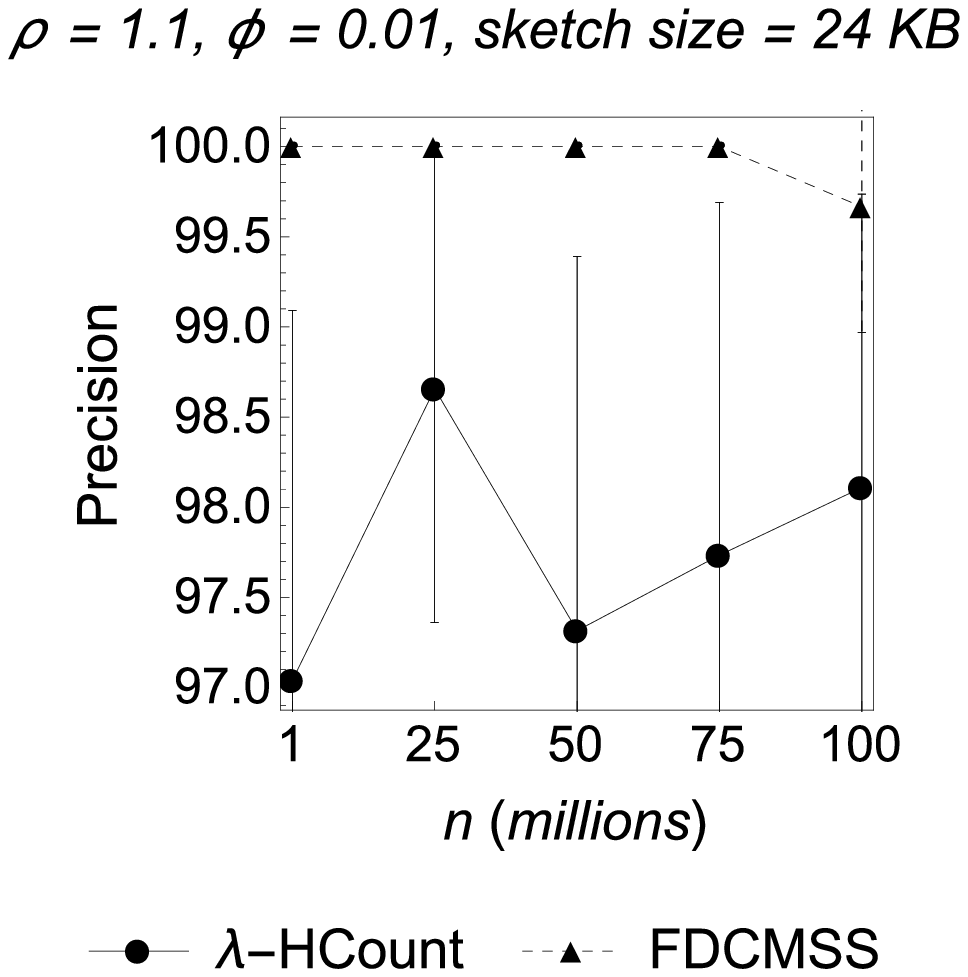}
           \label{ni-prec}
        } &
        
      \subfloat[varying $\phi$]{
           \includegraphics[scale=0.36]{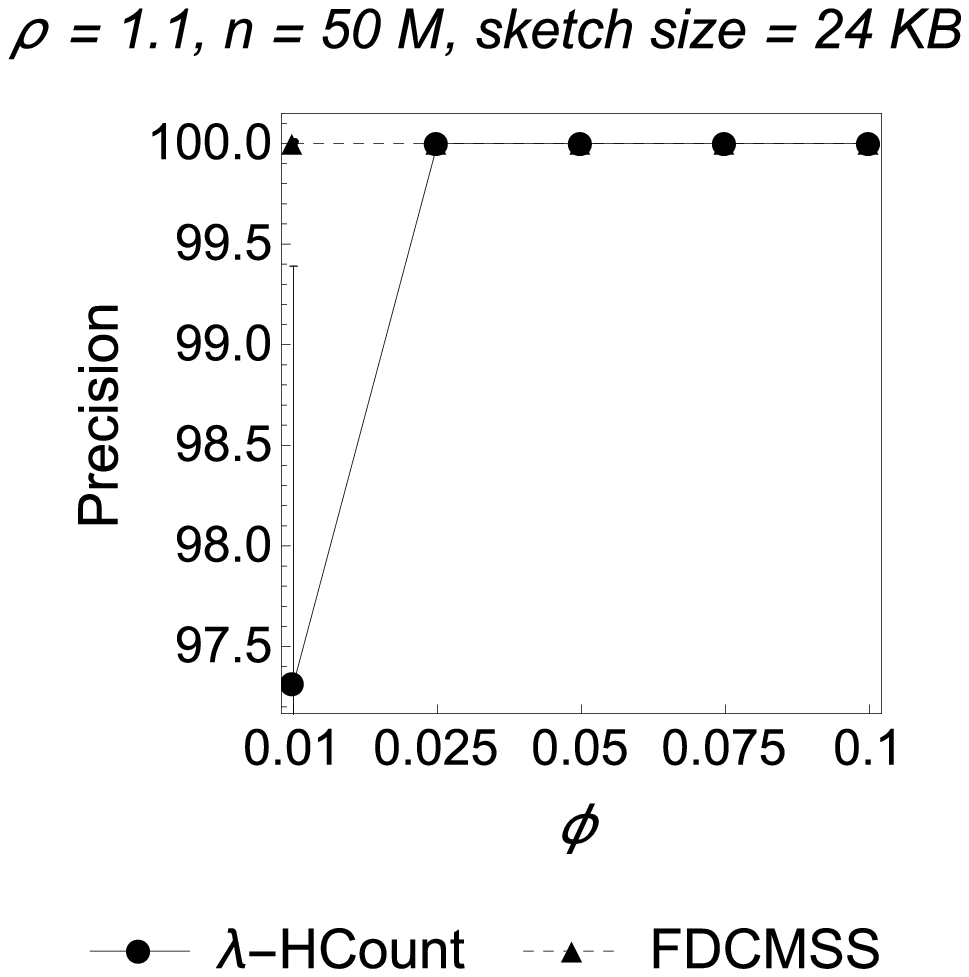}
           \label{phi-prec}
        } & 

      \subfloat[varying $\rho$]{
           \includegraphics[scale=0.36]{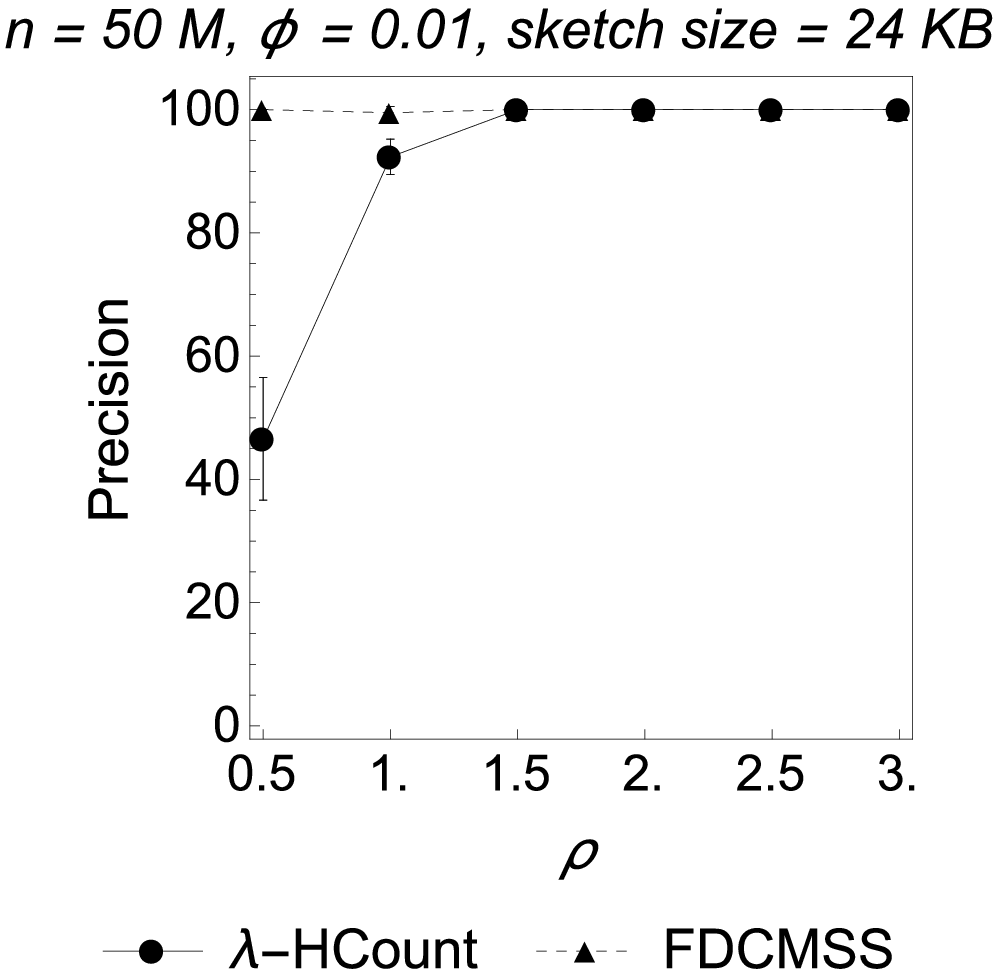}
           \label{sk-prec}
        } &
        
      \subfloat[varying the sketch size]{
           \includegraphics[scale=0.36]{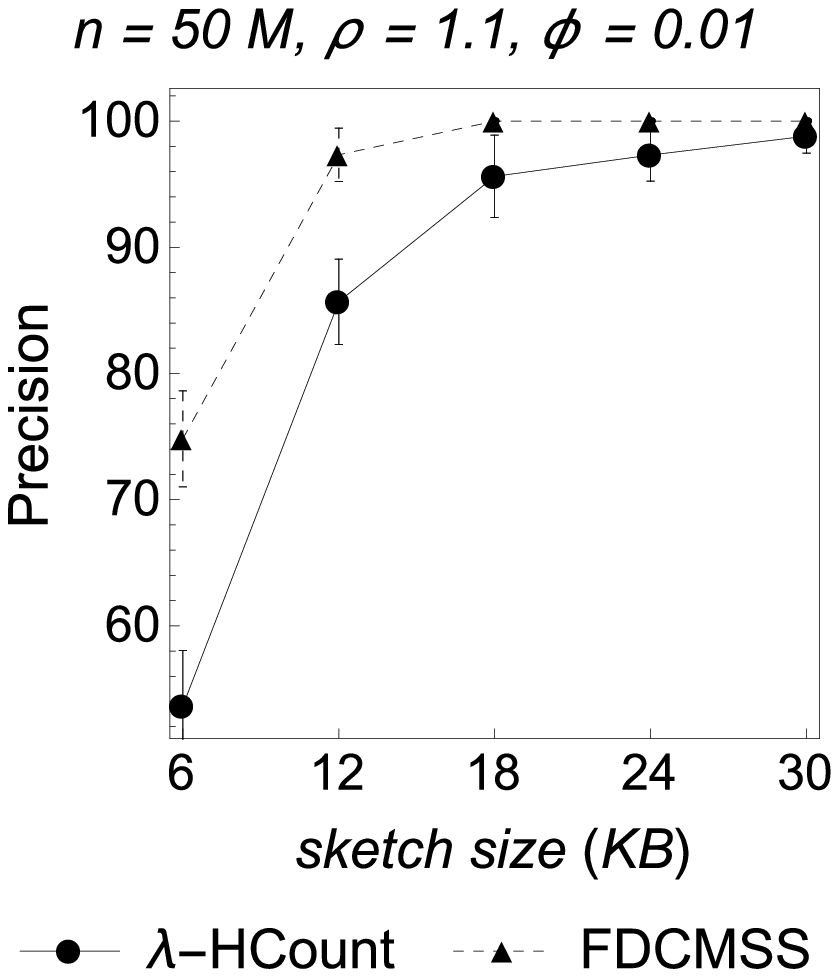}
           \label{sw-prec}
        } 

\end{tabular}
 
 \caption{Precision (mean and confidence interval)} 
 \label{precision}
\end{figure}

\begin{figure}[hbt]
  \centering
  \begin{tabular}{cccc}
     \subfloat[varying $n$]{
           \includegraphics[scale=0.36]{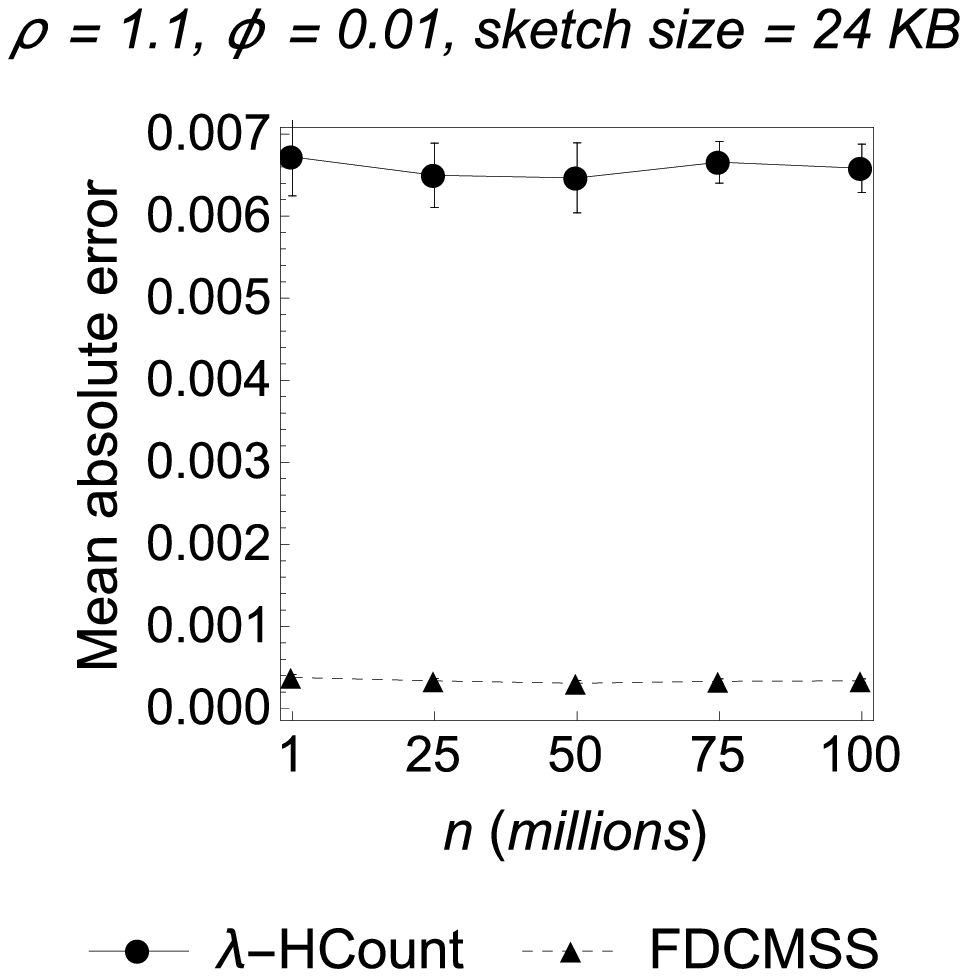}
           \label{ni-mae}
        } &
        
      \subfloat[varying $\phi$]{
           \includegraphics[scale=0.36]{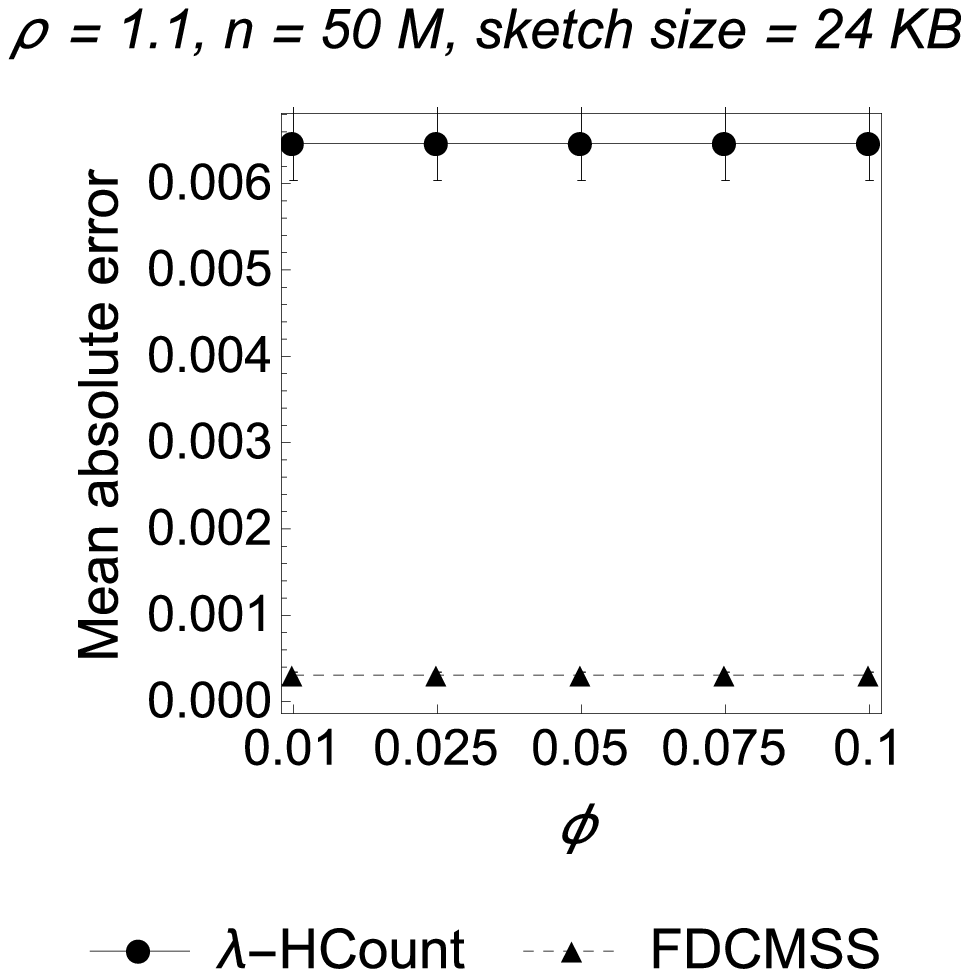}
           \label{phi-mae}
        } &

      \subfloat[varying $\rho$]{
           \includegraphics[scale=0.36]{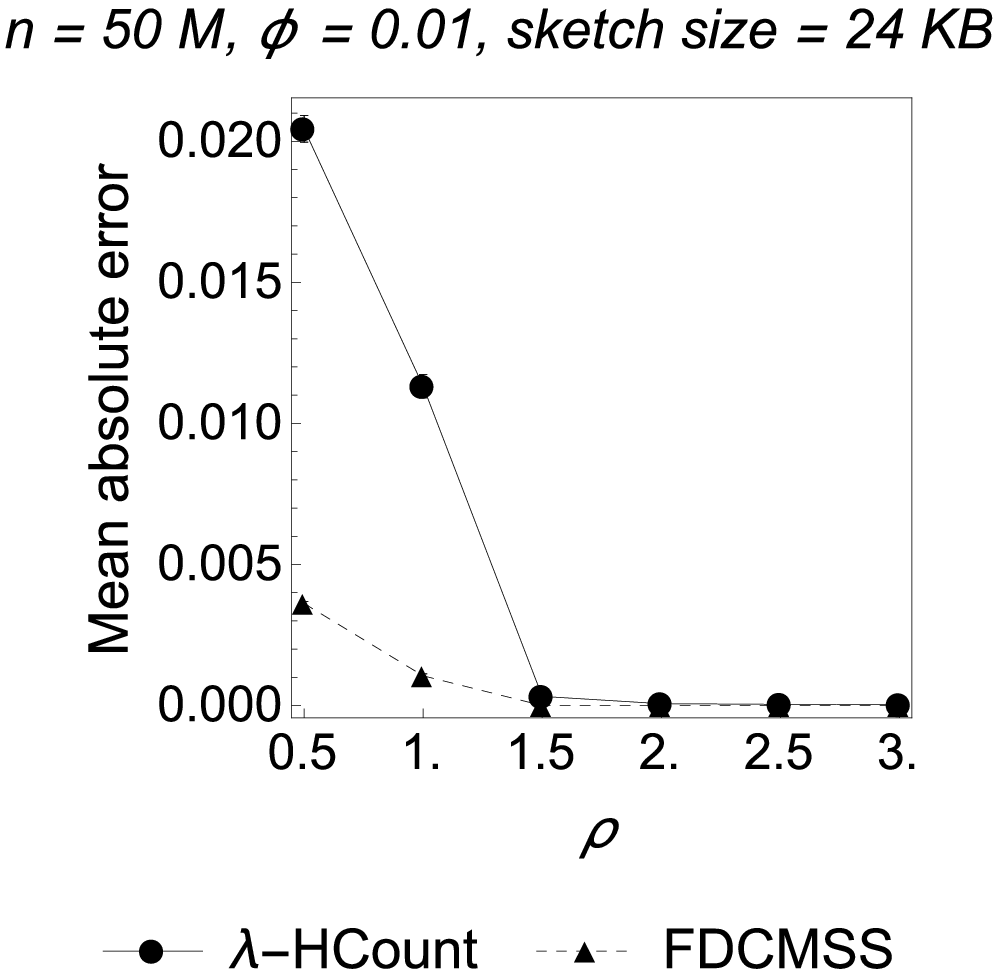}
           \label{sk-mae}
        } &
        
      \subfloat[varying the sketch size]{
           \includegraphics[scale=0.36]{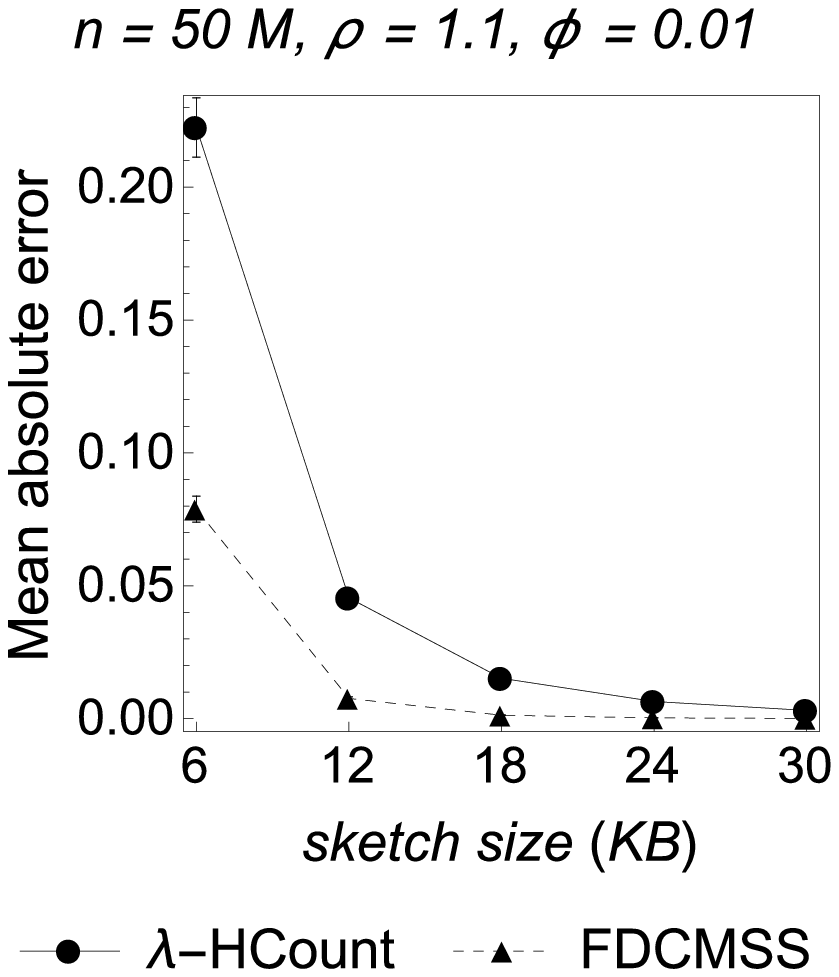}
           \label{sw-mae}
        } 

\end{tabular}
 
 \caption{Mean absolute error (mean and confidence interval)} 
 \label{mae}
\end{figure}

\begin{figure}[hbt]
  \centering
  \begin{tabular}{cccc}
     \subfloat[varying $n$]{
           \includegraphics[scale=0.36]{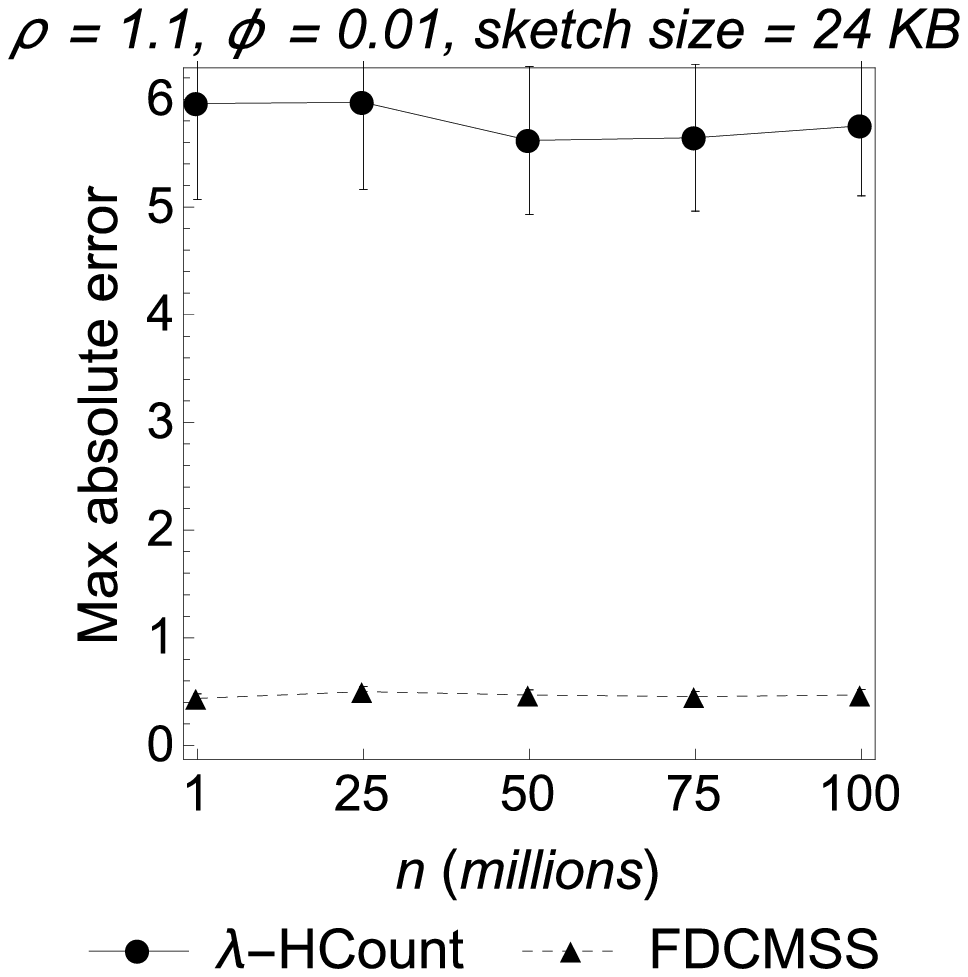}
           \label{ni-maxae}
        } &
        
      \subfloat[varying $\phi$]{
           \includegraphics[scale=0.36]{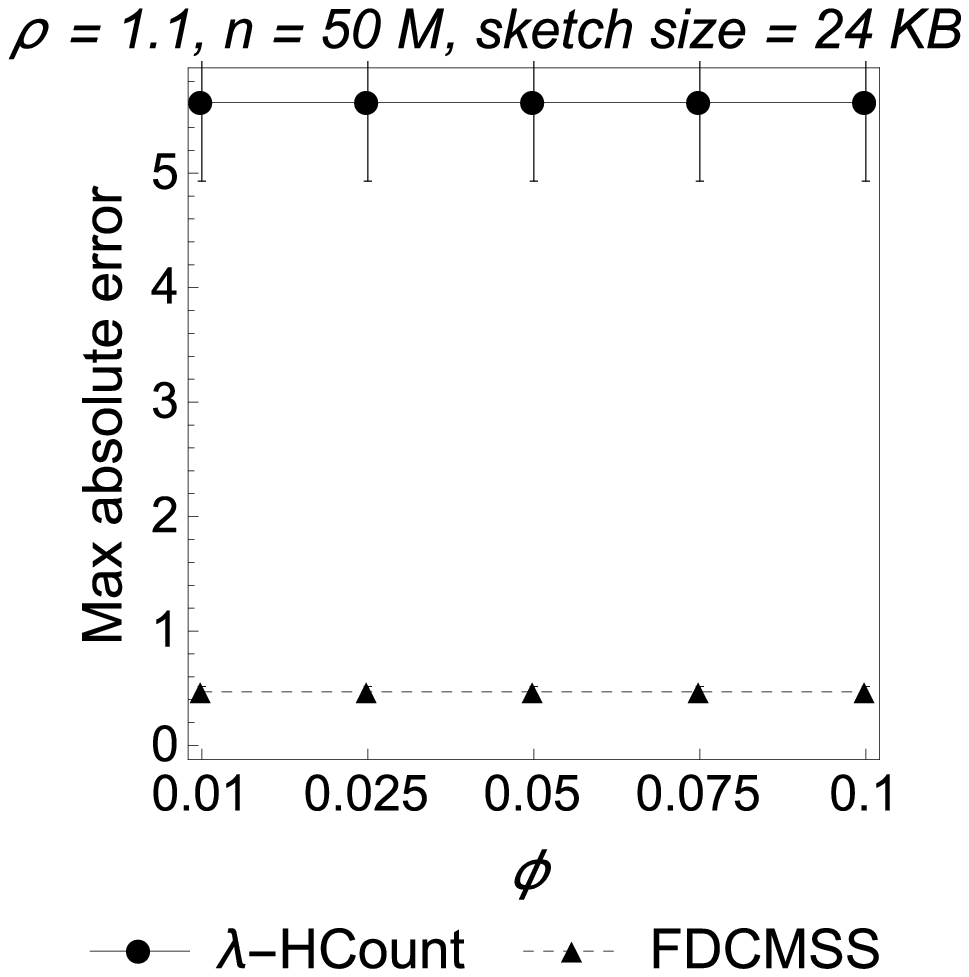}
           \label{phi-maxae}
        } &

      \subfloat[varying $\rho$]{
           \includegraphics[scale=0.36]{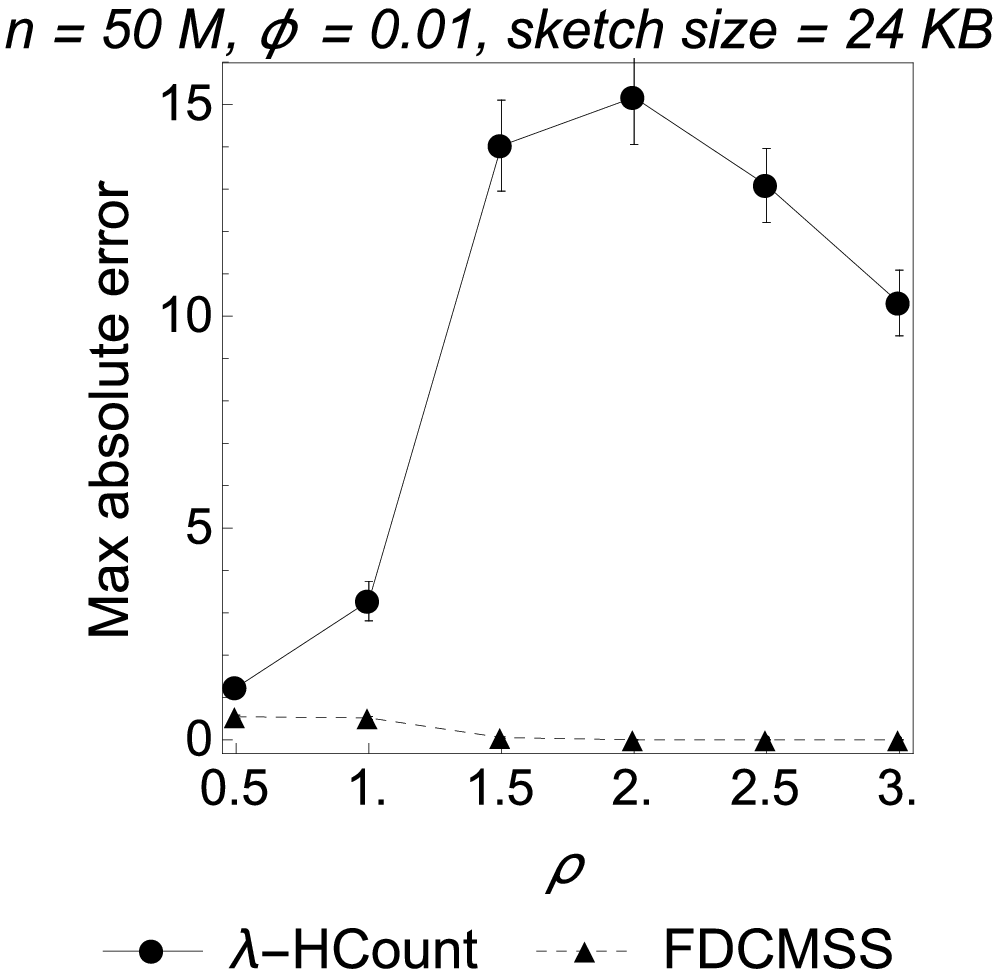}
           \label{sk-maxae}
        } &
        
      \subfloat[varying the sketch size]{
           \includegraphics[scale=0.36]{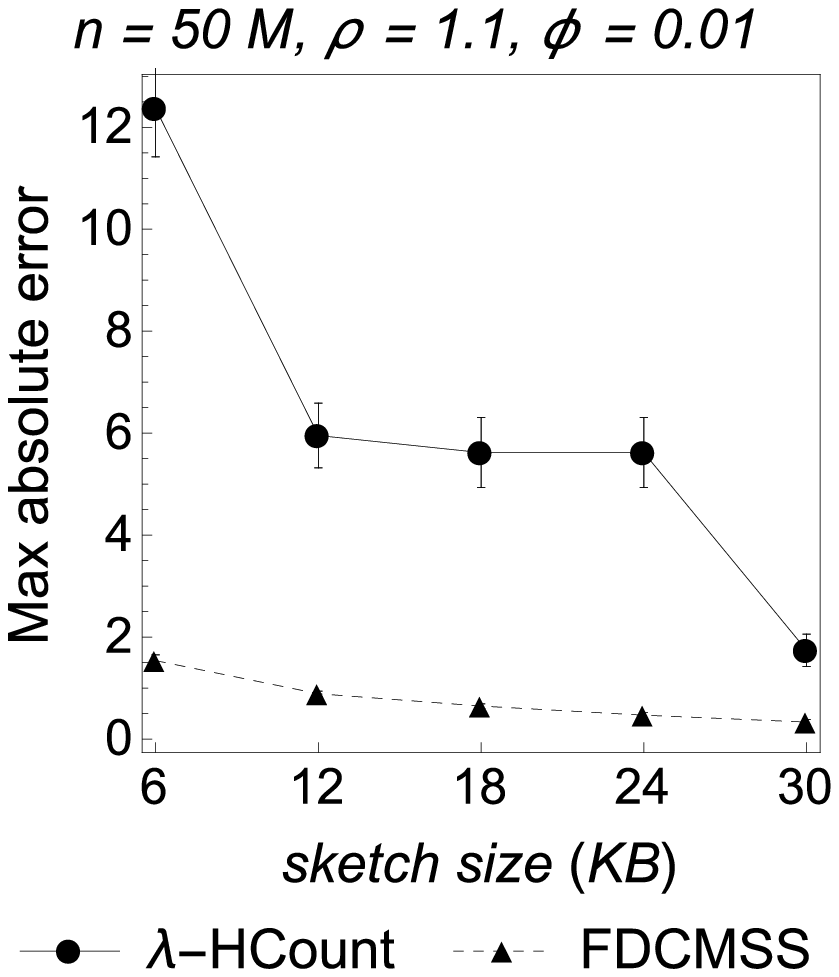}
           \label{sw-maxae}
        } 

\end{tabular}
 
 \caption{Max absolute error (mean and confidence interval)} 
 \label{maxae}
\end{figure}

\begin{figure}[hbt]
  \centering
  \begin{tabular}{cccc}
     \subfloat[varying $n$]{
           \includegraphics[scale=0.36]{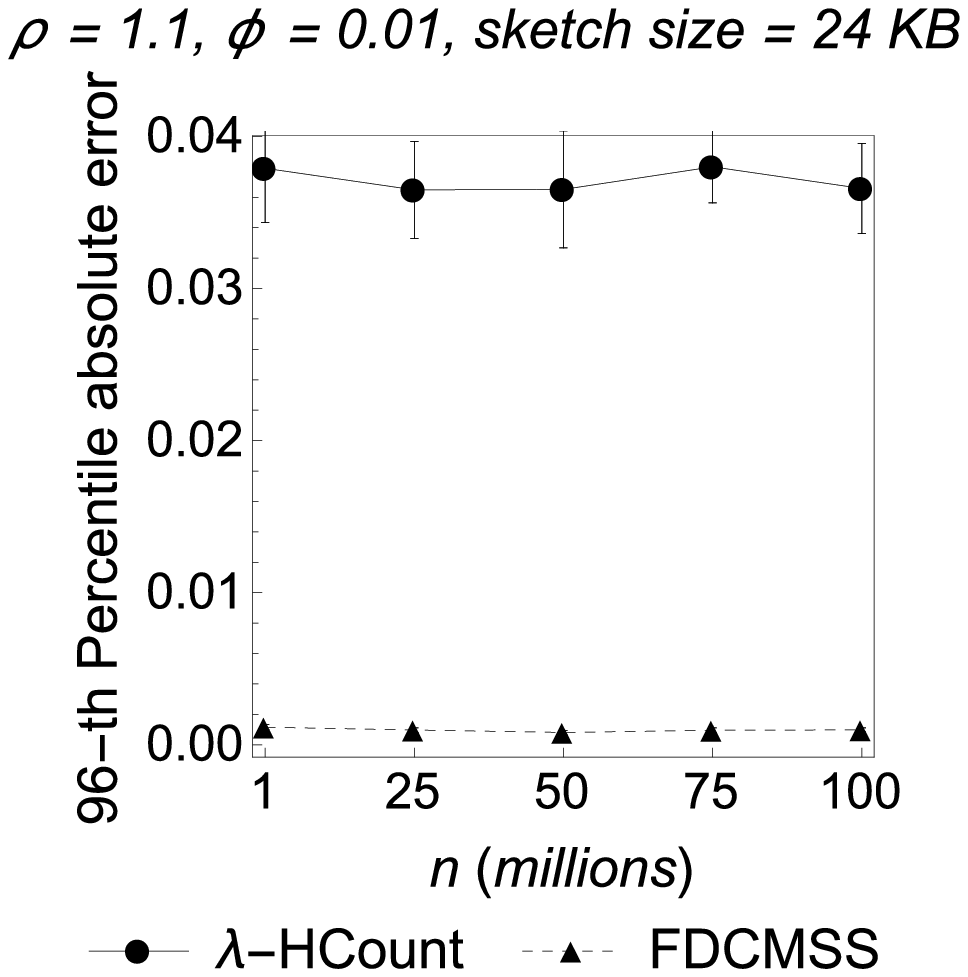}
           \label{ni-percentile}
        } &
        
      \subfloat[varying $\phi$]{
           \includegraphics[scale=0.36]{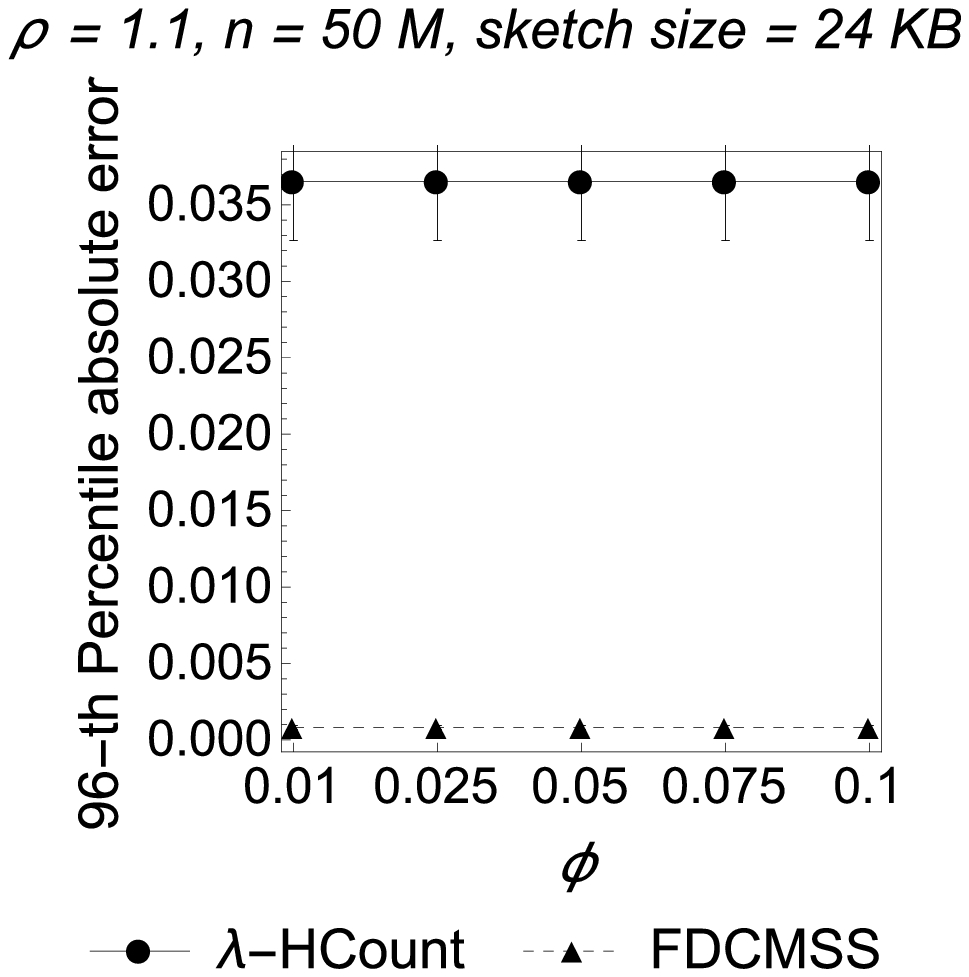}
           \label{phi-percentile}
        } &

      \subfloat[varying $\rho$]{
           \includegraphics[scale=0.36]{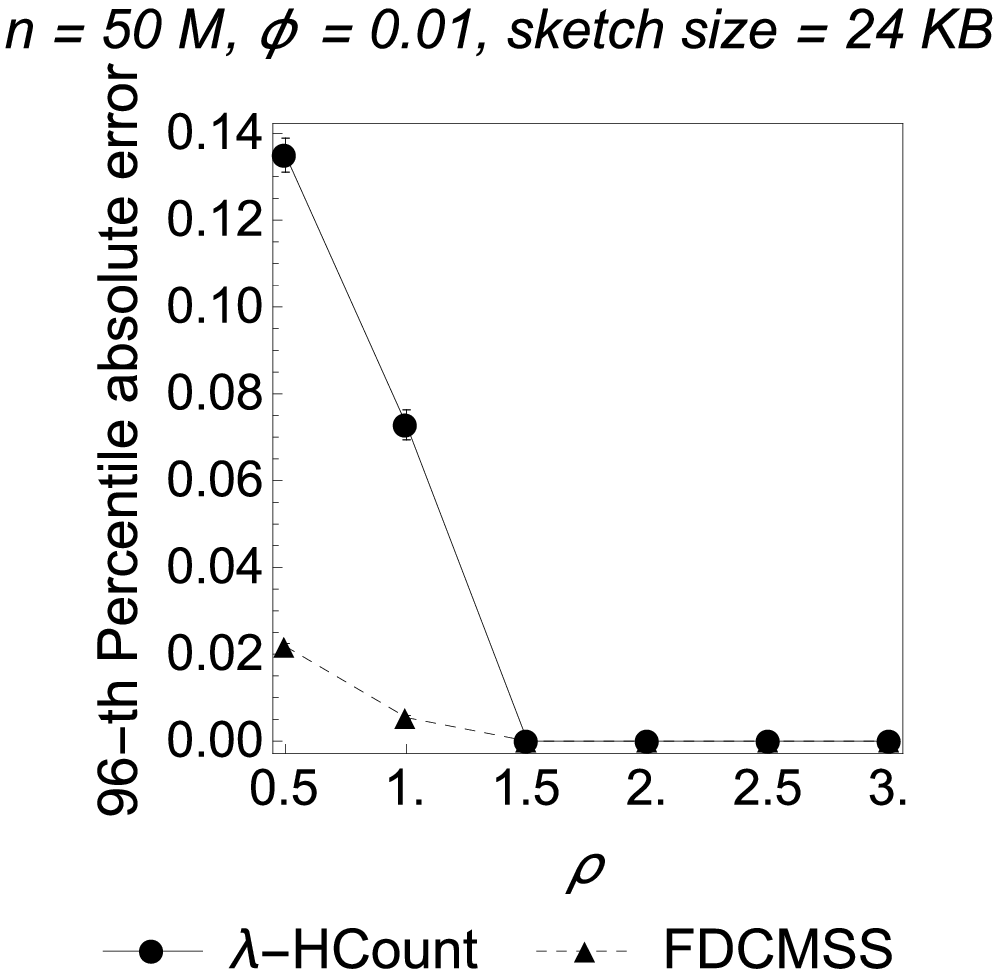}
           \label{sk-percentile}
        } &
        
      \subfloat[varying the sketch size]{
           \includegraphics[scale=0.36]{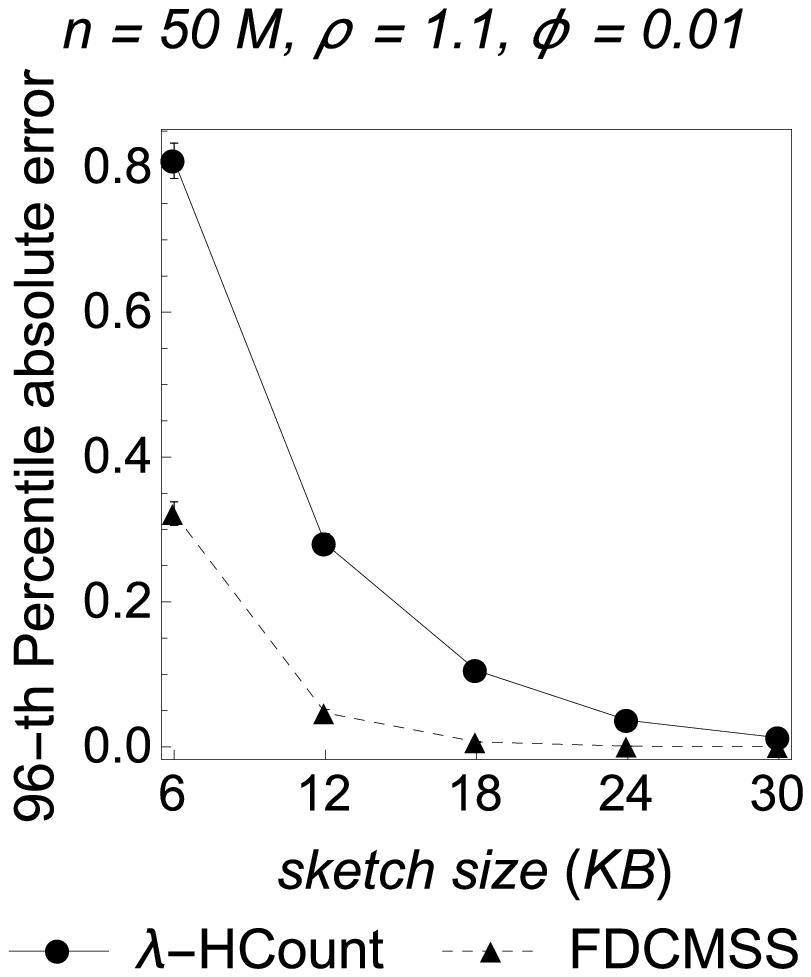}
           \label{sw-percentile}
        } 

\end{tabular}
 
 \caption{96-th Percentile absolute error (mean and confidence interval)} 
 \label{percentile}
\end{figure}

\begin{figure}[hbt]
  \centering
  \begin{tabular}{cccc}
     \subfloat[varying $n$]{
           \includegraphics[scale=0.36]{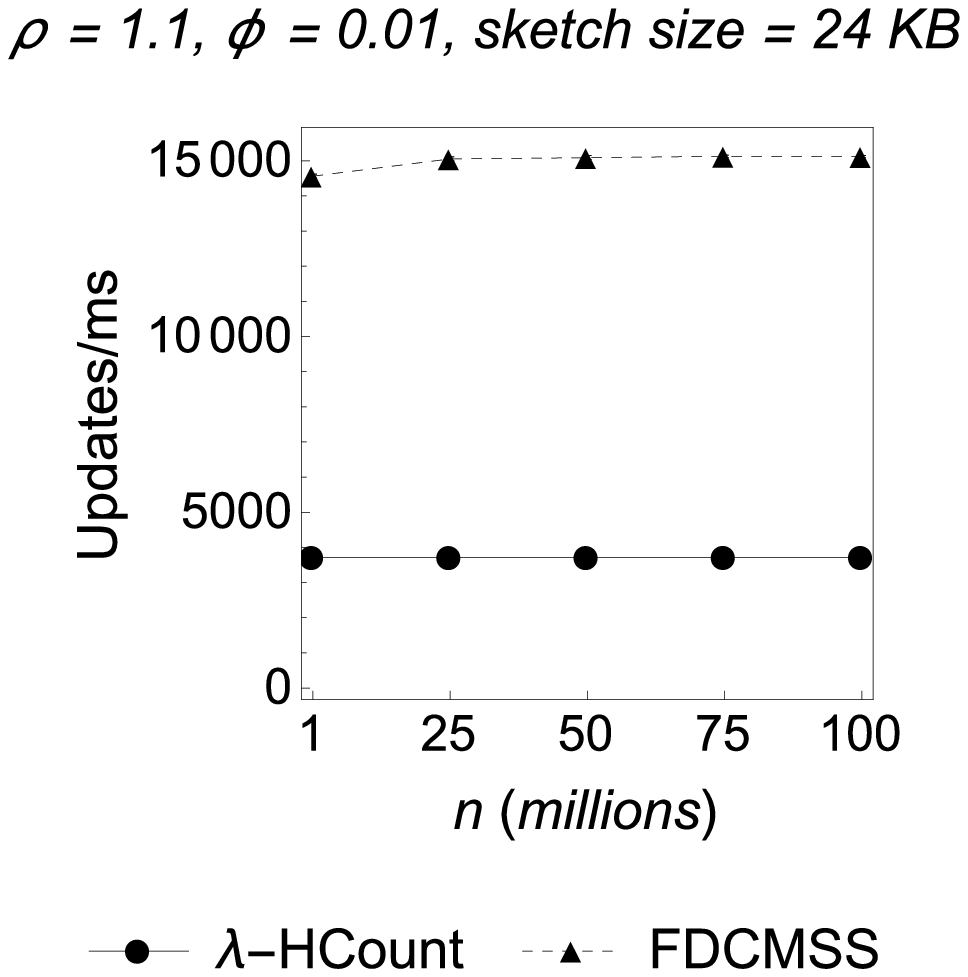}
           \label{ni-updates}
        } &
        
      \subfloat[varying $\phi$]{
           \includegraphics[scale=0.36]{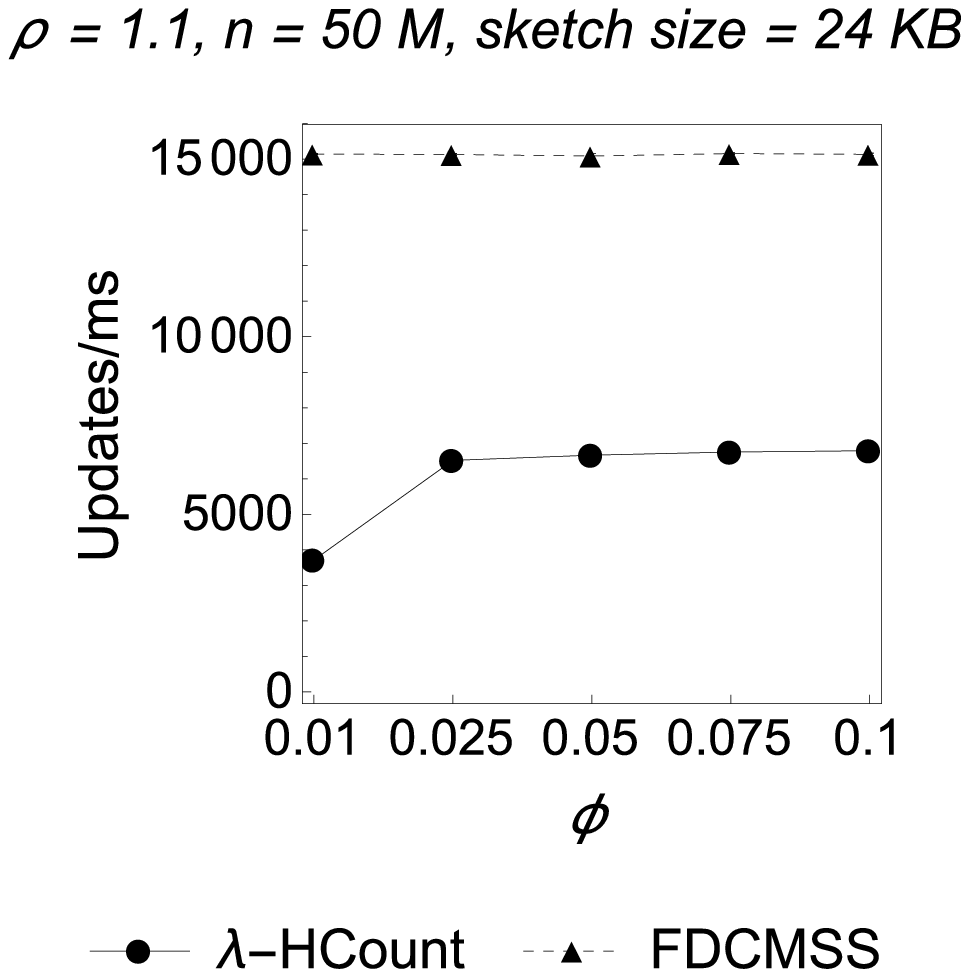}
           \label{phi-updates}
        } &

      \subfloat[varying $\rho$]{
           \includegraphics[scale=0.36]{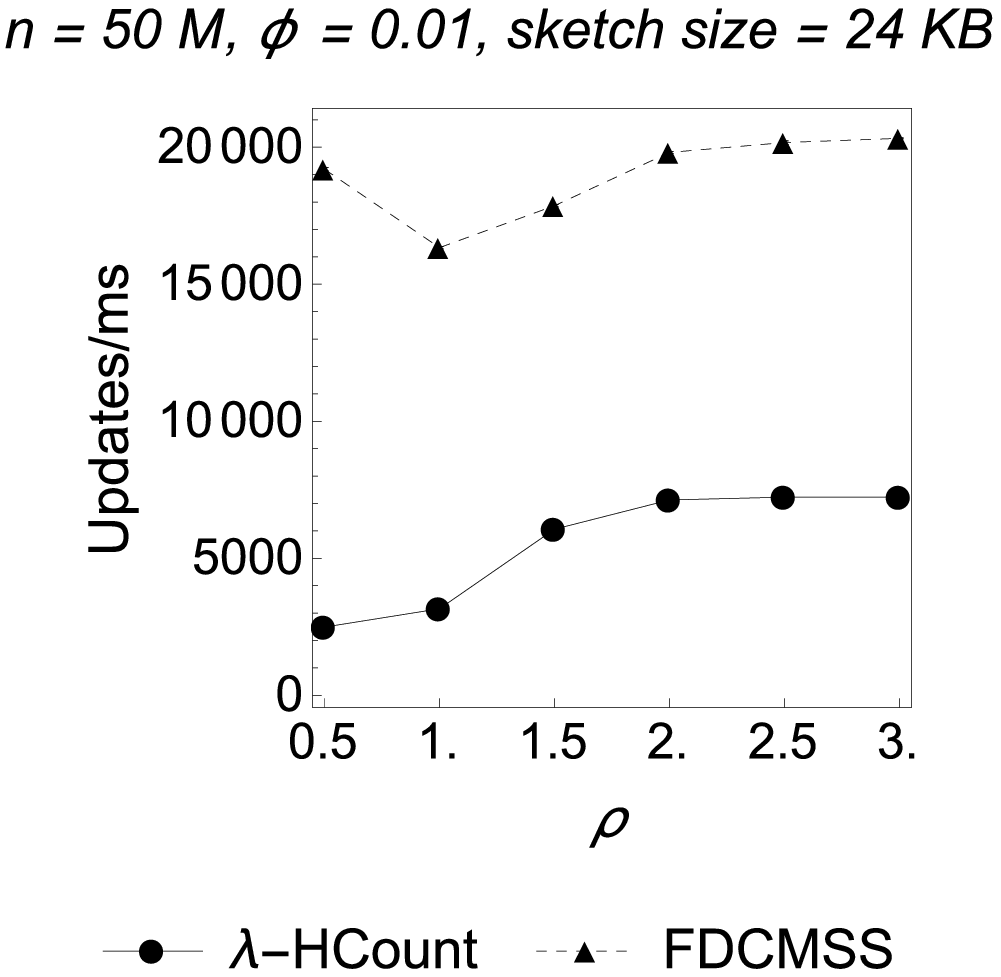}
           \label{sk-updates}
        } &
        
      \subfloat[varying the sketch size]{
           \includegraphics[scale=0.36]{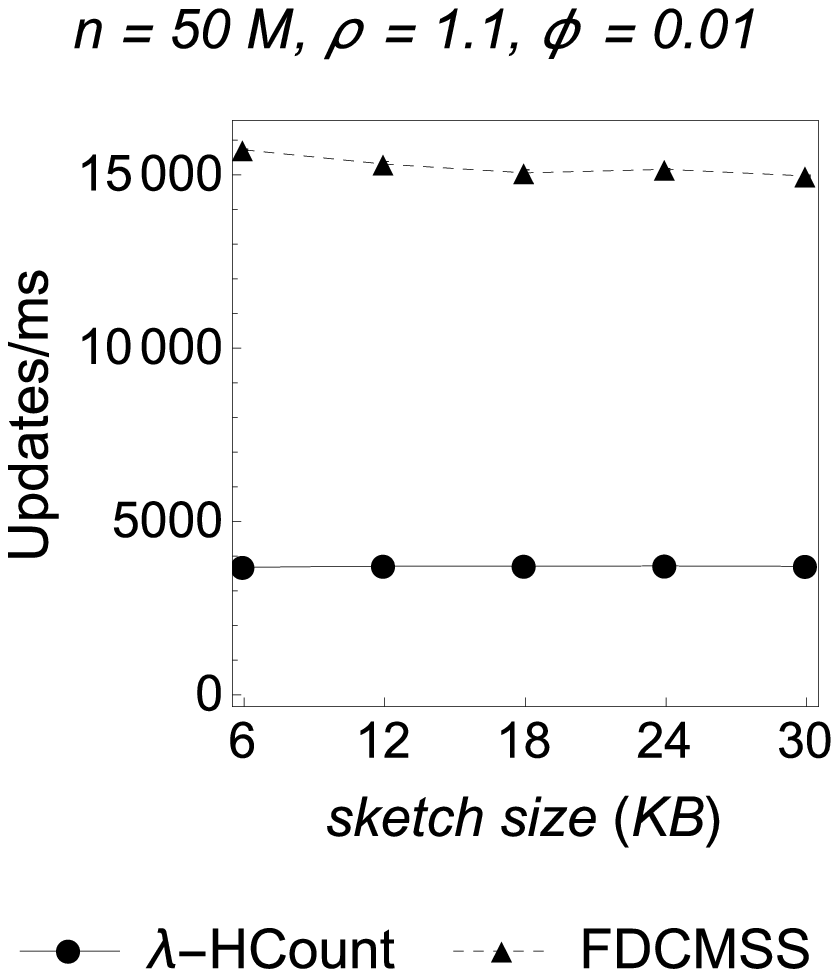}
           \label{sw-updates}
        } 

\end{tabular}
 
 \caption{Updates/ms (mean and confidence interval)} 
 \label{updates}
\end{figure}

%\begin{figure}[hbt]
%  \centering
%  \begin{tabular}{cc}
%     
%        
%      \subfloat[varying $\phi$]{
%           \includegraphics[scale=0.36]{phi-total-space.eps}
%           \label{phi-space}
%        } &
%        
%      \subfloat[varying the sketch size]{
%           \includegraphics[scale=0.36]{sw-total-space.eps}
%           \label{sw-space}
%        } 
%
%\end{tabular}
% 
% \caption{Total Space (KB)} 
% \label{space}
%\end{figure}

\clearpage

\subsection{Real datasets}

We also provide experimental results for real datasets from several domains \cite{FIMDR} \cite{Dallachiesa}. These datasets are public domain and commonly utilized for data mining experiments. Moreover, a wide variety of statistical features, reported in Table \ref{data}, characterize the datasets, described below.

\textbf{Kosarak:} This is a click-stream dataset of a Hungarian online news portal. It has been anonymized, and consists of transactions, each of which is comprised of several integer items. In the experiments, we have considered every single item in serial order.

\textbf{Retail:} This dataset contains retail market basket data coming from an anonymous Belgian store. Again, we consider all of the items belonging to the dataset in serial order.

\textbf{Q148:} Derived from the KDD Cup 2000 data, compliments of Blue Martini, this dataset contains several data. The ones we use for our experiments are the values of the attribute "Request Processing Time Sum" (attribute number 148), coming from the "clicks" dataset. A pre-processing step was required, in order to obtain the final dataset. We had to replace all of the missing values (appearing as question marks) with the value of 0.

\textbf{Webdocs:} This dataset derives from a spidered collection of web html documents. The whole collection contains about 1.7 millions documents, mainly written in English, and its size is about 5 GB. The resulting dataset, after preliminary filtering and pre-processing, has a size of about 1.48 GB.

\begin{table}
\renewcommand{\arraystretch}{1.3}
 \caption{Statistical characteristics of the real datasets}
      \label{data}
	\centering
    \begin{tabular}{| c |  c |  c  | c | c |}
    \hline
      & Kosarak & Retail & Q148 &  Webdocs \\ \hline
      \hline
    Count &  8019015 & 908576 & 234954 & 299887139 \\ \hline
    Distinct items & 41270 & 16470 & 11824 &  5267656 \\ \hline
    Min & 1 & 0 & 0 &  1 \\ \hline 
    Max & 41270 & 16469 & 149464496 &  5267656 \\ \hline 
    Mean & 2387.2 & 3264.7 & 3392.9 &  122715  \\ \hline
    Median & 640 & 1564 & 63 & 1988 \\ \hline
    Std. deviation & 4308.5 & 4093.2 & 309782.5 & 549736 \\ \hline
    Skewness & 3.5 & 1.5 & 478.1 &  6.1\\ \hline
    \end{tabular}
    \end{table}
    
The experiments have been carried out by varying respectively $\phi$ and the sketch size. As we did for synthetic datasets, for each plot we always compare the algorithms by using exactly the same sketch size in kilobytes. We are not taking into account the additional space required by $\lambda$-HCount for its $F$ data structure.

Figures \ref{kosarak-precision-recall}, \ref{retail-precision-recall}, \ref{q148-precision-recall} and \ref{webdocs-precision-recall} show precision and recall for the datasets under examination. In all of the experiments, FDCMSS and  $\lambda$-HCount achieve 100\% recall. Regarding precision, FDCMSS outperforms $\lambda$-HCount or provides the same precision. Indeed, both algorithms achieve 100\% precision when varying $\phi$, but FDCMSS outperforms $\lambda$-HCount when varying the sketch size.

\begin{figure}[hbt]
\centering
\begin{tabular}{cccc}
             
      \subfloat[varying $\phi$]{
           \includegraphics[scale=0.36]{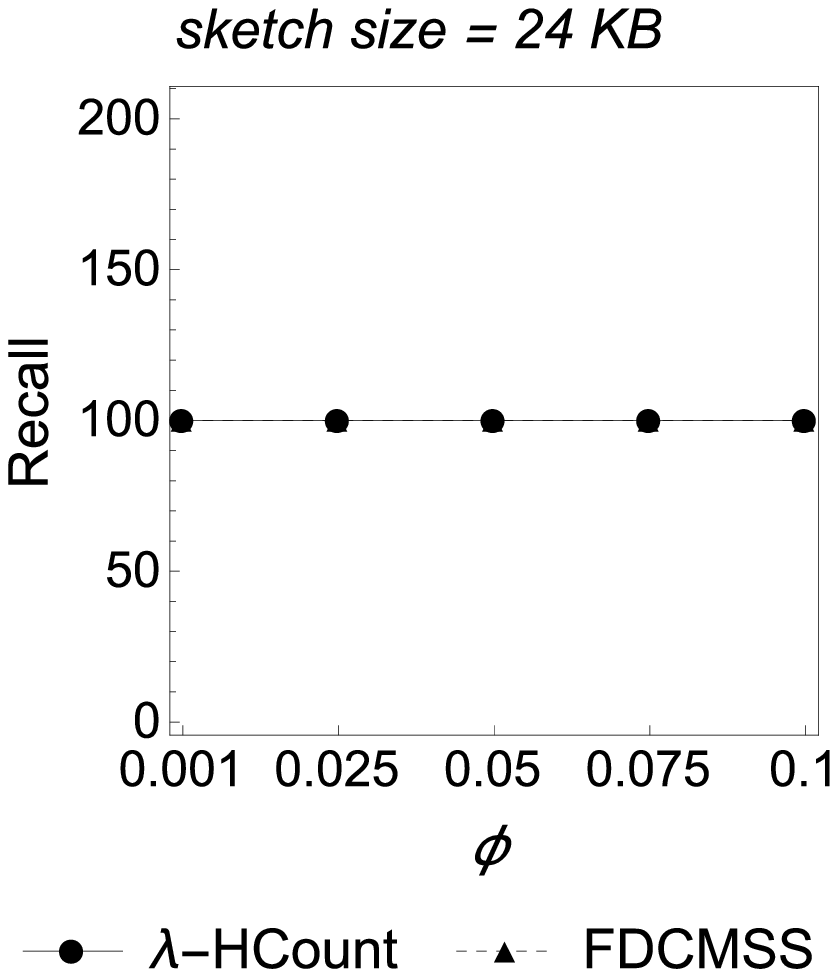}
           \label{phi-kosarak-recall}
        } &

      \subfloat[varying the sketch size]{
           \includegraphics[scale=0.36]{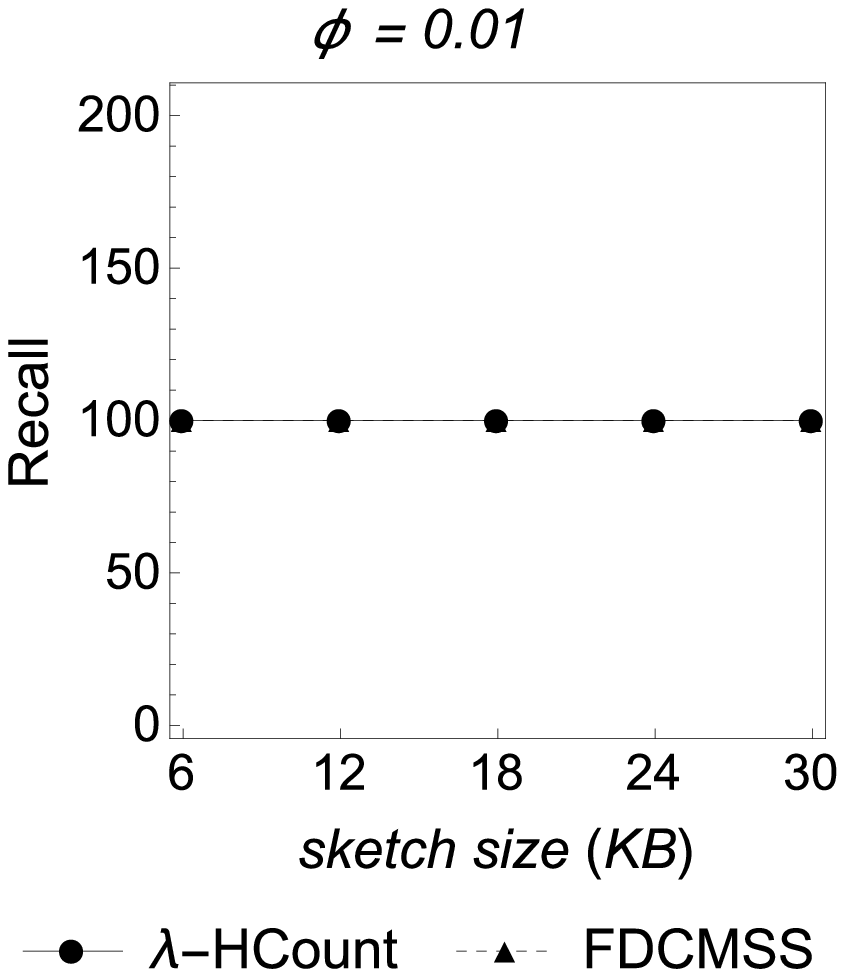}
           \label{sw-kosarak-recall}
        } &
        
     \subfloat[varying $\phi$]{
           \includegraphics[scale=0.36]{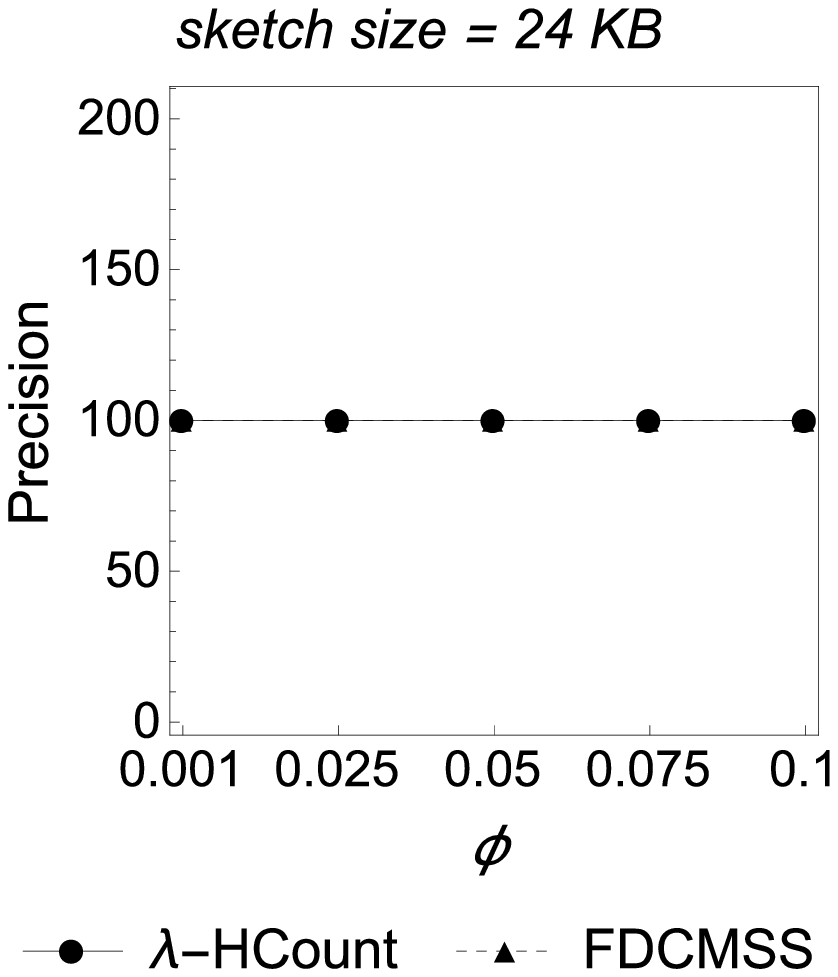}
           \label{phi-kosarak-prec}
        } &

      \subfloat[varying the sketch size]{
           \includegraphics[scale=0.36]{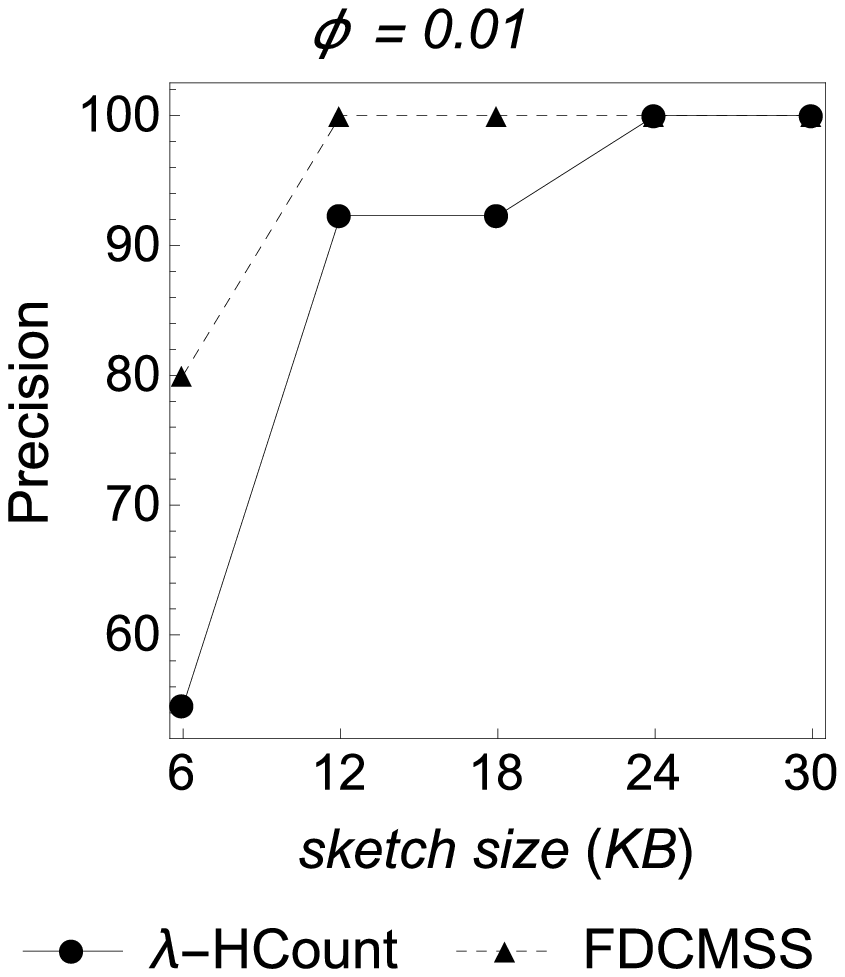}
           \label{sw-kosarak-prec}
        }

\end{tabular}
 
\caption{Kosarak: recall and precision (mean and confidence interval)} 
\label{kosarak-precision-recall}
\end{figure}

\begin{figure}[hbt]
\centering
\begin{tabular}{cccc}
             
      \subfloat[varying $\phi$]{
           \includegraphics[scale=0.36]{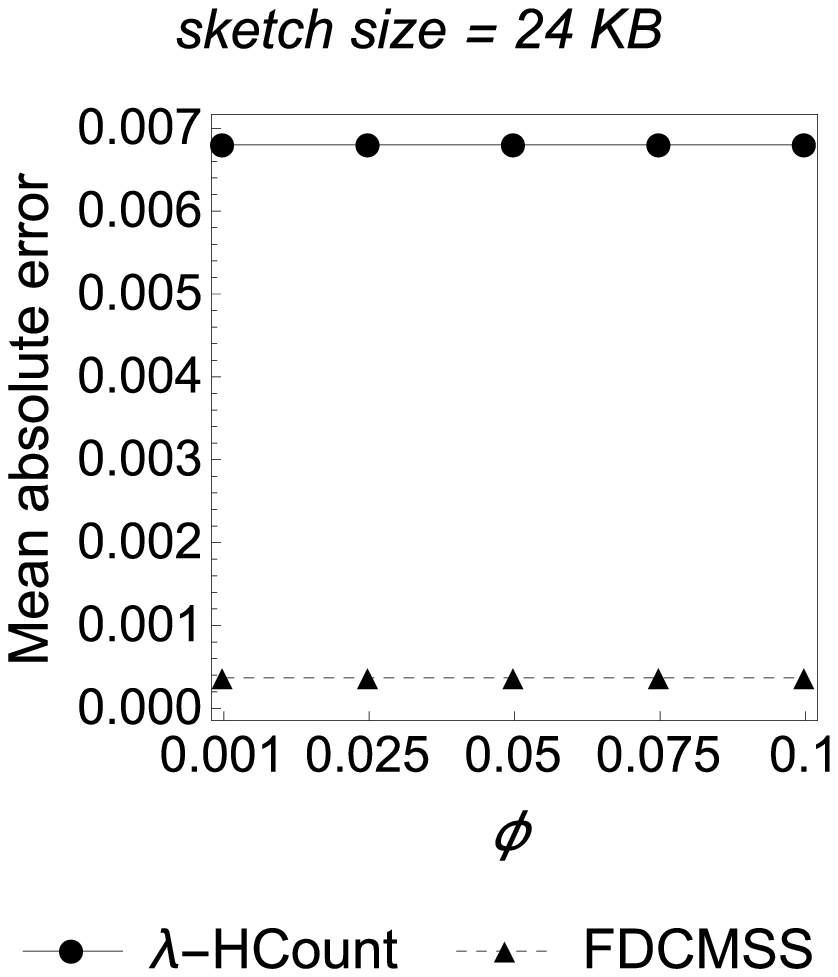}
           \label{phi-kosarak-mean-absolute-error}
        } &

      \subfloat[varying the sketch size]{
           \includegraphics[scale=0.36]{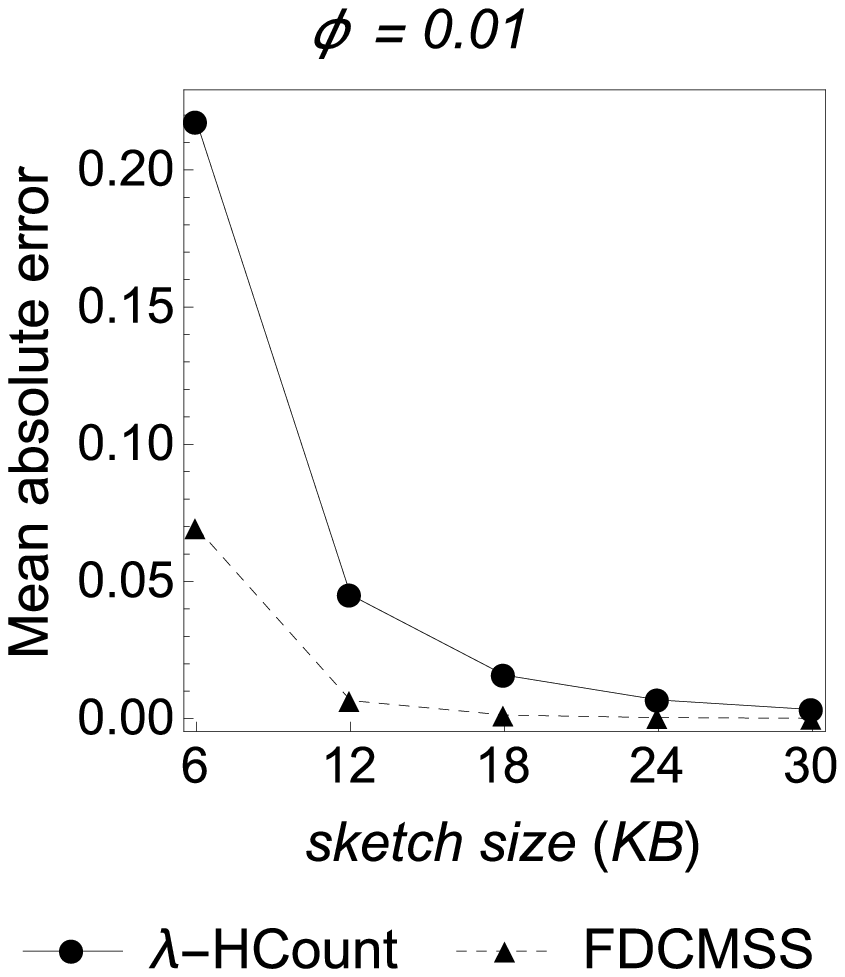}
           \label{sw-kosarak-mean-absolute-error}
        } &
        
     \subfloat[varying $\phi$]{
           \includegraphics[scale=0.36]{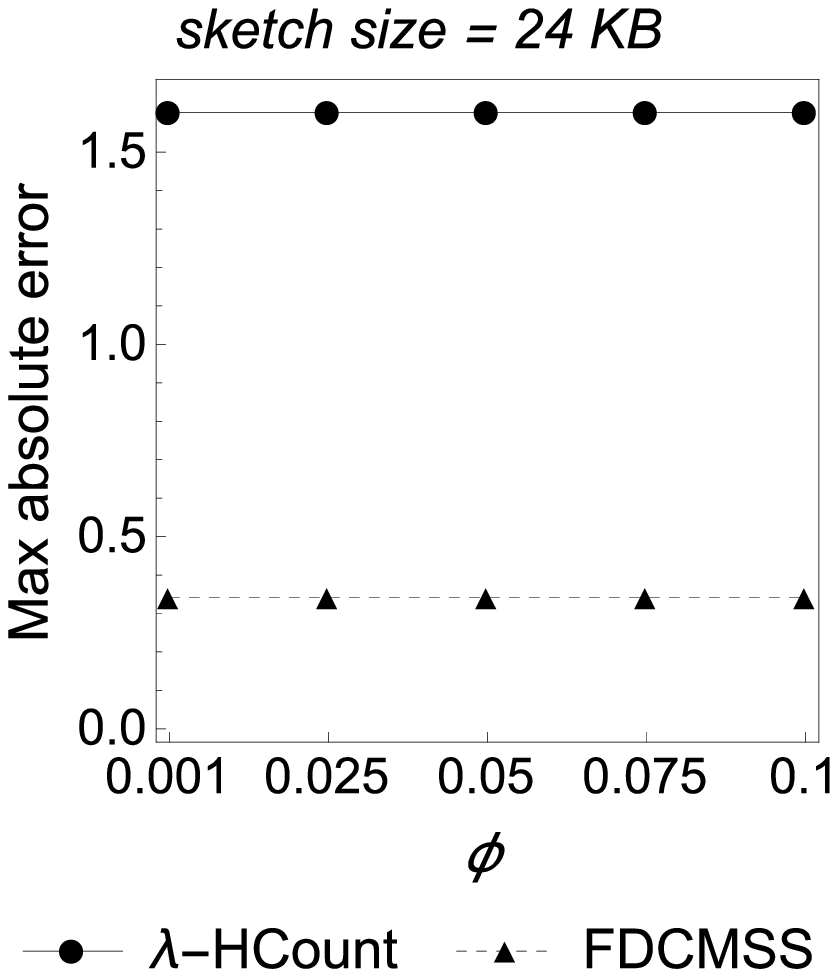}
           \label{phi-kosarak-max-absolute-error}
        } &

      \subfloat[varying the sketch size]{
           \includegraphics[scale=0.36]{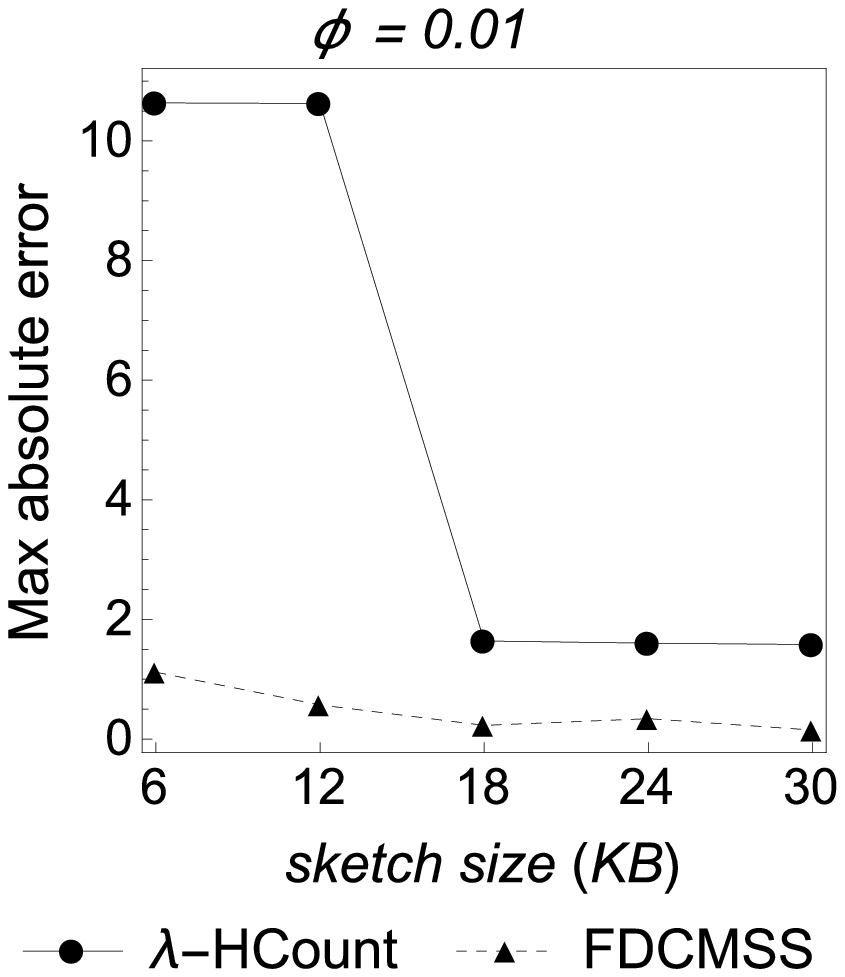}
           \label{sw-kosarak-max-absolute-error}
        }

\end{tabular}
 
\caption{Kosarak: Mean and max absolute error (mean and confidence interval)} 
\label{kosarak-mean-max-absolute-error}
\end{figure}

\begin{figure}[hbt]
\centering
\begin{tabular}{cccc}
             
      \subfloat[varying $\phi$]{
           \includegraphics[scale=0.36]{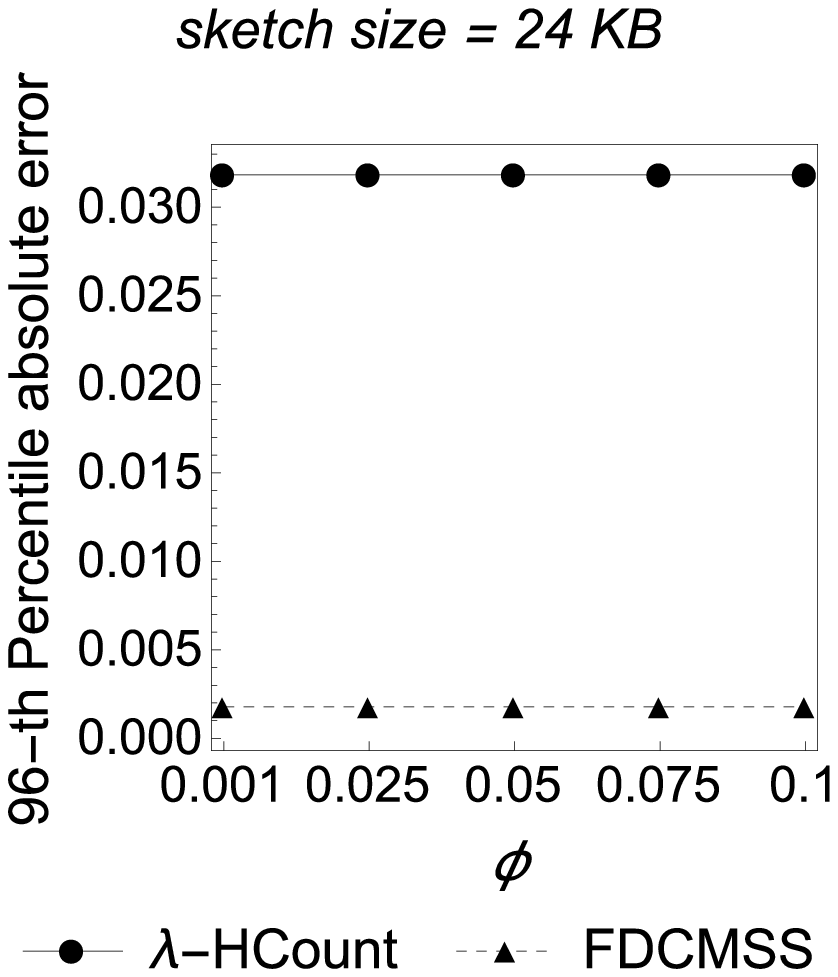}
           \label{phi-kosarak-percentile-absolute-error}
        } &

      \subfloat[varying the sketch size]{
           \includegraphics[scale=0.36]{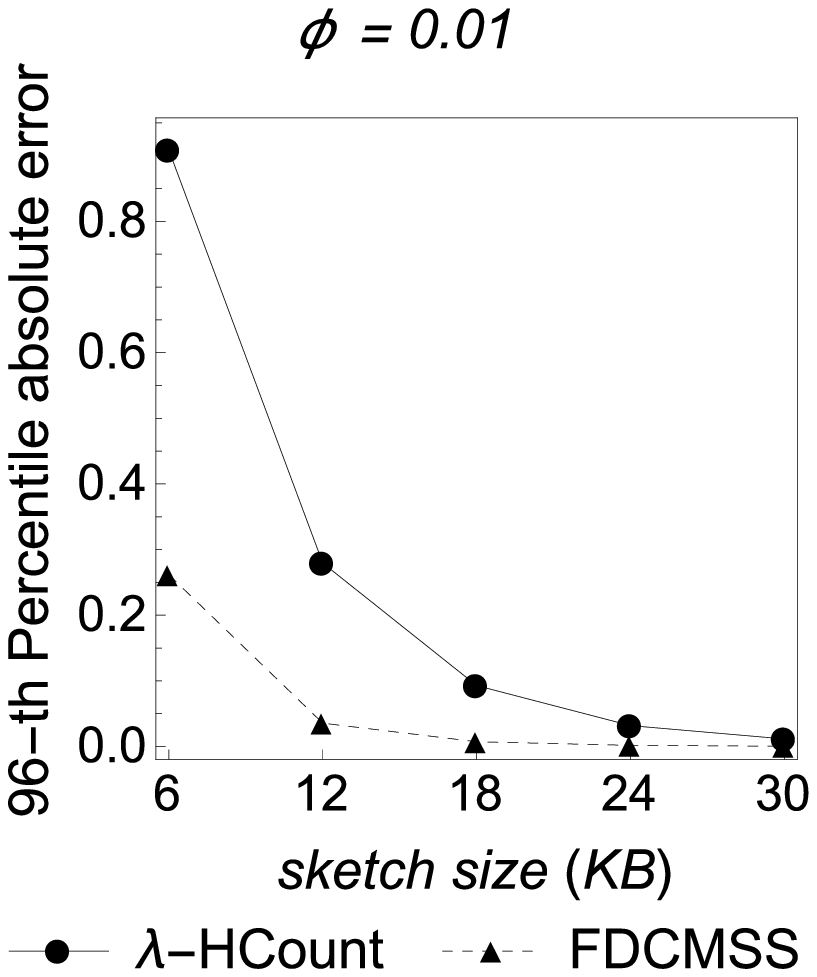}
           \label{sw-kosarak-percentile-absolute-error}
        } &
        
     \subfloat[varying $\phi$]{
           \includegraphics[scale=0.36]{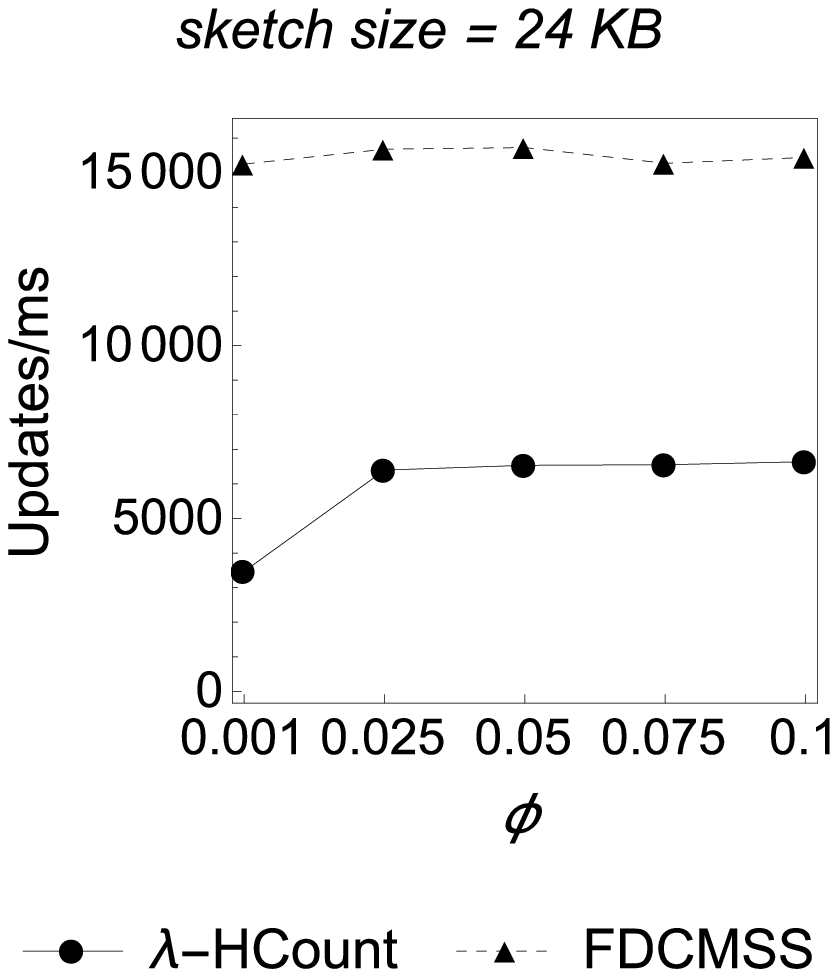}
           \label{phi-kosarak-updates}
        } &

      \subfloat[varying the sketch size]{
           \includegraphics[scale=0.36]{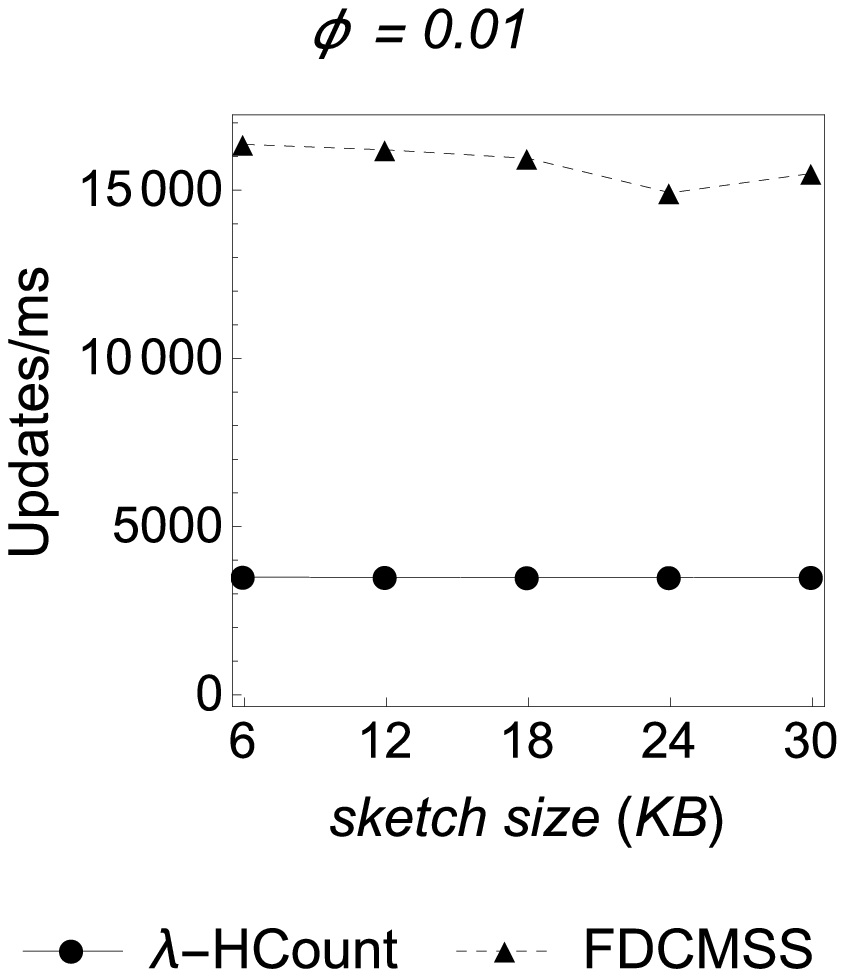}
           \label{sw-kosarak-updates}
        }

\end{tabular}
 
\caption{Kosarak: Percentile absolute error and updates/ms (mean and confidence interval)} 
\label{kosarak-percentile-absolute-error-updates}
\end{figure}

Figures \ref{kosarak-mean-max-absolute-error}, \ref{retail-mean-max-absolute-error}, \ref{q148-mean-max-absolute-error} and \ref{webdocs-mean-max-absolute-error} show mean and max absolute errors. Clearly, FDCMSS outperforms $\lambda$-HCount in all of the experiments carried out. Figures \ref{kosarak-percentile-absolute-error-updates}, \ref{retail-percentile-absolute-error-updates}, \ref{q148-percentile-absolute-error-updates} and \ref{webdocs-percentile-absolute-error-updates} show the 96-th percentile of absolute error and the number of updates per millisecond. FDCMSS outperforms $\lambda$-HCount in all of the experiments carried out for these metrics.

\begin{figure}[hbt]
\centering
\begin{tabular}{cccc}
             
      \subfloat[varying $\phi$]{
           \includegraphics[scale=0.36]{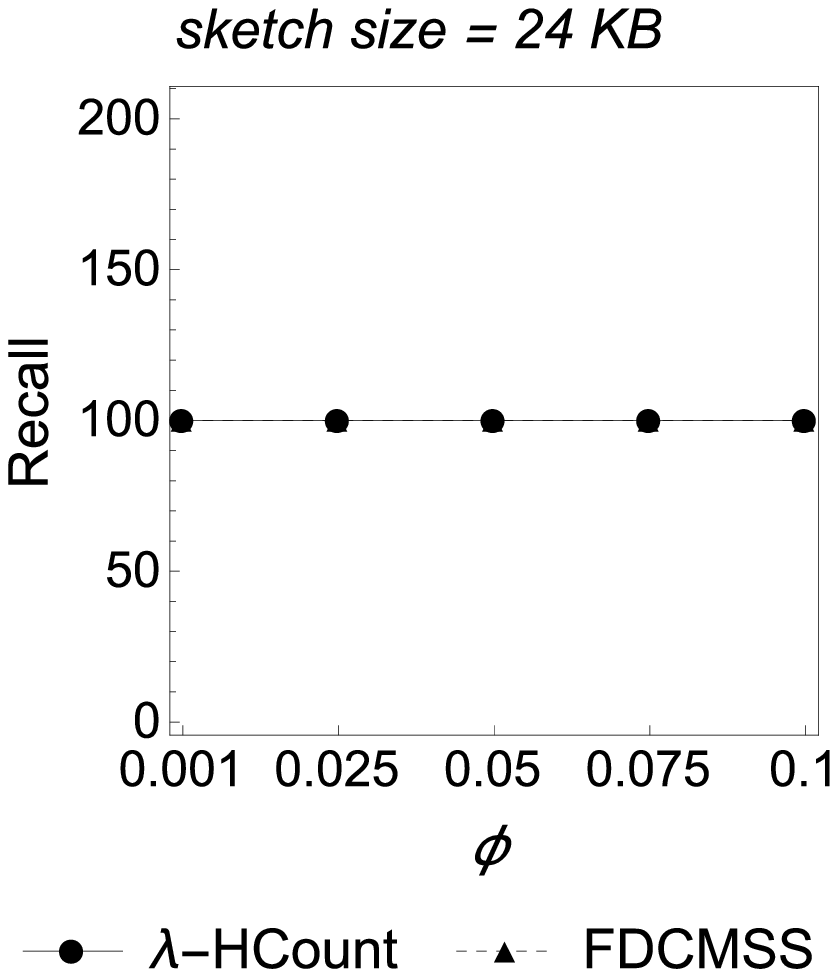}
           \label{phi-retail-recall}
        } &

      \subfloat[varying the sketch size]{
           \includegraphics[scale=0.36]{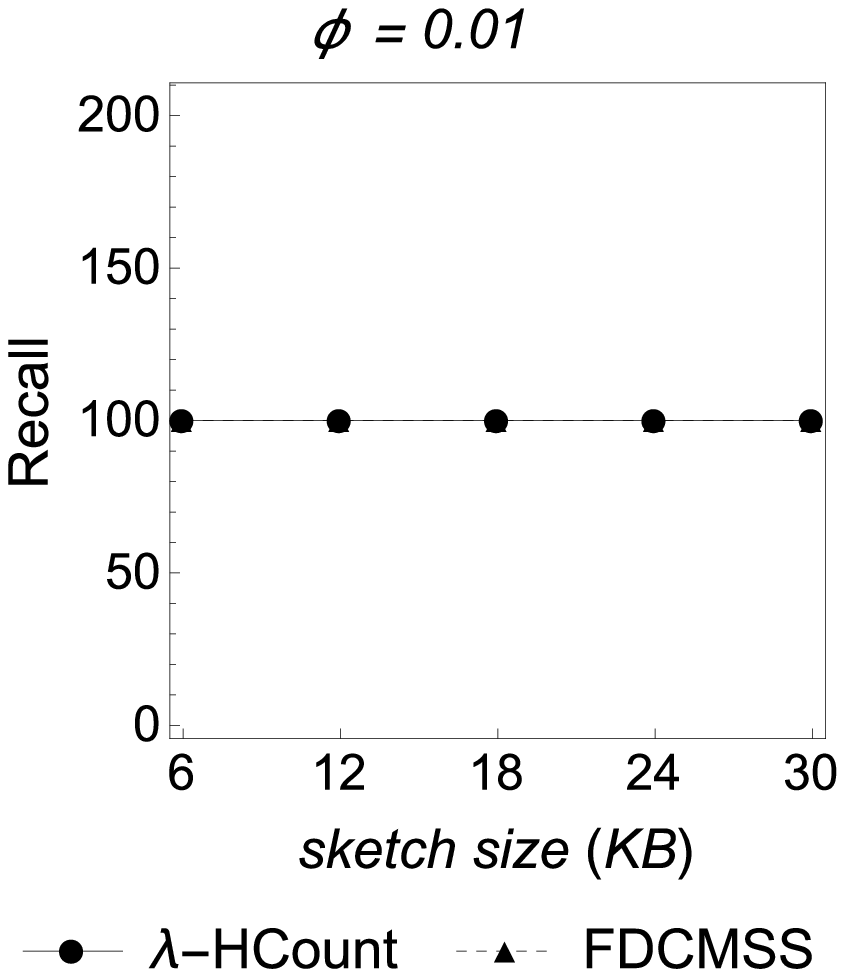}
           \label{sw-retail-recall}
        } &
        
     \subfloat[varying $\phi$]{
           \includegraphics[scale=0.36]{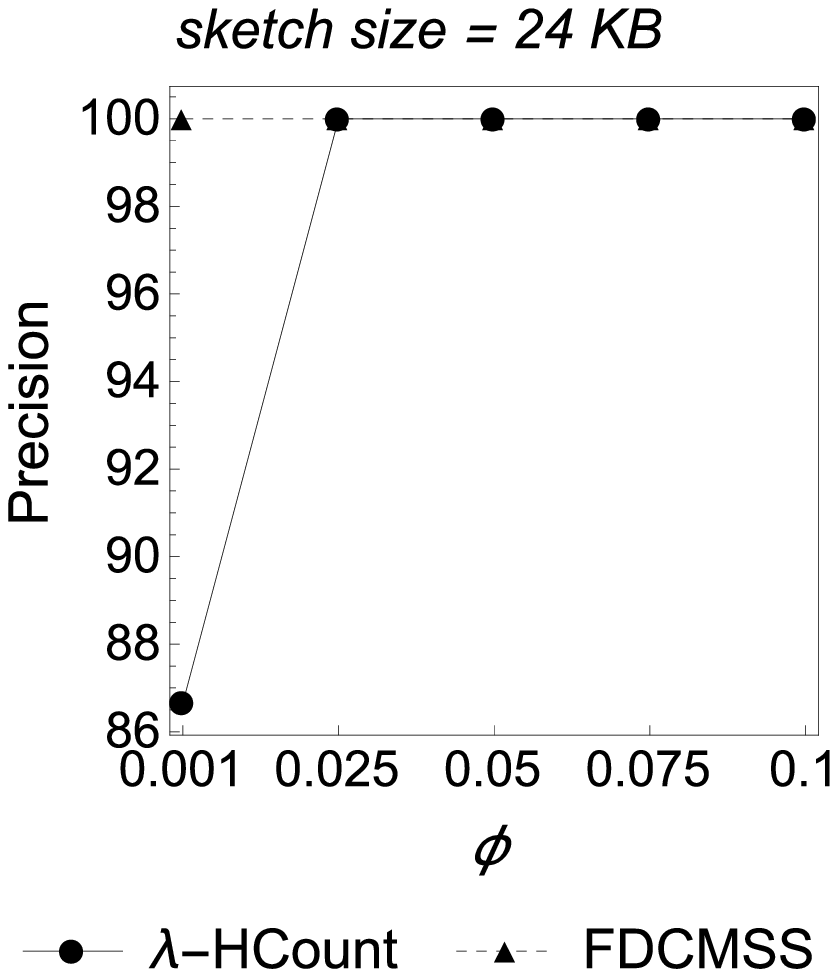}
           \label{phi-retail-prec}
        } &

      \subfloat[varying the sketch size]{
           \includegraphics[scale=0.36]{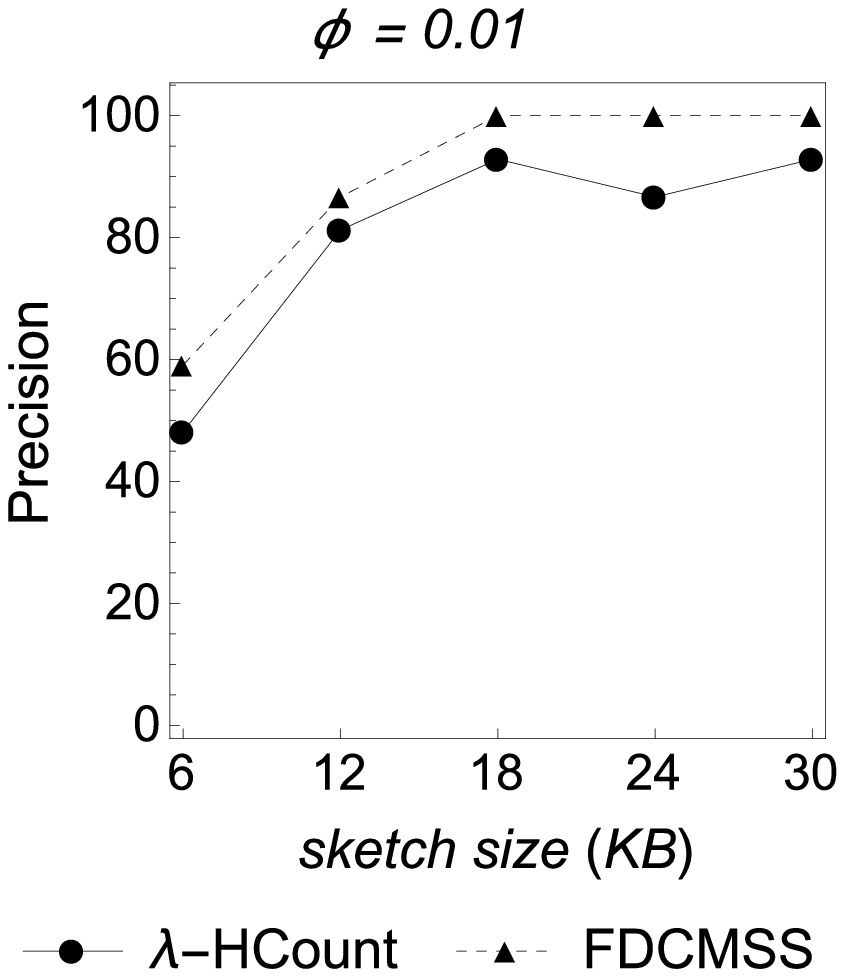}
           \label{sw-retail-prec}
        }

\end{tabular}
 
\caption{Retail: recall and precision (mean and confidence interval)} 
\label{retail-precision-recall}
\end{figure}

\begin{figure}[hbt]
\centering
\begin{tabular}{cccc}
             
      \subfloat[varying $\phi$]{
           \includegraphics[scale=0.36]{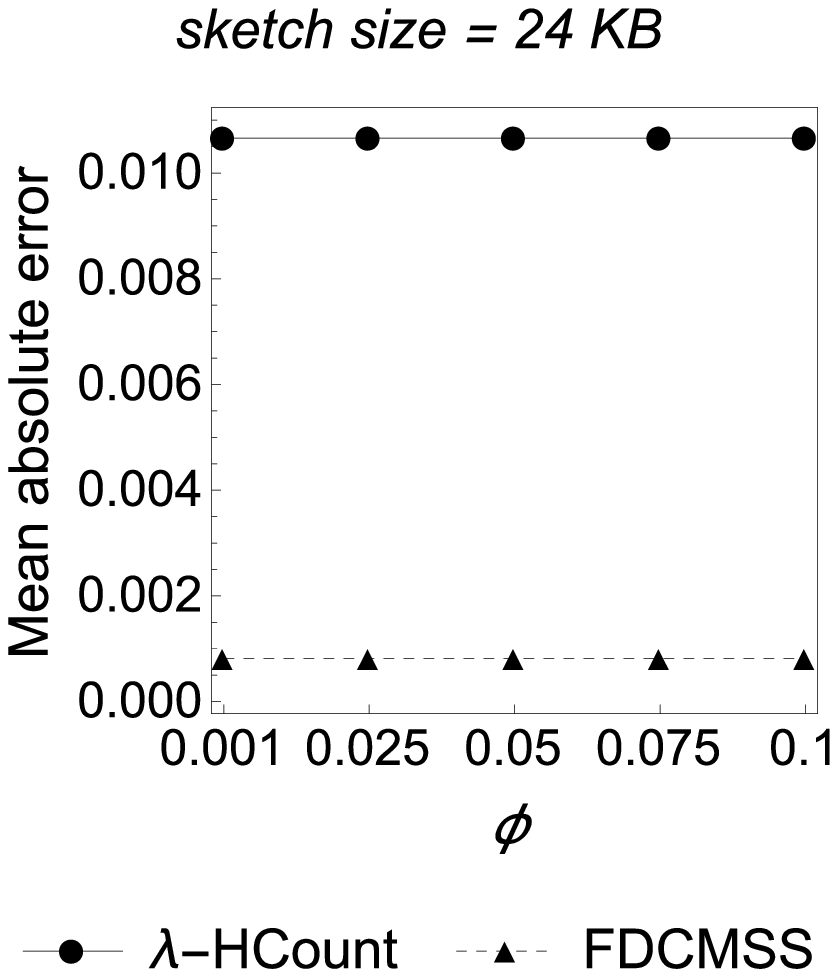}
           \label{phi-retail-mean-absolute-error}
        } &

      \subfloat[varying the sketch size]{
           \includegraphics[scale=0.36]{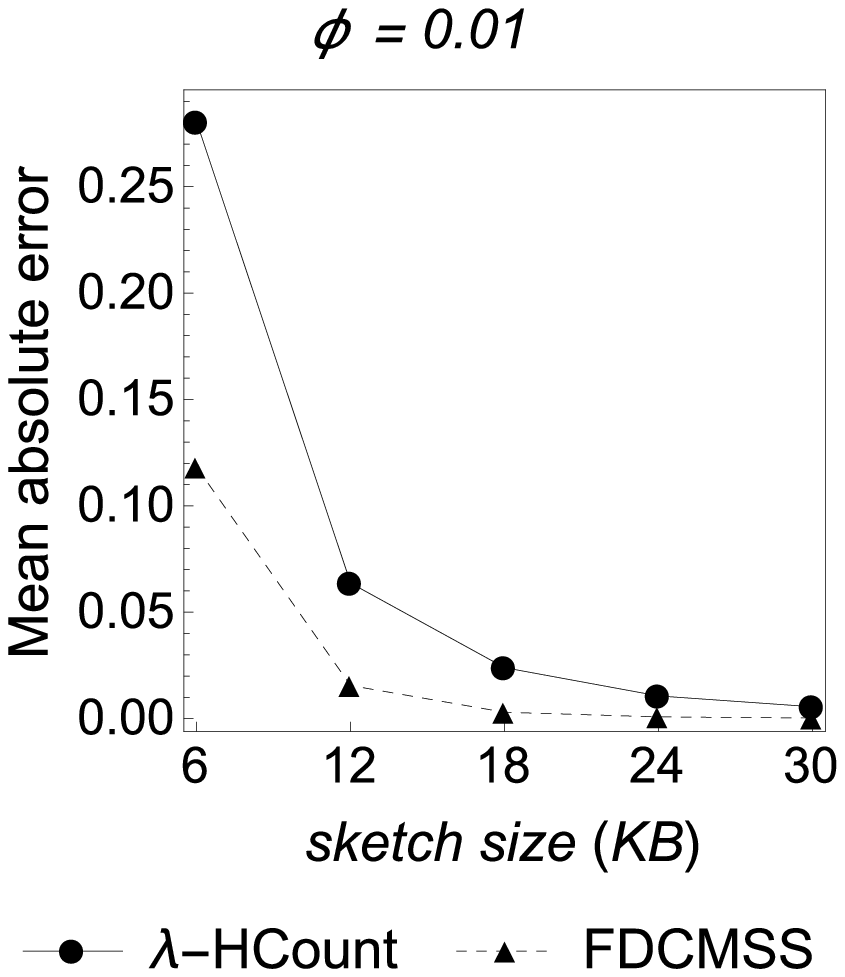}
           \label{sw-retail-mean-absolute-error}
        } &
        
     \subfloat[varying $\phi$]{
           \includegraphics[scale=0.36]{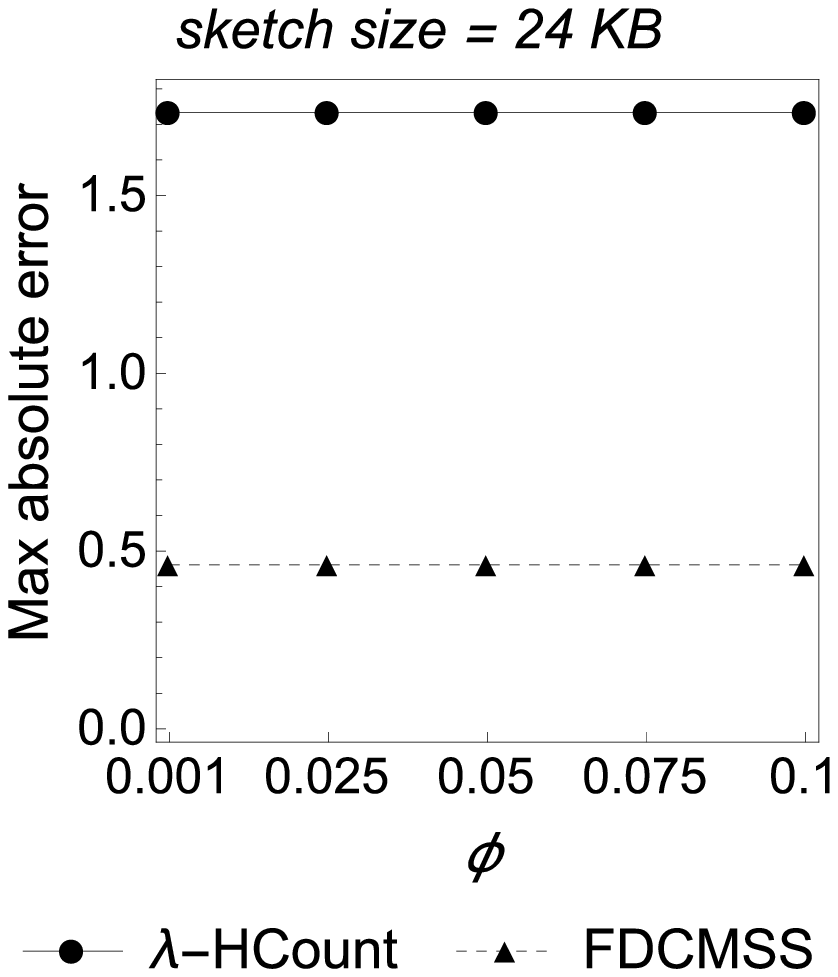}
           \label{phi-retail-max-absolute-error}
        } &

      \subfloat[varying the sketch size]{
           \includegraphics[scale=0.36]{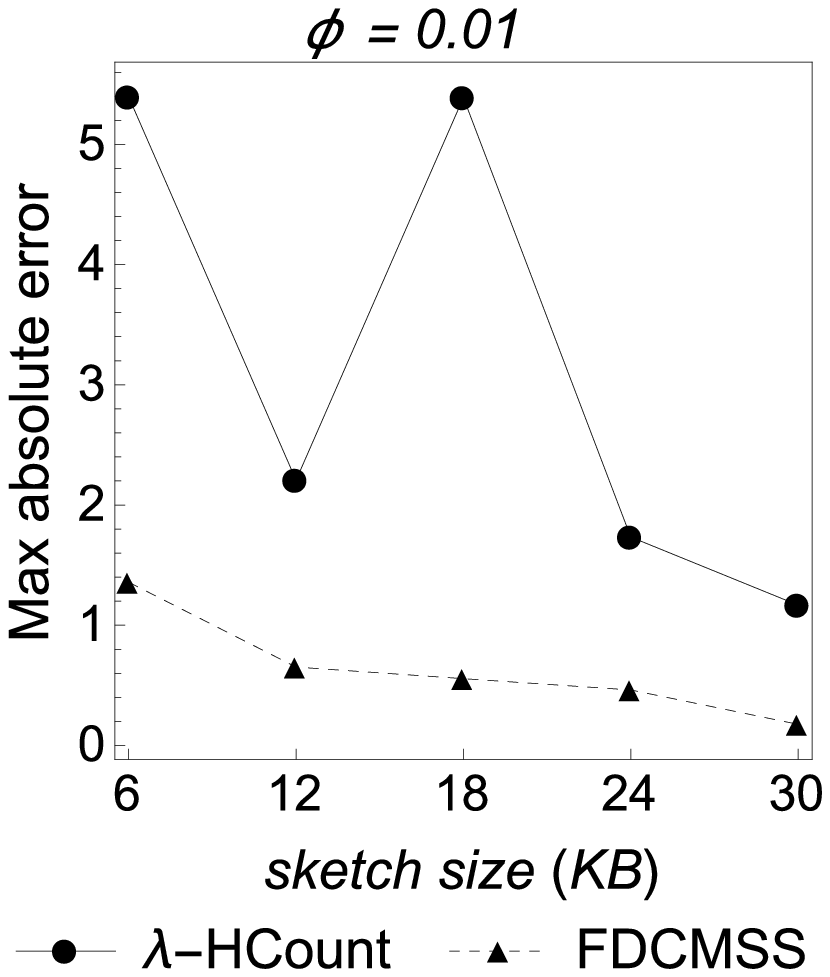}
           \label{sw-retail-max-absolute-error}
        }

\end{tabular}
 
\caption{Retail: Mean and max absolute error (mean and confidence interval)} 
\label{retail-mean-max-absolute-error}
\end{figure}

\begin{figure}[hbt]
\centering
\begin{tabular}{cccc}
             
      \subfloat[varying $\phi$]{
           \includegraphics[scale=0.36]{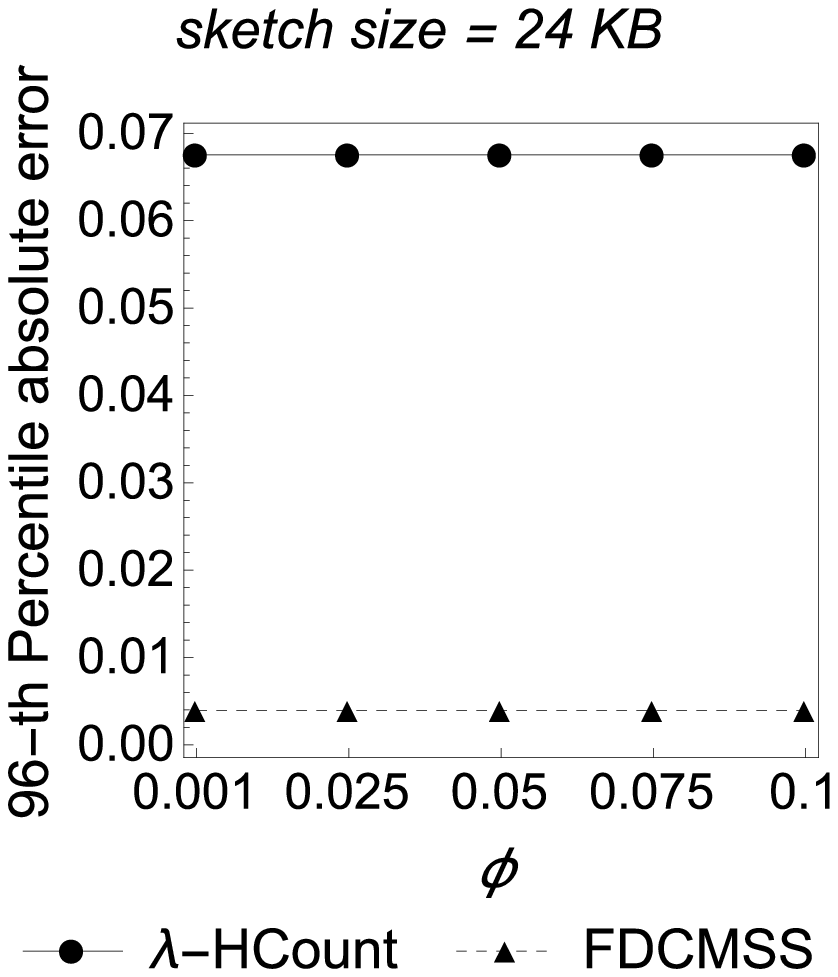}
           \label{phi-retail-percentile-absolute-error}
        } &

      \subfloat[varying the sketch size]{
           \includegraphics[scale=0.36]{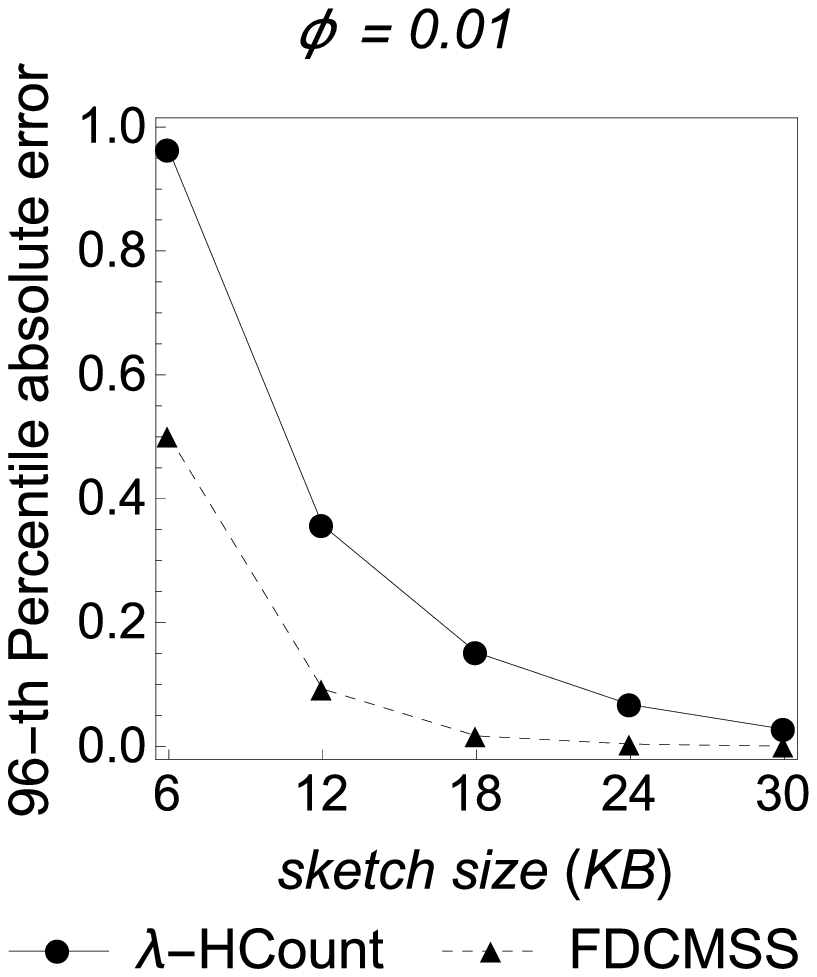}
           \label{sw-retail-percentile-absolute-error}
        } &
        
     \subfloat[varying $\phi$]{
           \includegraphics[scale=0.36]{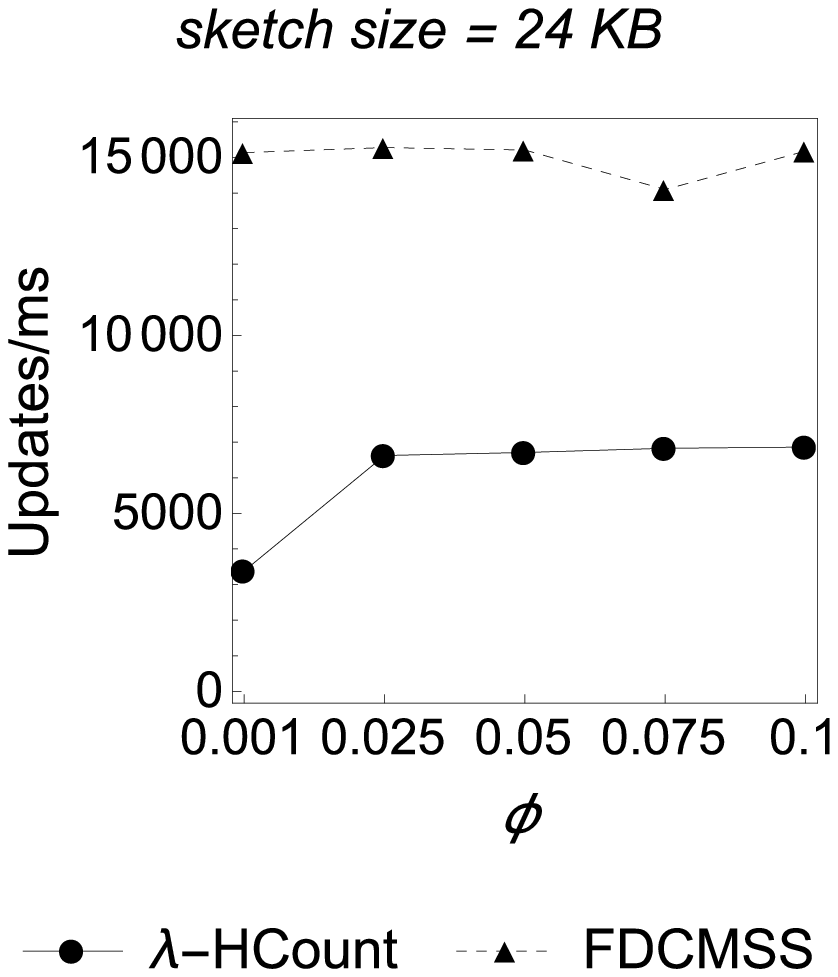}
           \label{phi-retail-updates}
        } &

      \subfloat[varying the sketch size]{
           \includegraphics[scale=0.36]{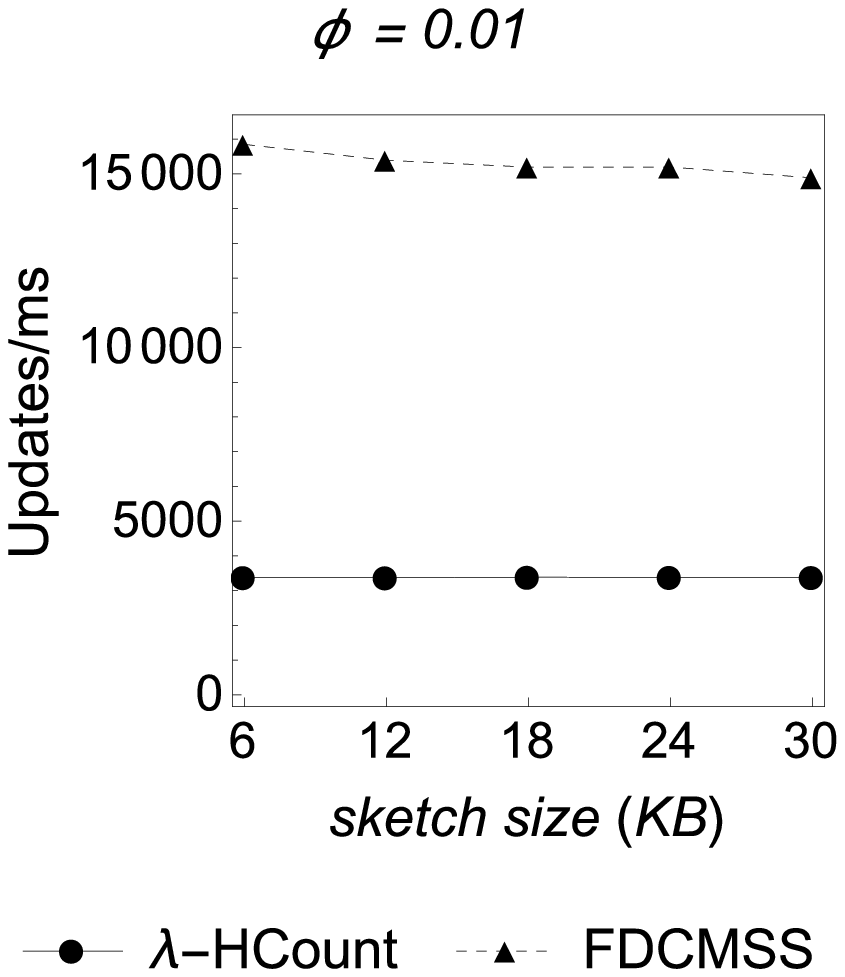}
           \label{sw-retail-updates}
        }

\end{tabular}
 
\caption{Retail: Percentile absolute error and updates/ms (mean and confidence interval)} 
\label{retail-percentile-absolute-error-updates}
\end{figure}

\begin{figure}[hbt]
\centering
\begin{tabular}{cccc}
             
      \subfloat[varying $\phi$]{
           \includegraphics[scale=0.36]{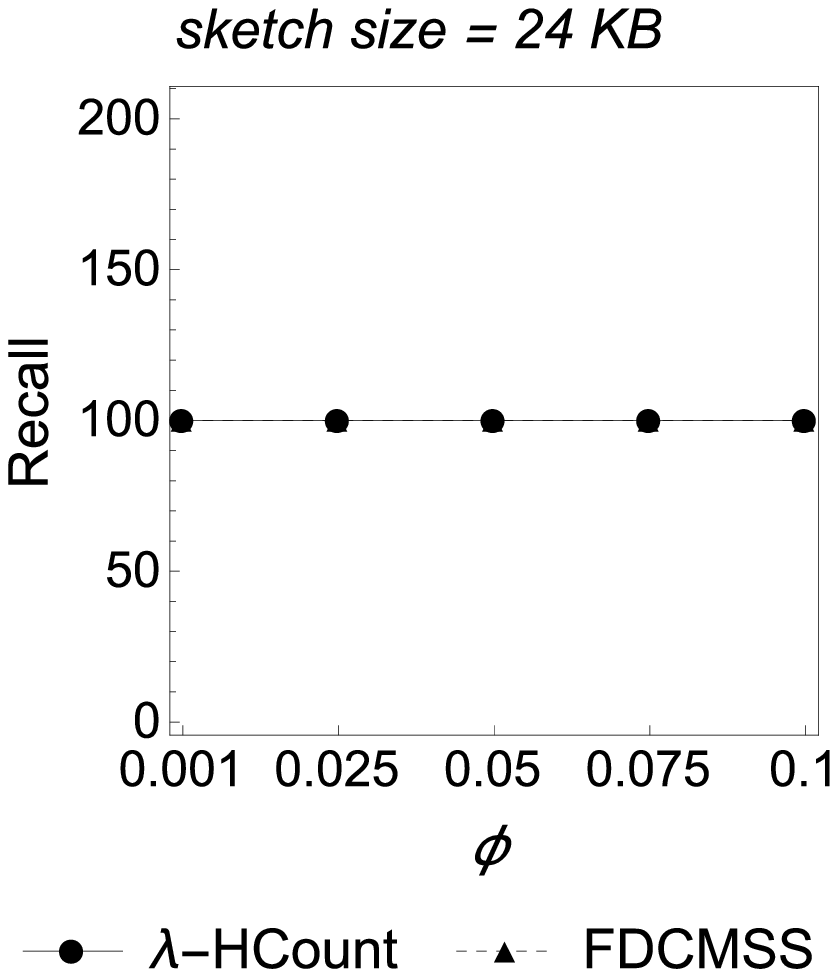}
           \label{phi-q148-recall}
        } &

      \subfloat[varying the sketch size]{
           \includegraphics[scale=0.36]{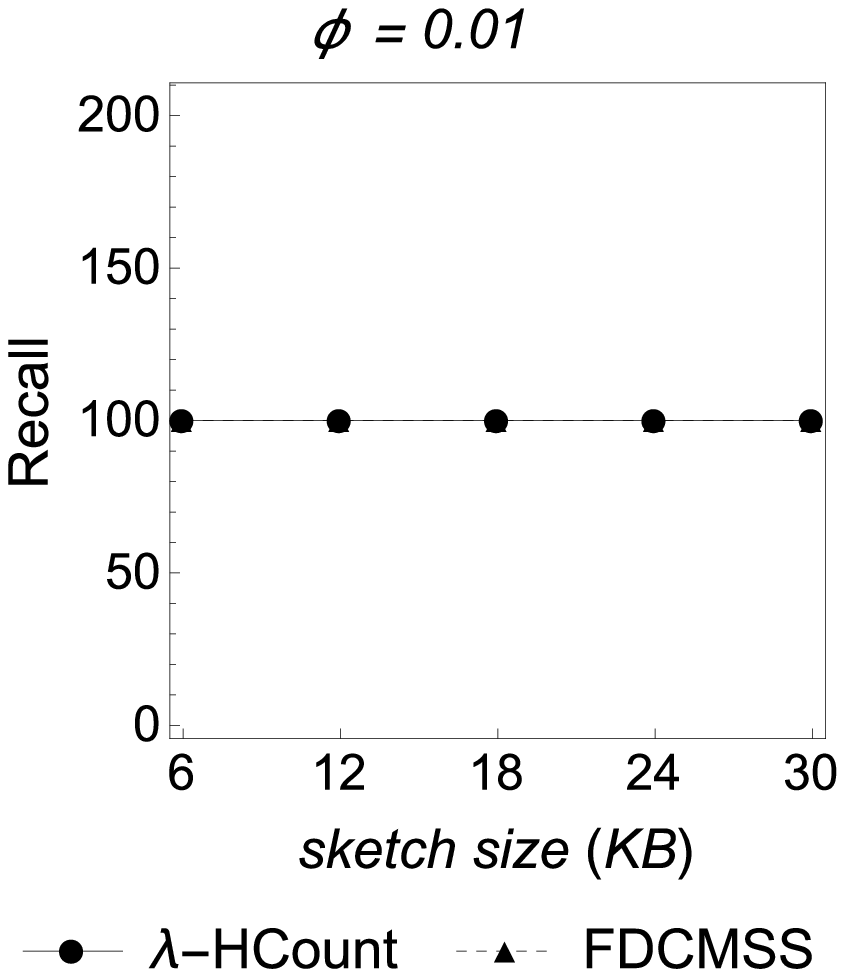}
           \label{sw-q148-recall}
        } &
        
     \subfloat[varying $\phi$]{
           \includegraphics[scale=0.36]{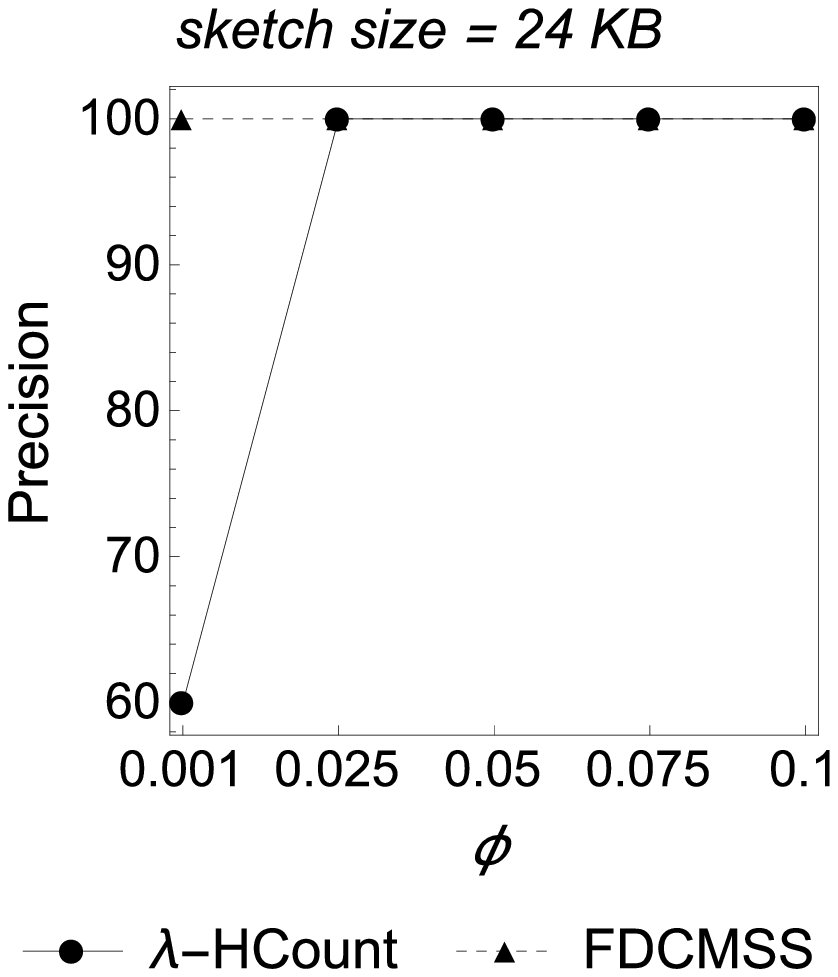}
           \label{phi-q148-prec}
        } &

      \subfloat[varying the sketch size]{
           \includegraphics[scale=0.36]{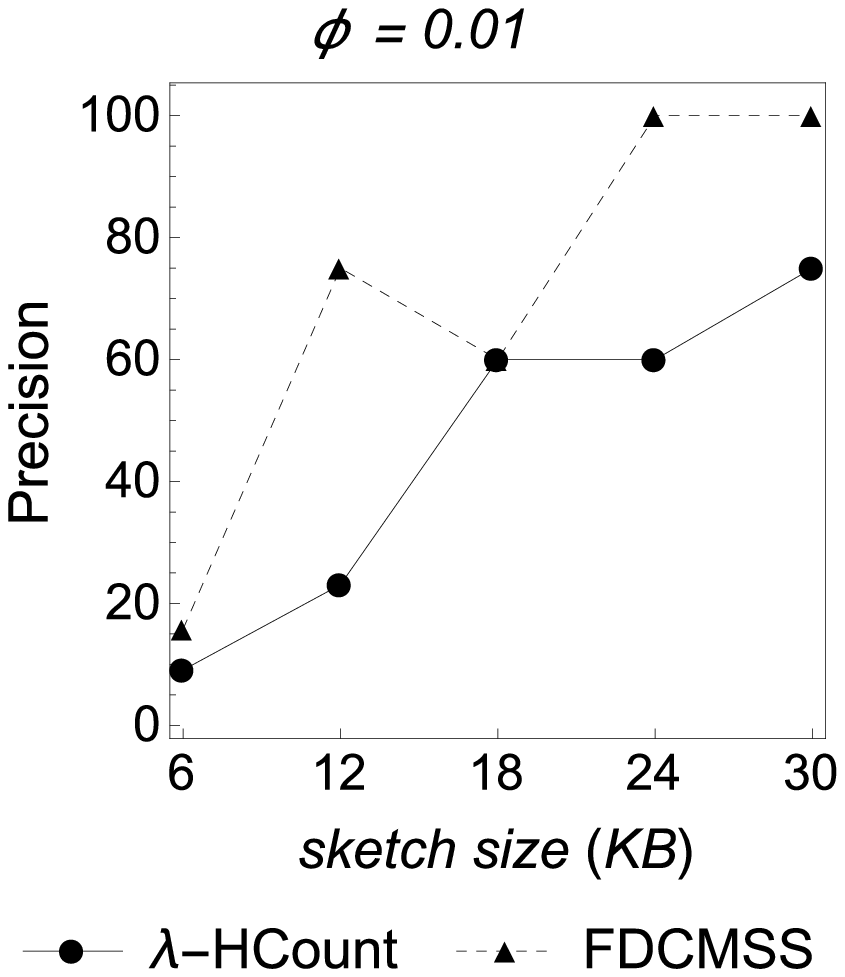}
           \label{sw-q148-prec}
        }

\end{tabular}
 
\caption{Q148: recall and precision (mean and confidence interval)} 
\label{q148-precision-recall}
\end{figure}

\begin{figure}[hbt]
\centering
\begin{tabular}{cccc}
             
      \subfloat[varying $\phi$]{
           \includegraphics[scale=0.36]{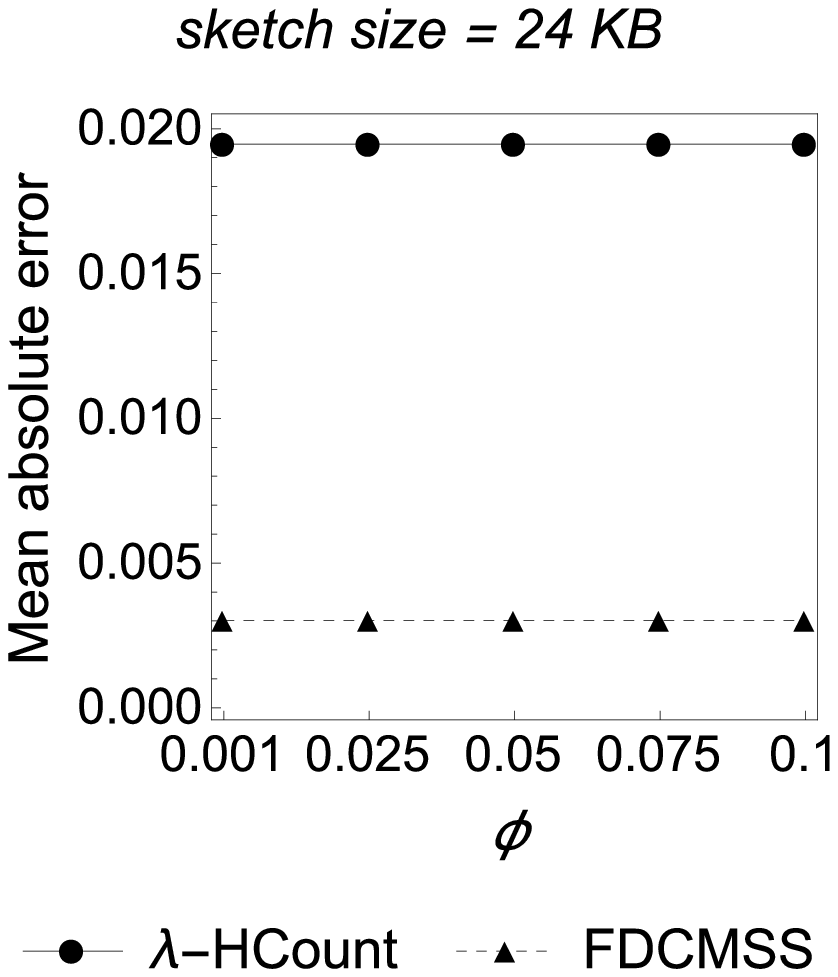}
           \label{phi-q148-mean-absolute-error}
        } &

      \subfloat[varying the sketch size]{
           \includegraphics[scale=0.36]{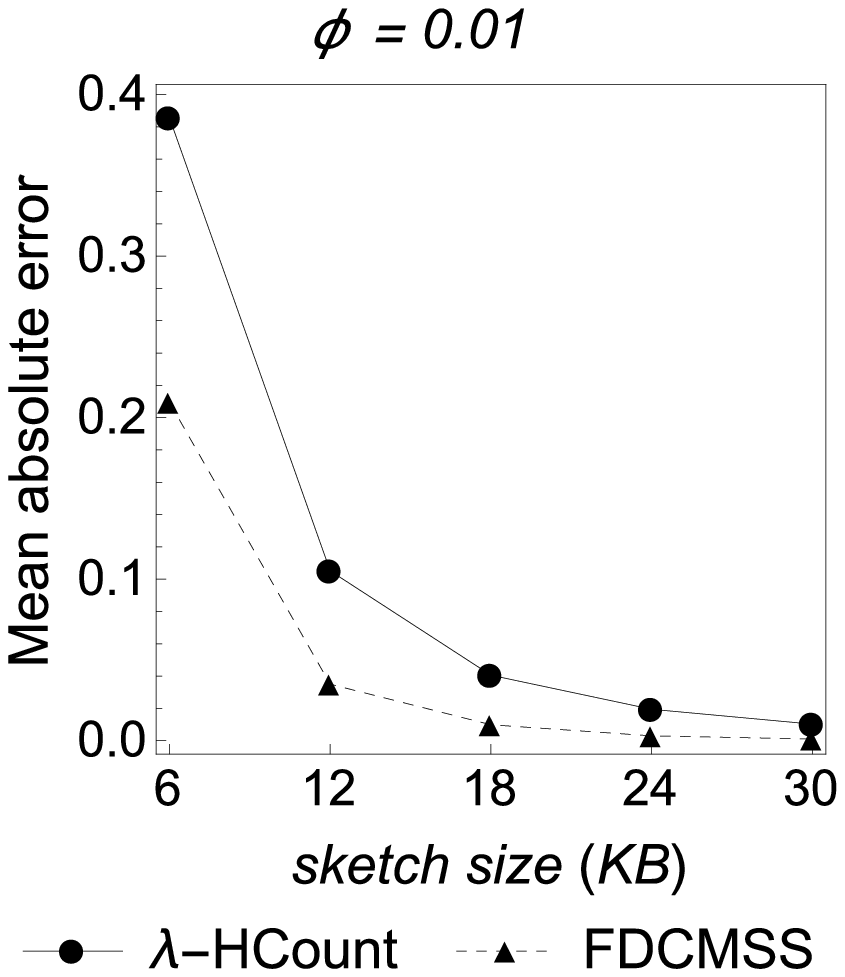}
           \label{sw-q148-mean-absolute-error}
        } &
        
     \subfloat[varying $\phi$]{
           \includegraphics[scale=0.36]{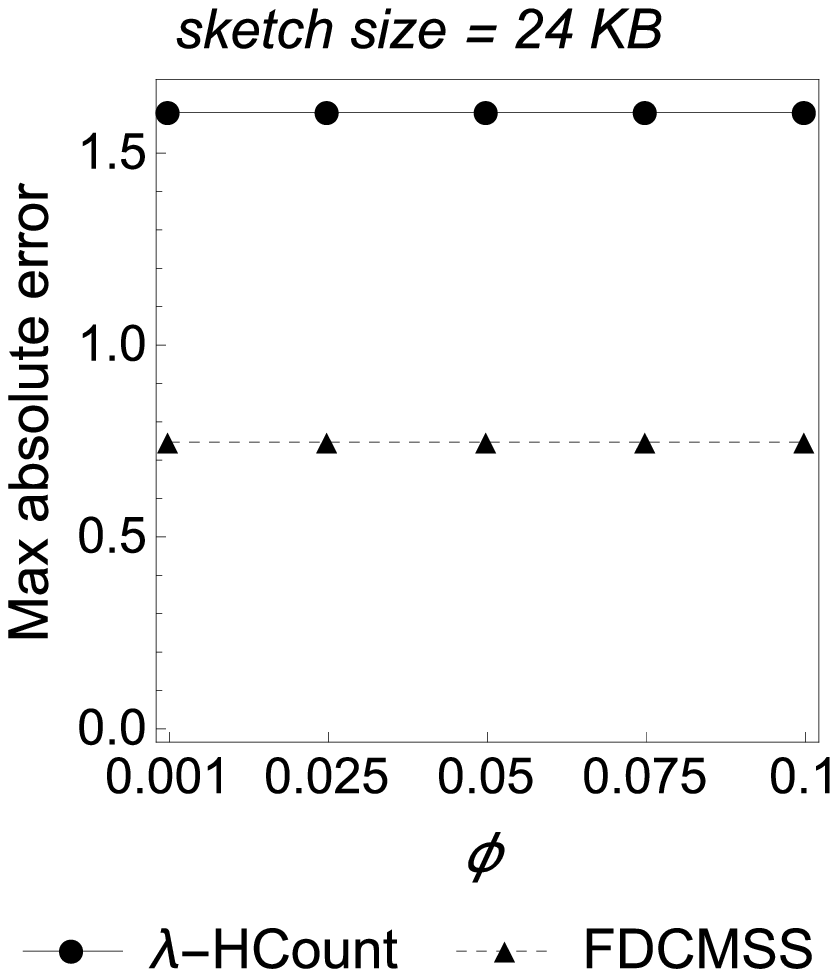}
           \label{phi-q148-max-absolute-error}
        } &

      \subfloat[varying the sketch size]{
           \includegraphics[scale=0.36]{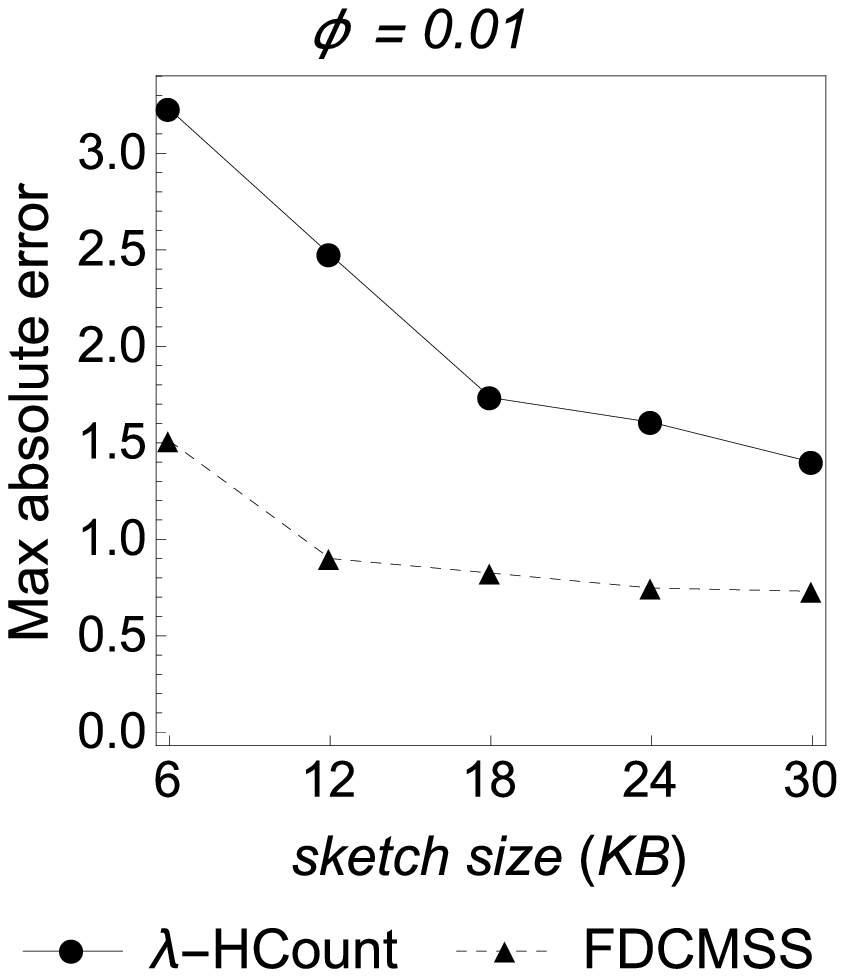}
           \label{sw-q148-max-absolute-error}
        }

\end{tabular}
 
\caption{Q148: Mean and max absolute error (mean and confidence interval)} 
\label{q148-mean-max-absolute-error}
\end{figure}

\begin{figure}[hbt]
\centering
\begin{tabular}{cccc}
             
      \subfloat[varying $\phi$]{
           \includegraphics[scale=0.36]{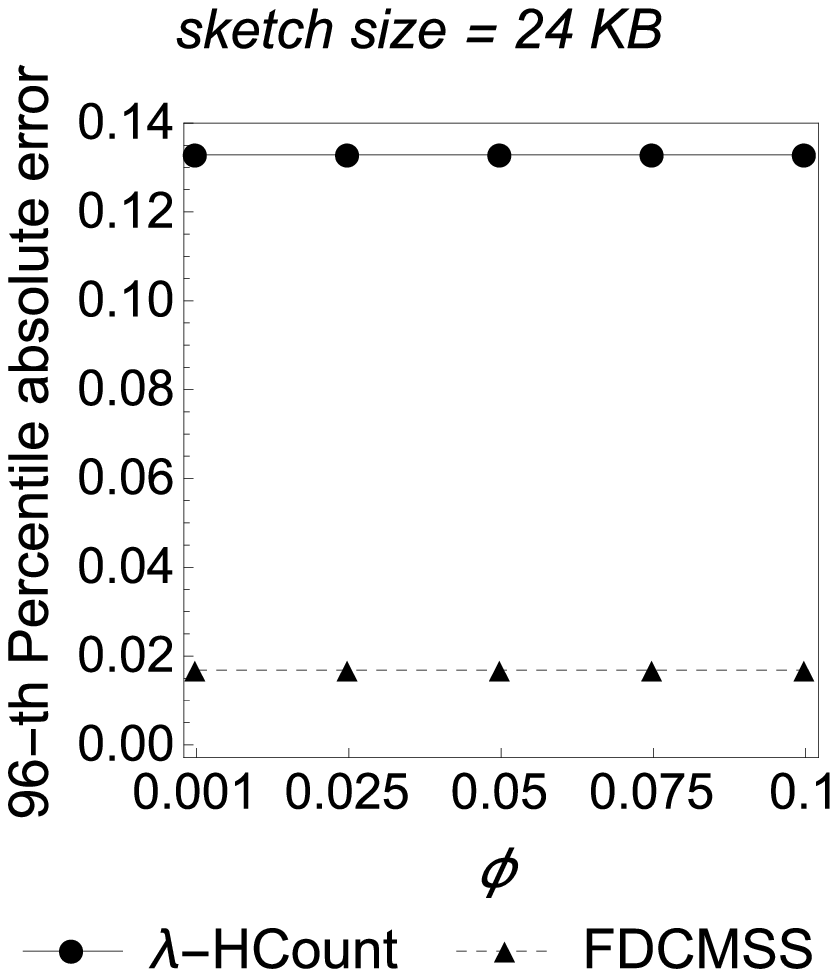}
           \label{phi-q148-percentile-absolute-error}
        } &

      \subfloat[varying the sketch size]{
           \includegraphics[scale=0.36]{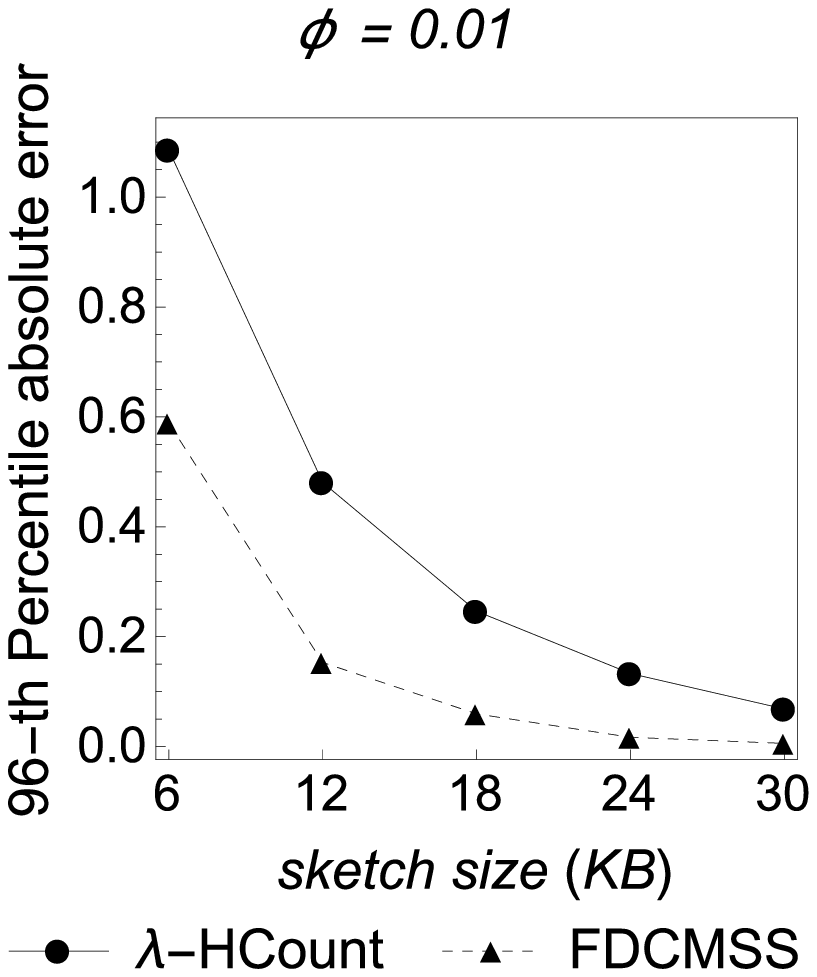}
           \label{sw-q148-percentile-absolute-error}
        } &
        
     \subfloat[varying $\phi$]{
           \includegraphics[scale=0.36]{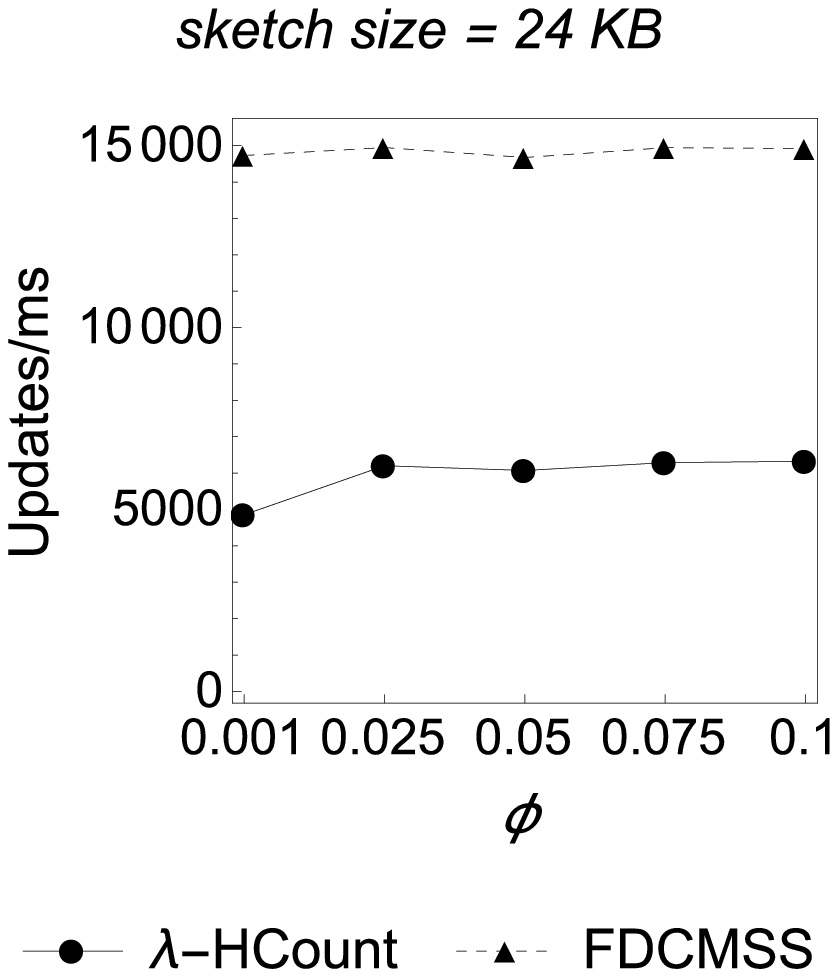}
           \label{phi-q148-updates}
        } &

      \subfloat[varying the sketch size]{
           \includegraphics[scale=0.36]{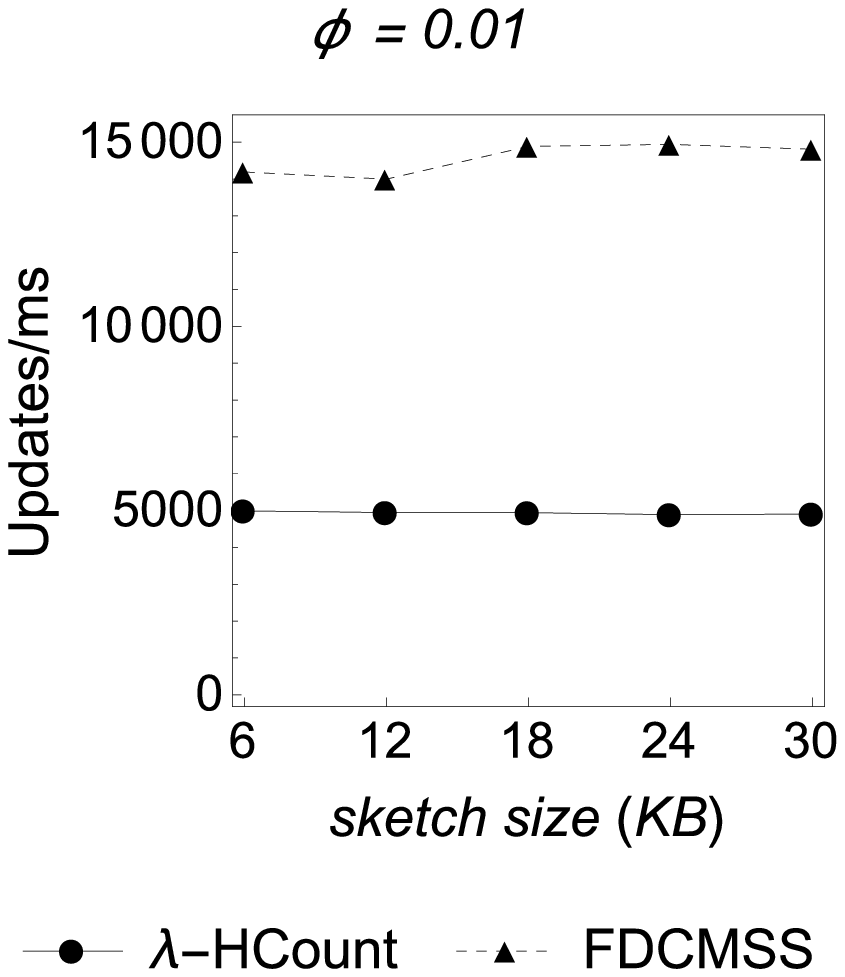}
           \label{sw-q148-updates}
        }

\end{tabular}
 
\caption{Q148: Percentile absolute error and updates/ms (mean and confidence interval)} 
\label{q148-percentile-absolute-error-updates}
\end{figure}

\begin{figure}[hbt]
\centering
\begin{tabular}{cccc}
             
      \subfloat[varying $\phi$]{
           \includegraphics[scale=0.36]{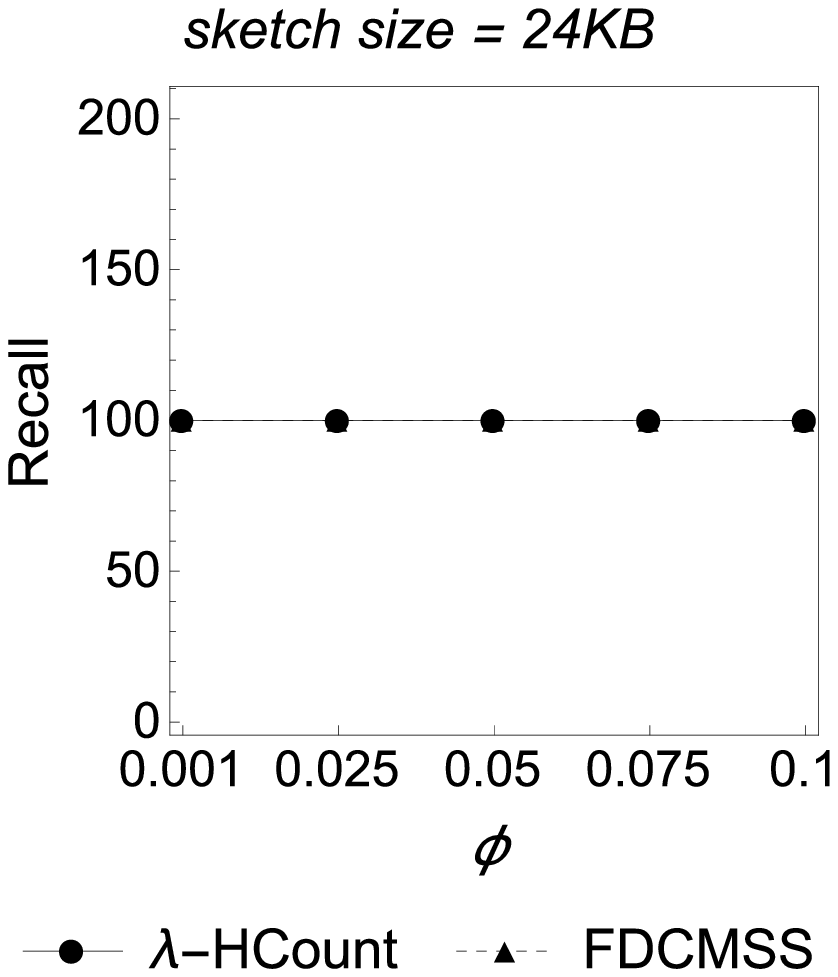}
           \label{phi-webdocs-recall}
        } &

      \subfloat[varying the sketch size]{
           \includegraphics[scale=0.36]{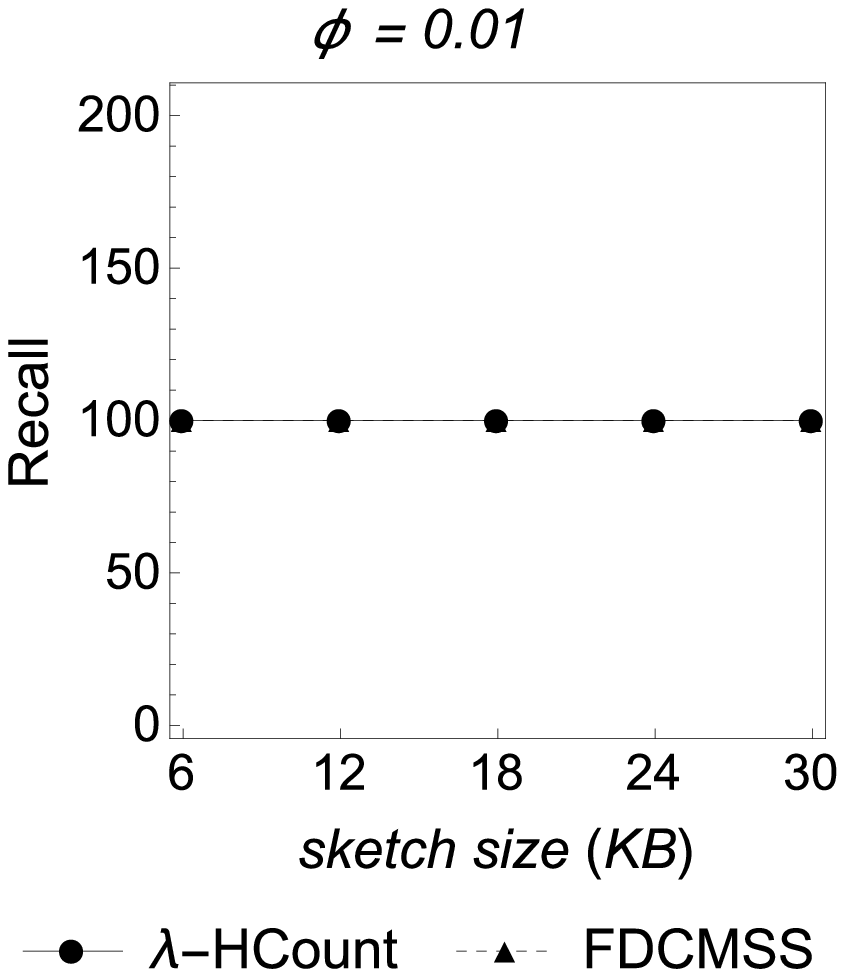}
           \label{sw-webdocs-recall}
        } &
        
     \subfloat[varying $\phi$]{
           \includegraphics[scale=0.36]{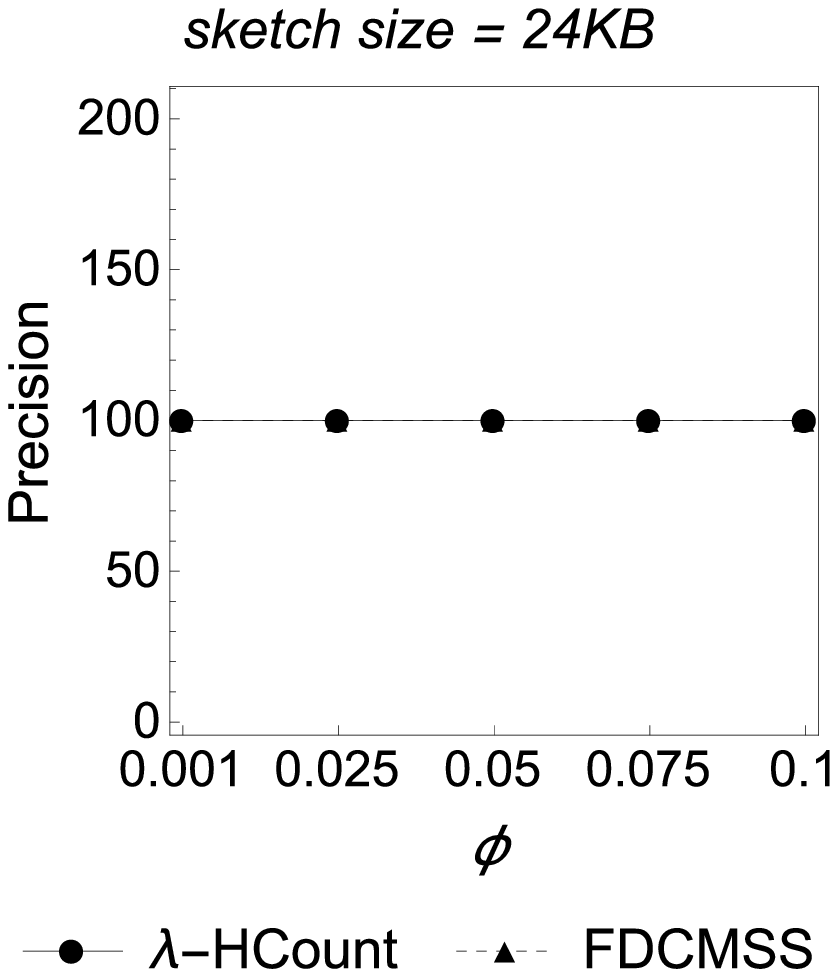}
           \label{phi-webdocs-prec}
        } &

      \subfloat[varying the sketch size]{
           \includegraphics[scale=0.36]{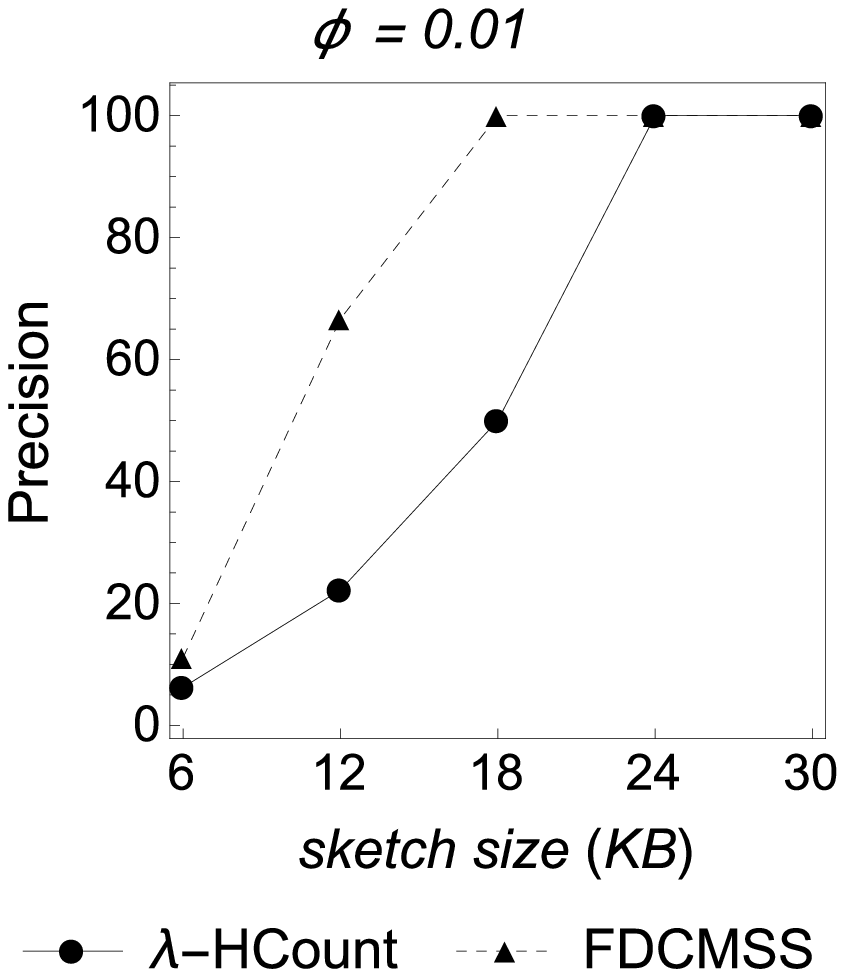}
           \label{sw-webdocs-prec}
        }

\end{tabular}
 
\caption{Webdocs: recall and precision (mean and confidence interval)} 
\label{webdocs-precision-recall}
\end{figure}

\begin{figure}[hbt]
\centering
\begin{tabular}{cccc}
             
      \subfloat[varying $\phi$]{
           \includegraphics[scale=0.36]{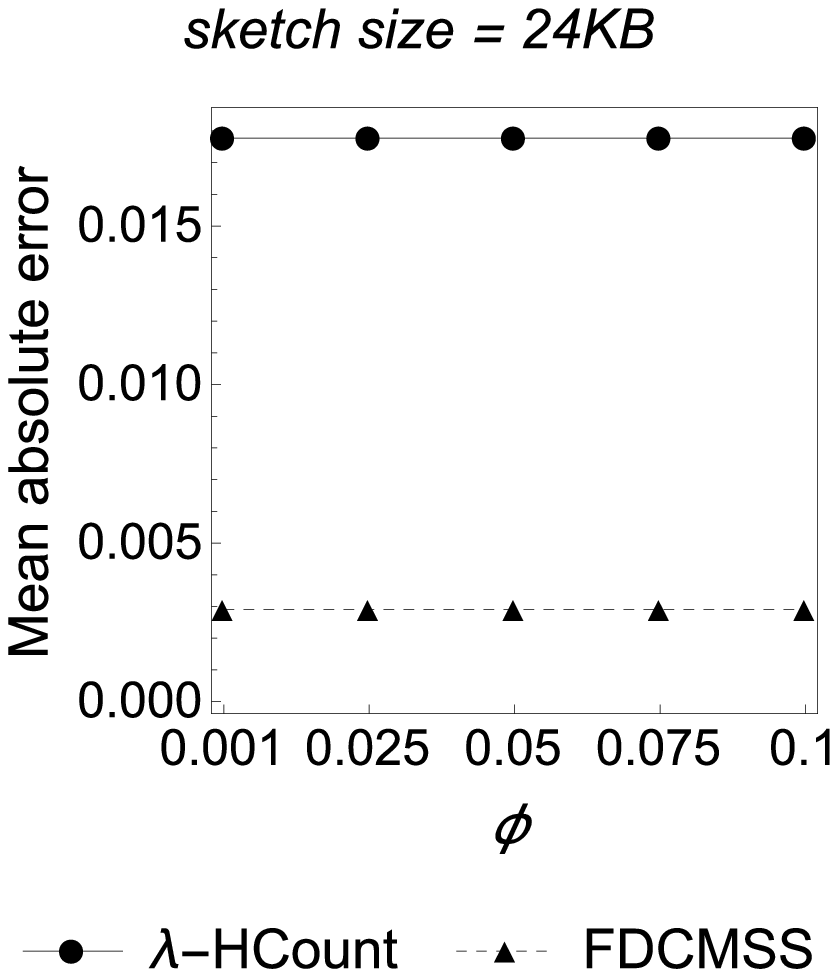}
           \label{phi-webdocs-mean-absolute-error}
        } &

      \subfloat[varying the sketch size]{
           \includegraphics[scale=0.36]{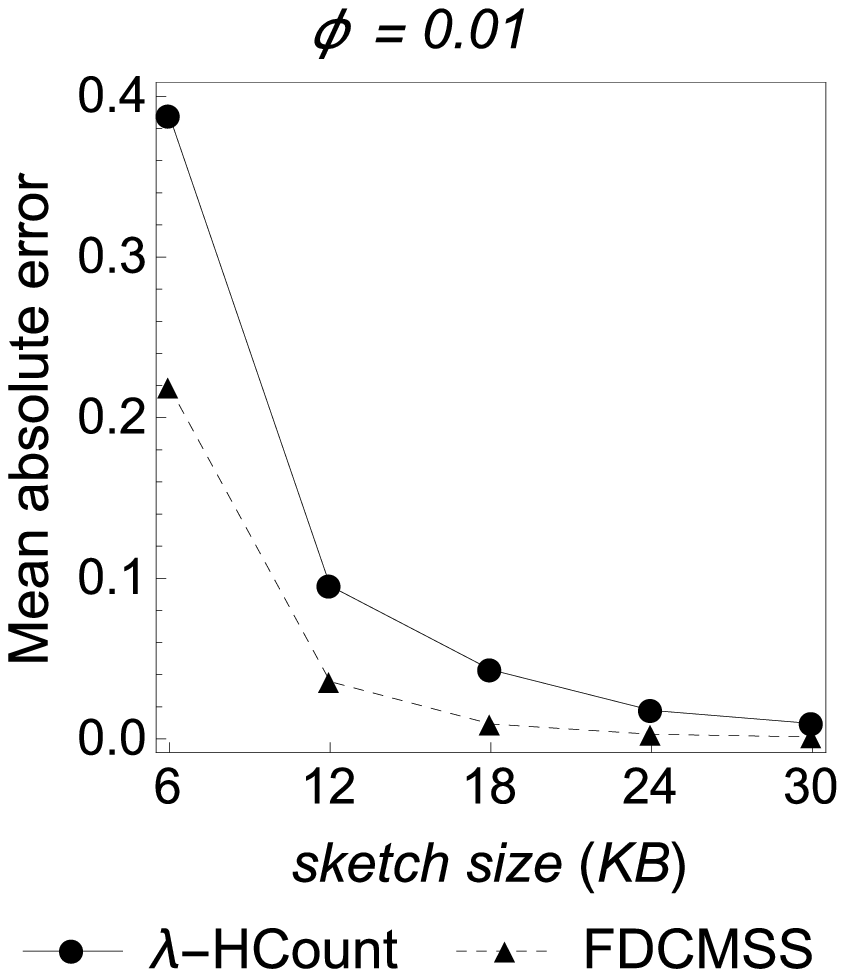}
           \label{sw-webdocs-mean-absolute-error}
        } &
        
     \subfloat[varying $\phi$]{
           \includegraphics[scale=0.36]{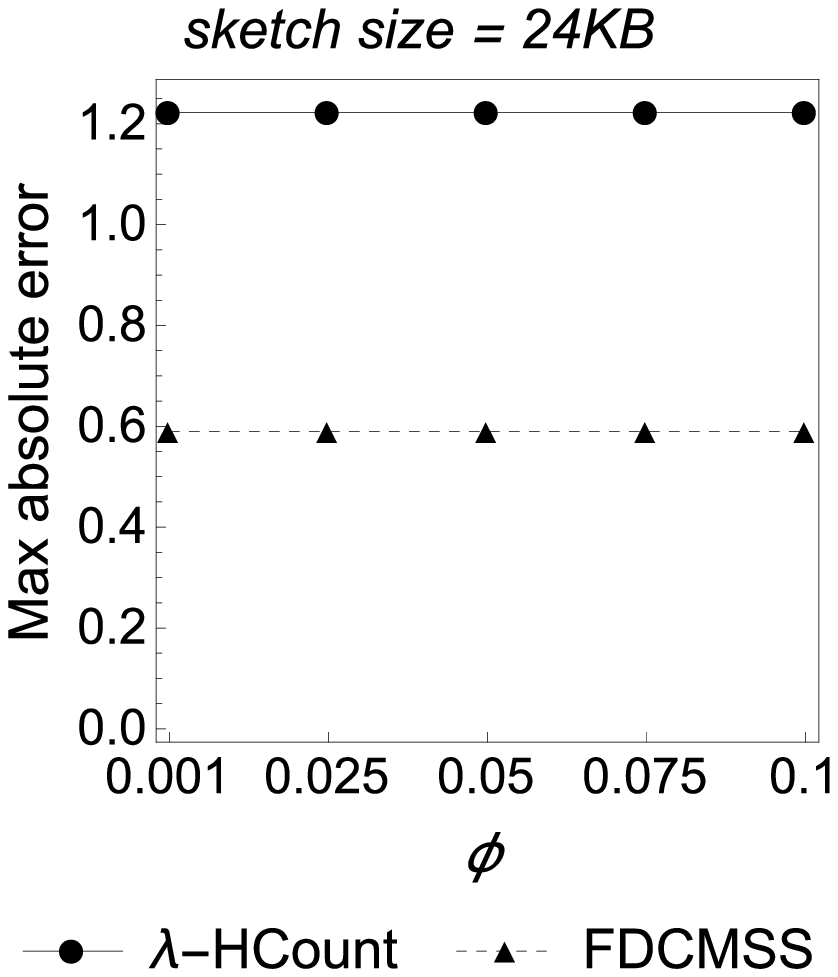}
           \label{phi-webdocs-max-absolute-error}
        } &

      \subfloat[varying the sketch size]{
           \includegraphics[scale=0.36]{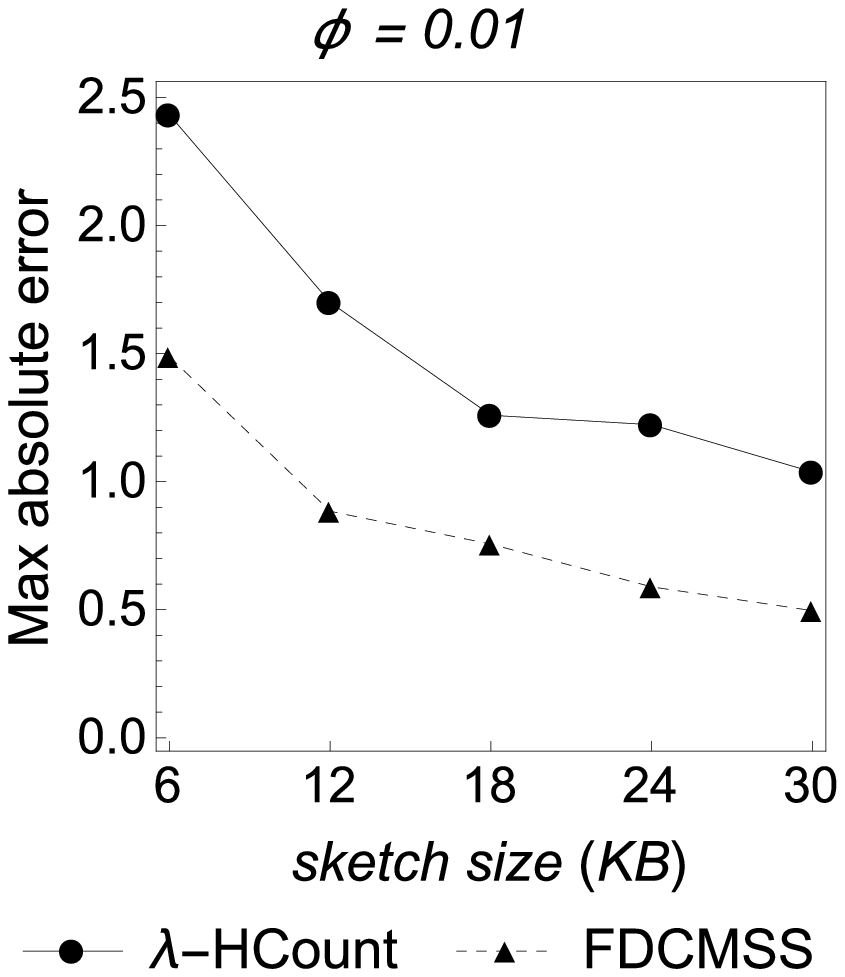}
           \label{sw-webdocs-max-absolute-error}
        }

\end{tabular}
 
\caption{Webdocs: Mean and max absolute error (mean and confidence interval)} 
\label{webdocs-mean-max-absolute-error}
\end{figure}

\begin{figure}[hbt]
\centering
\begin{tabular}{cccc}
             
      \subfloat[varying $\phi$]{
           \includegraphics[scale=0.36]{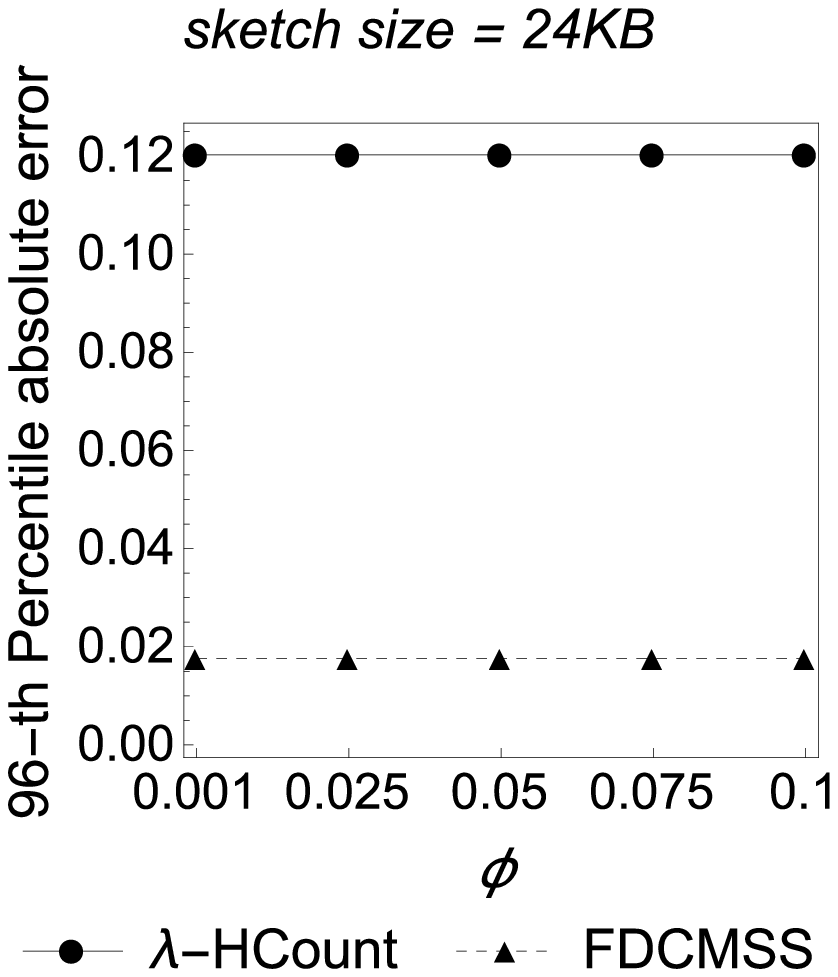}
           \label{phi-webdocs-percentile-absolute-error}
        } &

      \subfloat[varying the sketch size]{
           \includegraphics[scale=0.36]{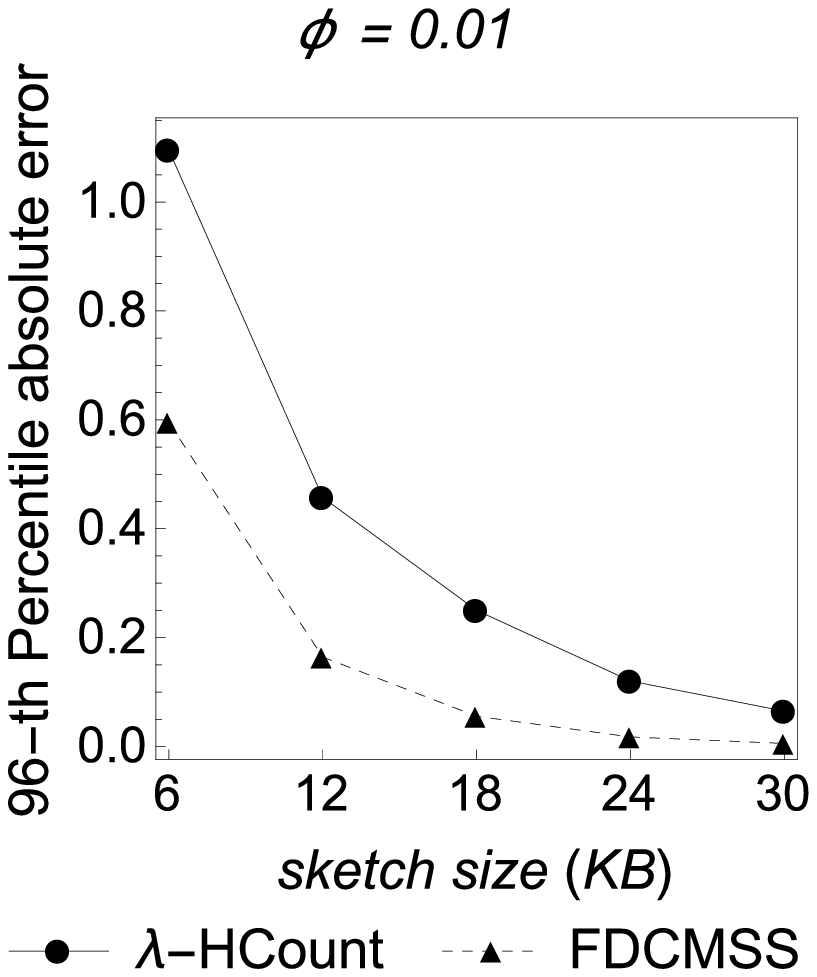}
           \label{sw-webdocs-percentile-absolute-error}
        } &
        
     \subfloat[varying $\phi$]{
           \includegraphics[scale=0.36]{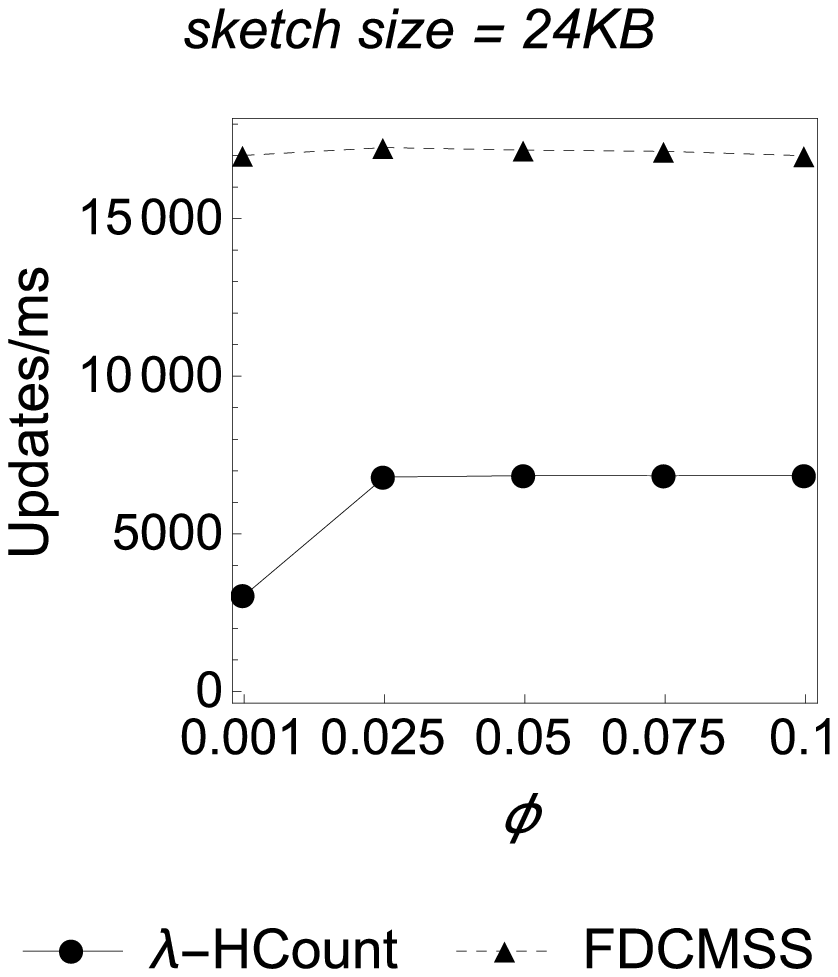}
           \label{phi-webdocs-updates}
        } &

      \subfloat[varying the sketch size]{
           \includegraphics[scale=0.36]{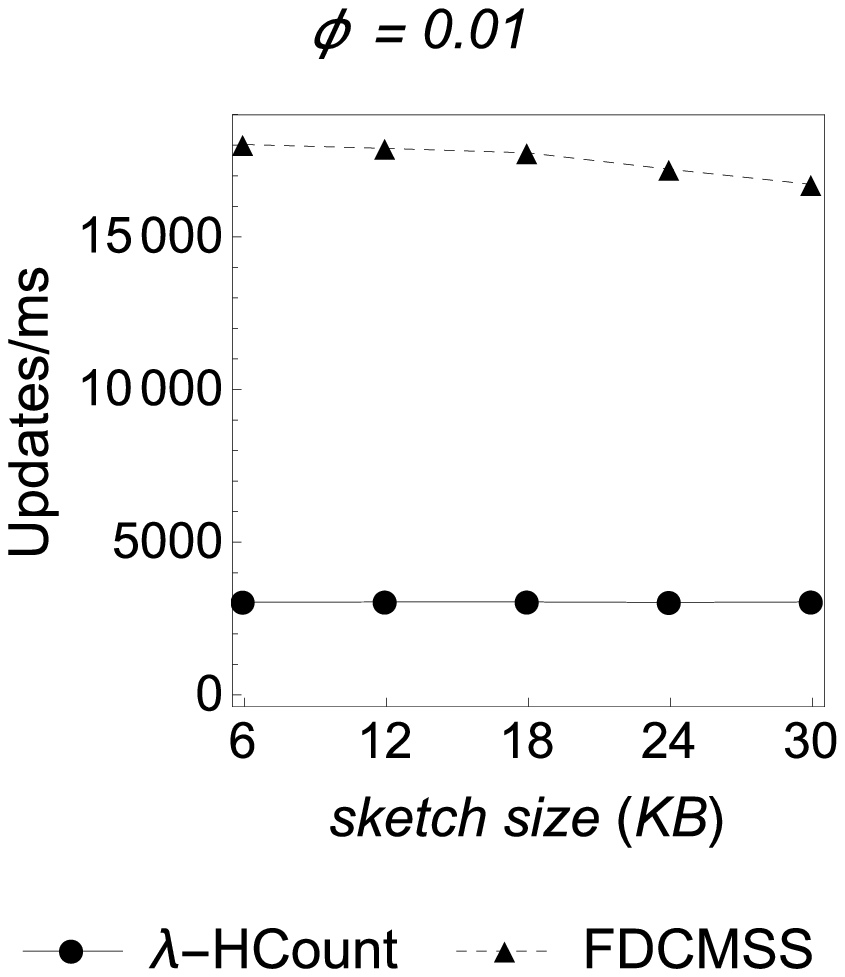}
           \label{sw-webdocs-updates}
        }

\end{tabular}
 
\caption{Webdocs: Percentile absolute error and updates/ms (mean and confidence interval)} 
\label{webdocs-percentile-absolute-error-updates}
\end{figure}

\section{Conclusions}
\label{conclusions}
We have presented the design and implementation of FDCMSS, a new algorithm for mining frequent items in the time fading model. Our algorithm is sketch based, and cleverly combines key ideas borrowed from forward decay, the Count-Min and the Space Saving algorithms. We have formally proved the correctness of our algorithm and shown, through extensive experimental results, that FDCMSS outperforms $\lambda$-HCount, a recently developed algorithm, with regard to speed, space used, precision attained and error committed on both synthetic and real datasets. Future work include parallelizing the algorithm on both shared-memory and message-passing architectures.  

\section*{Acknowledgment}
We are indebted to the unknown referees for enlightening observations, which helped us to improve the paper.

%%%%%%%%%%%%%%%%%%%%%%%%%%%%%%%%%%%%%%%%%%%%%%%%%%%%%%%%%%%
% the following \clearpage command will prevent floats to appear in or after the references
\clearpage

\bibliographystyle{elsarticle-num}
\bibliography{bibliography}

\begin{bibdiv}
\begin{biblist}

\bib{FIMDR}{misc}{
       title={Frequent itemset mining dataset repository},
        date={2016 (accessed August 1, 2016)},
        note={\url{http://fimi.ua.ac.be/data/}},
}

\bib{Beyer99bottom-upcomputation}{inproceedings}{
      author={Beyer, Kevin},
      author={Ramakrishnan, Raghu},
       title={Bottom--up computation of sparse and iceberg cubes},
        date={1999},
   booktitle={Proceedings of the acm sigmod international conference on
  management of data. acm, new york},
       pages={359\ndash 370},
}

\bib{Brin}{inproceedings}{
      author={Brin, Sergey},
      author={Motwani, Rajeev},
      author={Ullman, Jeffrey~D.},
      author={Tsur, Shalom},
       title={Dynamic itemset counting and implication rules for market basket
  data},
        date={1997},
   booktitle={Sigmod '97: Proceedings of the 1997 acm sigmod international
  conference on management of data},
   publisher={ACM},
       pages={255\ndash 264},
}

\bib{Cafaro-Pulimeno-Tempesta}{article}{
      author={Cafaro, Massimo},
      author={Pulimeno, Marco},
      author={Tempesta, Piergiulio},
       title={A parallel space saving algorithm for frequent items and the
  hurwitz zeta distribution},
        date={2016},
        ISSN={0020-0255},
     journal={Information Sciences},
      volume={329},
       pages={1 \ndash  19},
  url={http://www.sciencedirect.com/science/article/pii/S002002551500657X},
}

\bib{cafaro-tempesta}{article}{
      author={Cafaro, Massimo},
      author={Tempesta, Piergiulio},
       title={Finding frequent items in parallel},
        date={2011-10},
        ISSN={1532-0626},
     journal={Concurr. Comput. : Pract. Exper.},
      volume={23},
      number={15},
       pages={1774\ndash 1788},
         url={http://dx.doi.org/10.1002/cpe.1761},
}

\bib{Charikar}{inproceedings}{
      author={Charikar, Moses},
      author={Chen, Kevin},
      author={Farach-Colton, Martin},
       title={Finding frequent items in data streams},
        date={2002},
   booktitle={Icalp '02: Proceedings of the 29th international colloquium on
  automata, languages and programming},
   publisher={Springer--Verlag},
       pages={693\ndash 703},
}

\bib{Chen-Mei}{article}{
      author={Chen, Ling},
      author={Mei, Qingling},
       title={Mining frequent items in data stream using time fading model},
        date={2014},
        ISSN={0020-0255},
     journal={Information Sciences},
      volume={257},
       pages={54 \ndash  69},
  url={http://www.sciencedirect.com/science/article/pii/S0020025513006403},
}

\bib{exp-decay}{inproceedings}{
      author={Cormode, G.},
      author={Korn, F.},
      author={Tirthapura, S.},
       title={Exponentially decayed aggregates on data streams},
        date={2008April},
   booktitle={Data engineering, 2008. icde 2008. ieee 24th international
  conference on},
       pages={1379\ndash 1381},
}

\bib{forward-decay}{inproceedings}{
      author={Cormode, G.},
      author={Shkapenyuk, V.},
      author={Srivastava, D.},
      author={Xu, Bojian},
       title={Forward decay: A practical time decay model for streaming
  systems},
        date={2009March},
   booktitle={Data engineering, 2009. icde '09. ieee 25th international
  conference on},
       pages={138\ndash 149},
}

\bib{Cormode}{article}{
      author={Cormode, Graham},
      author={Hadjieleftheriou, Marios},
       title={Finding the frequent items in streams of data},
        date={2009},
        ISSN={0001-0782},
     journal={Commun. ACM},
      volume={52},
      number={10},
       pages={97\ndash 105},
}

\bib{Cormode05}{article}{
      author={Cormode, Graham},
      author={Muthukrishnan, S.},
       title={An improved data stream summary: the count-min sketch and its
  applications},
        date={2005},
        ISSN={0196-6774},
     journal={J. Algorithms},
      volume={55},
      number={1},
       pages={58\ndash 75},
}

\bib{Cormode-grouptest}{article}{
      author={Cormode, Graham},
      author={Muthukrishnan, S.},
       title={What's hot and what's not: Tracking most frequent items
  dynamically},
        date={2005-03},
        ISSN={0362-5915},
     journal={ACM Trans. Database Syst.},
      volume={30},
      number={1},
       pages={249\ndash 278},
         url={http://doi.acm.org/10.1145/1061318.1061325},
}

\bib{Dallachiesa}{article}{
      author={Dallachiesa, Michele},
      author={Palpanas, Themis},
       title={Identifying streaming frequent items in ad hoc time windows},
        date={2013},
     journal={Data Knowl. Eng.},
      volume={87},
       pages={66\ndash 90},
         url={http://dx.doi.org/10.1016/j.datak.2013.05.007},
}

\bib{Das2009}{article}{
      author={Das, Sudipto},
      author={Antony, Shyam},
      author={Agrawal, Divyakant},
      author={El~Abbadi, Amr},
       title={Thread cooperation in multicore architectures for frequency
  counting over multiple data streams},
        date={2009-08},
        ISSN={2150-8097},
     journal={Proc. VLDB Endow.},
      volume={2},
      number={1},
       pages={217\ndash 228},
         url={http://dx.doi.org/10.14778/1687627.1687653},
}

\bib{Datar}{inproceedings}{
      author={Datar, Mayur},
      author={Gionis, Aristides},
      author={Indyk, Piotr},
      author={Motwani, Rajeev},
       title={Maintaining stream statistics over sliding windows: (extended
  abstract)},
        date={2002},
   booktitle={Proceedings of the thirteenth annual acm-siam symposium on
  discrete algorithms},
      series={SODA '02},
   publisher={Society for Industrial and Applied Mathematics},
     address={Philadelphia, PA, USA},
       pages={635\ndash 644},
}

\bib{DemaineLM02}{inproceedings}{
      author={Demaine, Erik~D.},
      author={L{\'o}pez-Ortiz, Alejandro},
      author={Munro, J.~Ian},
       title={Frequency estimation of internet packet streams with limited
  space},
        date={2002},
   booktitle={Esa},
       pages={348\ndash 360},
}

\bib{Erra2012}{article}{
      author={Erra, Ugo},
      author={Frola, Bernardino},
       title={Frequent items mining acceleration exploiting fast parallel
  sorting on the \{GPU\}},
        date={2012},
        ISSN={1877-0509},
     journal={Procedia Computer Science},
      volume={9},
      number={0},
       pages={86 \ndash  95},
  url={http://www.sciencedirect.com/science/article/pii/S1877050912001317},
        note={Proceedings of the International Conference on Computational
  Science, \{ICCS\} 2012},
}

\bib{Estan}{inproceedings}{
      author={Estan, Cristian},
      author={Varghese, George},
       title={New directions in traffic measurement and accounting},
        date={2001},
   booktitle={Imw '01: Proceedings of the 1st acm sigcomm workshop on internet
  measurement},
   publisher={ACM},
       pages={75\ndash 80},
}

\bib{Fang98computingiceberg}{inproceedings}{
      author={Fang, Min},
      author={Shivakumar, Narayanan},
      author={Garcia-Molina, Hector},
      author={Motwani, Rajeev},
      author={Ullman, Jeffrey~D.},
       title={Computing iceberg queries efficiently},
        date={1998},
   booktitle={Proceedings of the 24th international conference on very large
  data bases, vldb. morgan--kaufmann, san mateo, calif.},
       pages={299\ndash 310},
}

\bib{CICLing}{proceedings}{
      editor={Gelbukhl, Alexander},
       title={Computational linguistics and intelligent text processing, 7th
  international conference, cicling 2006},
      series={Lecture Notes in Computer Science},
   publisher={Springer--Verlag},
        date={2006},
      volume={3878},
}

\bib{Gibbons}{inproceedings}{
      author={Gibbons, Phillip~B.},
      author={Matias, Yossi},
       title={Synopsis data structures for massive data sets},
        date={1999},
   booktitle={Dimacs: Series in discrete mathematics and theoretical computer
  science: Special issue on external memory algorithms and visualization, vol.
  a.},
   publisher={American Mathematical Society},
       pages={39\ndash 70},
}

\bib{Govindaraju2005}{inproceedings}{
      author={Govindaraju, Naga~K.},
      author={Raghuvanshi, Nikunj},
      author={Manocha, Dinesh},
       title={Fast and approximate stream mining of quantiles and frequencies
  using graphics processors},
        date={2005},
   booktitle={Proceedings of the 2005 acm sigmod international conference on
  management of data},
      series={SIGMOD '05},
   publisher={ACM},
       pages={611\ndash 622},
         url={http://doi.acm.org/10.1145/1066157.1066227},
}

\bib{Jin03}{inproceedings}{
      author={Jin, Cheqing},
      author={Qian, Weining},
      author={Sha, Chaofeng},
      author={Yu, Jeffrey~X.},
      author={Zhou, Aoying},
       title={Dynamically maintaining frequent items over a data stream},
        date={2003},
   booktitle={In proc. of cikm},
   publisher={ACM Press},
       pages={287\ndash 294},
}

\bib{Karp}{article}{
      author={Karp, Richard~M.},
      author={Shenker, Scott},
      author={Papadimitriou, Christos~H.},
       title={A simple algorithm for finding frequent elements in streams and
  bags},
        date={2003},
        ISSN={0362-5915},
     journal={ACM Trans. Database Syst.},
      volume={28},
      number={1},
       pages={51\ndash 55},
}

\bib{Manerikar}{article}{
      author={Manerikar, Nishad},
      author={Palpanas, Themis},
       title={Frequent items in streaming data: An experimental evaluation of
  the state-of-the-art},
        date={2009},
        ISSN={0169-023X},
     journal={Data Knowl. Eng.},
      volume={68},
      number={4},
       pages={415\ndash 430},
}

\bib{recent-freq-items}{inproceedings}{
      author={Manjhi, A.},
      author={Shkapenyuk, V.},
      author={Dhamdhere, K.},
      author={Olston, C.},
       title={Finding (recently) frequent items in distributed data streams},
        date={2005April},
   booktitle={Data engineering, 2005. icde 2005. proceedings. 21st
  international conference on},
       pages={767\ndash 778},
}

\bib{Manku02approximatefrequency}{inproceedings}{
      author={Manku, Gurmeet~Singh},
      author={Motwani, Rajeev},
       title={Approximate frequency counts over data streams},
        date={2002},
   booktitle={In vldb},
       pages={346\ndash 357},
}

\bib{Metwally2006}{article}{
      author={Metwally, Ahmed},
      author={Agrawal, Divyakant},
      author={Abbadi, Amr~El},
       title={An integrated efficient solution for computing frequent and top-k
  elements in data streams},
        date={2006-09},
        ISSN={0362-5915},
     journal={ACM Trans. Database Syst.},
      volume={31},
      number={3},
       pages={1095\ndash 1133},
         url={http://doi.acm.org/10.1145/1166074.1166084},
}

\bib{Misra82}{article}{
      author={Misra, Jayadev},
      author={Gries, David},
       title={Finding repeated elements},
        date={1982},
     journal={Sci. Comput. Program.},
      volume={2},
      number={2},
       pages={143\ndash 152},
}

\bib{TCS-002}{article}{
      author={Muthukrishnan, S.},
       title={Data streams: Algorithms and applications},
        date={2005},
        ISSN={1551-305X},
     journal={Foundations and Trends® in Theoretical Computer Science},
      volume={1},
      number={2},
       pages={117\ndash 236},
         url={http://dx.doi.org/10.1561/0400000002},
}

\bib{Pan}{article}{
      author={Pan, Rong},
      author={Breslau, Lee},
      author={Prabhakar, Balaji},
      author={Shenker, Scott},
       title={Approximate fairness through differential dropping},
        date={2003},
        ISSN={0146-4833},
     journal={SIGCOMM Comput. Commun. Rev.},
      volume={33},
      number={2},
       pages={23\ndash 39},
}

\bib{Roy2012}{inproceedings}{
      author={Roy, Pratanu},
      author={Teubner, Jens},
      author={Alonso, Gustavo},
       title={Efficient frequent item counting in multi-core hardware},
        date={2012},
   booktitle={Proceedings of the 18th acm sigkdd international conference on
  knowledge discovery and data mining},
      series={KDD '12},
   publisher={ACM},
       pages={1451\ndash 1459},
         url={http://doi.acm.org/10.1145/2339530.2339757},
}

\bib{Tangwongsan2014}{inproceedings}{
      author={Tangwongsan, Kanat},
      author={Tirthapura, Srikanta},
      author={Wu, Kun-Lung},
       title={Parallel streaming frequency-based aggregates},
        date={2014},
   booktitle={Proceedings of the 26th acm symposium on parallelism in
  algorithms and architectures},
      series={SPAA '14},
   publisher={ACM},
       pages={236\ndash 245},
         url={http://doi.acm.org/10.1145/2612669.2612695},
}

\bib{Zhang2012}{inproceedings}{
      author={Zhang, Yu},
       title={Parallelizing the weighted lossy counting algorithm in high-speed
  network monitoring},
        date={2012},
   booktitle={Instrumentation, measurement, computer, communication and control
  (imccc), second international conference on},
       pages={757\ndash 761},
}

\bib{Zhang2013}{article}{
      author={Zhang, Yu},
      author={Sun, Yue},
      author={Zhang, Jianzhong},
      author={Xu, Jingdong},
      author={Wu, Ying},
       title={An efficient framework for parallel and continuous frequent item
  monitoring},
        date={2014},
        ISSN={1532-0634},
     journal={Concurrency and Computation: Practice and Experience},
      volume={26},
      number={18},
       pages={2856\ndash 2879},
         url={http://dx.doi.org/10.1002/cpe.3182},
}

\end{biblist}
\end{bibdiv}

\end{document}